\documentclass[journal]{IEEEtran}

\pdfoutput=1

\usepackage{graphicx}
\usepackage{mathrsfs}
\usepackage{mathptmx}
\usepackage{times}
\usepackage{amsmath}
\usepackage{amsthm}
\usepackage{subcaption}
\usepackage{comment}
\usepackage[export]{adjustbox}
\usepackage{amssymb} 
\usepackage{epstopdf}
\usepackage{tikz}
\usepackage{pgfplots} 
\usetikzlibrary{external,positioning,decorations.pathreplacing,shapes,arrows,patterns}

\newtheorem{theorem}{\bf Theorem}
\newtheorem{prop}[theorem]{\bf Proposition} 
 
\newtheorem{definition}{\bf Definition} 
\newtheorem{example}{\bf Example} [section]
\newtheorem{lemma}{\bf Lemma} 

\newcommand{\mmse}{\mathsf{mmse}}
\newcommand{\expphi}{\left(e^{2\pi i \phi} \right)}

\tikzstyle{int}=[draw, fill=blue!10, minimum height = 1cm, minimum width=1.5cm,thick ]
\tikzstyle{sint}=[draw, fill=blue!10, minimum height = 0.5cm, minimum width=0.8cm,thick ]
\tikzstyle{sum}=[circle, fill=blue!10, draw=black,line width=1pt,minimum size = 0.5cm, thick ]
\tikzstyle{ssum}=[circle, fill=blue!10,draw=black,line width=1pt,minimum size = 0.1cm]
\tikzstyle{int1}=[draw, fill=blue!10, minimum height = 0.5cm, minimum width=1cm,thick ]
\tikzstyle{enc}=[draw, fill=blue!10, minimum height = 2.7cm, minimum width=1cm,thick ]
\tikzstyle{int}=[draw, fill=blue!10, minimum height = 1cm, minimum width=1.5cm,thick ]

\title{ \LARGE \bf
Distortion Rate Function of Sub-Nyquist \\Sampled Gaussian Sources}
\author{ 
\IEEEauthorblockN{
Alon Kipnis, Andrea J. Goldsmith, Yonina C. Eldar and Tsachy Weissman  }

\thanks{  A. Kipnis, A. J. Goldsmith and T. Weissman is with the Department of Electrical Engineering, Stanford University, Stanford, CA 94305 USA. 

Y. C. Eldar is with the Department of Electrical Engineering, Technion - Israel Institute of Technology Haifa 32000, Israel.}
\thanks{
This work was supported in part by the NSF Center for Science of Information (CSoI) under grant CCF-0939370, the BSF Transformative Science Grant 2010505 and the Intel Collaborative Research Institute for Computational Intelligence (ICRI-CI). \par
This paper was presented in part at the 51st Annual Allerton Conference on Communication, Control, and Computing (Allerton), September 2013.}
}

\begin{document}
\graphicspath{{../Figures/}}
\maketitle
\thispagestyle{empty}
\pagestyle{empty}

\begin{abstract}
The amount of information lost in sub-Nyquist sampling of a continuous-time Gaussian stationary process is quantified. We consider a combined source coding and sub-Nyquist reconstruction problem in which the input to the encoder is a noisy sub-Nyquist sampled version of the analog source. We first derive an expression for the mean squared error in the reconstruction of the process from a noisy and information rate-limited version of its samples. This expression is a function of the sampling frequency and the average number of bits describing each sample. It is given as the sum of two terms: Minimum mean square error in estimating the source from its noisy but otherwise fully observed sub-Nyquist samples, and a second term obtained by reverse waterfilling over an average of spectral densities associated with the polyphase components of the source. We extend this result to multi-branch uniform sampling, where the samples are available through a set of parallel channels with a uniform sampler and a pre-sampling filter in each branch. Further optimization to reduce distortion is then performed over the pre-sampling filters, and an optimal set of pre-sampling filters associated with the statistics of the input signal and the sampling frequency is found. This results in an expression for the minimal possible distortion achievable under any analog to digital conversion scheme involving uniform sampling and linear filtering. These results thus unify the Shannon-Whittaker-Kotelnikov sampling theorem and Shannon rate-distortion theory for Gaussian sources.
\end{abstract}


\section{INTRODUCTION}
\label{sec:Intro}
Consider the task of storing an analog source in digital memory. The trade-off between the bit-rate of the samples and the minimal possible distortion in the reconstruction of the signal from these samples is described by the distortion-rate function (DRF) of the source. A key idea in determining the DRF of an analog source is to map the continuous-time process into a discrete-time process based on sampling above the Nyquist frequency \cite[Sec. 4.5.3]{berger1971rate}.  Since wideband signaling and A/D technology limitations can preclude sampling signals at their Nyquist frequency \cite{mishali2011sub, eldar2015sampling}, an optimal source code based on such a discrete-time representation may be impractical in certain scenarios. In addition, some applications may be less sensitive to inaccuracies in the data, which suggests that the sampling frequency can be reduced far below the Nyquist frequency without significantly affecting performance. These considerations motivate us to consider the source coding problem in Fig.~\ref{fig:operational_scheme}, in which an analog random signal $X(\cdot)$ with additive noise needs to be reconstructed from its rate-limited samples. This introduces a combined sampling and source coding problem, which lies at the intersection of information theory and signal processing.\par
The parameters in this problem formulation are the sampling frequency $f_{s}$, the source coding rate $R$ and the average distortion $D$. If the sampling frequency is such that  the sampled process can be reconstructed from its samples, then the sampling operation has no effect on distortion and the trade-off between the source coding rate and the distortion is given by the indirect DRF (iDRF) of the source \cite{1057738}. The other extreme is when the source coding rate $R$ goes to infinity, in which case we are left with a signal processing problem: reconstructing an undersampled signal in the presence of noise \cite{815501}. The reductions of the general problem in these two special cases are illustrated by the diagram in Fig.~\ref{fig:diagram_scheme}. 

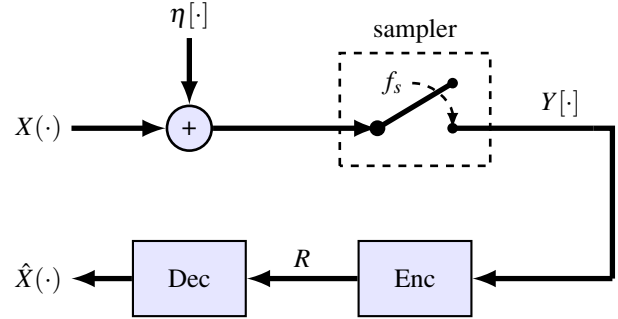
\begin{figure}
\begin{center}
\begin{tikzpicture}[node distance=2cm,auto,>=latex]
  \node at (0,0) (source) {$X(\cdot)$};
  \node [ssum, right of = source, node distance = 2cm] (sum) {+};
  \node [above of = sum,node distance = 1.5cm] (noise) {$\eta[\cdot]$};
  \node [coordinate, right of = sum,node distance = 2.5cm] (smp_in) {};
  \node [coordinate, right of = smp_in,node distance = 1cm] (smp_out){};
  \node [coordinate,above of = smp_out,node distance = 0.6cm] (tip) {};
\fill  (smp_out) circle [radius=2pt];
\fill  (smp_in) circle [radius=3pt];
\fill  (tip) circle [radius=2pt];
\node[left,left of = tip, node distance = 0.8 cm] (ltop) {$f_s$};
\draw[->,dashed,line width = 1pt] (ltop) to [out=0,in=90] (smp_out.north);
\node [right of = smp_out, node distance = 2cm] (right_edge) {};
\node [below of = right_edge] (right_b_edge) {};

           \node [right] (dest) [below of=source]{$\hat{X}(\cdot)$};
         \node [int] (dec) [right of=dest, node distance = 2cm] {$\mathrm{Dec}$};
                \node [int] (enc) [right of = dec, node distance = 3cm]{$\mathrm{Enc}$};

  \draw[-,line width=2pt] (smp_out) -- node[above, xshift = 0.5cm] {$Y[\cdot]$} (right_edge);
  \draw[-,line width = 2]  (right_edge.west) -| (right_b_edge.east);
    \draw[->,line width = 2]  (right_b_edge.east) -- (enc.east);
   \draw[->,line width=2pt] (enc) -- node[above] {$R$} (dec);

   \draw[->,line width=2pt] (dec) -- (dest);
    \draw[->,line width=2pt] (source) -- (sum);
    \draw[->,line width=2pt] (noise) -- (sum);
    \draw[->, line width = 2pt] (sum) --  (smp_in);
    \draw[line width=2pt]  (smp_in) -- (tip);
    \draw[dashed, line width = 1pt] (smp_in)+(-0.5,-0.5) -- +(-0.5,1) -- node[above] {sampler} +(1.5,1) -- + (1.5,-0.5) -- +(-0.5,-0.5);
    
\end{tikzpicture}
\caption{\label{fig:operational_scheme} Combined sampling and source coding model.}
\end{center}
\end{figure}


In this work we focus on uniform sampling of Gaussian stationary processes under quadratic distortion, using single branch and multi-branch uniform sampling. We determine the expression for the three-dimensional manifold representing the trade-off among $f_{s}$, $R$ and $D$ in terms of the power spectral density (PSD) of the source, the noise and the sampling mechanism. In addition, we derive an expression for the optimal pre-sampling filter and the corresponding minimal distortion attainable under any such uniform sampling scheme. This minimal distortion provides a lower bound on the distortion achieved by any A/D conversion scheme with uniform sampling. In this sense, the distortion-rate sampling frequency function associated with our model quantifies the excess distortion incurred when source encoding is based on the information in any uniform sub-Nyquist sampling scheme of a Gaussian stationary source in lieu of the full source information about the analog source. \\

It is important to emphasize that in the model in Fig.~\ref{fig:operational_scheme} and throughout the paper, the bitrate $R$ is fixed and represents the number of \emph{bits per time unit} rather then the number of bits per sample. In particular, this model does not capture memory and quantization constraints of the samples at the encoder. This means that for any fixed $R$, minimal distortion is achieved by taking $f_s$ greater than or equal to $f_{Nyq}$, the Nyquist frequency of $X(\cdot)$, such that $X(\cdot)$ can be reconstructed from the samples $Y[\cdot]$ with zero error. In particular, our model shows no benefit for oversampling schemes, in agreement with the observations in \cite{370112} and \cite{335948}. In practice, memory and computational constraints may preclude the encoder from processing information at high sampling rates or high quantizer resolution. Our setting provides a distortion-rate bound regardless of the actual implementation of the ADC, which may be a sampler followed by a scalar quantizer as in pulse code modulation or with a feedback loop as in Sigma-Delta modulation \cite{sklar2001digital}. Such constraints on the encoder (not included in our model) lead to an interesting trade-off between sampling frequency and the number of bits per sample, which is investigated in \cite{KipnisAllerton2015}. 
%

\begin{figure}
\begin{center}
\begin{tikzpicture}

\node[int,align=center,minimum width=2cm] (root) at (0,0) {combined \\
source coding \\
and sampling};

\node[int,align=center,minimum width=2cm] (right) at (3,-2) {MMSE under \\
sub-Nyquist \\
sampling};

\node[int,align=center, minimum height = 1cm, minimum width=2cm] (left) at (-3,-2) {indirect \\ source \\ coding};

\draw [->, line width=1pt] (root) -| node[above, xshift = 0.5cm] {$f_s > f_{Nyq}$} (left);

\draw [->, line width=1pt] (root) -| node[above, xshift = -0.5cm] {$R\rightarrow \infty$} (right);

\end{tikzpicture}

\caption{\label{fig:diagram_scheme} The source coding problem of Fig.~\ref{fig:operational_scheme} subsumes two classical problems in information theory and signal processing.}
\end{center}
\end{figure}
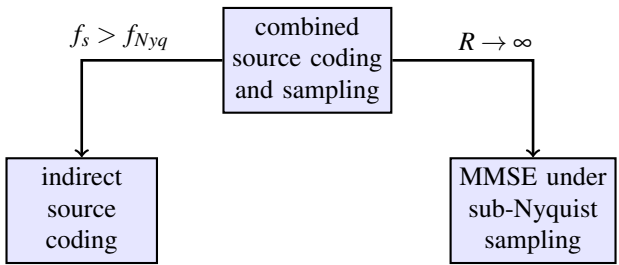

\subsection{Related work}

Shannon derived the quadratic DRF of a Gaussian bandlimited white-noise source \cite[Thm. 22]{Shannon1948}. Shannon's expression was extended to continuous-time Gaussian stationary sources with arbitrary PSD by Pinsker  and Kolmogorov in \cite{1056823}, now known as the Shannon-Kolmogorov-Pinsker (SKP) reverse waterfilling expression \cite{Berger1998}. Gel’fand and Yaglom \cite{gelfand1959calculation} used the Karhunen-Lo\'{e}ve expansion of the source over a finite time interval to map the continuous-time problem back to a discrete-time problem, and in this way provided a first source coding theorem for second order continuous-time stationary processes under quadratic distortion. A source coding theorem for a more general class of continuous-time sources was later proved by Berger \cite{Berger1968254}. In addition to these source coding theorems Berger \cite[Sec. 4.5.3]{berger1971rate} also suggested an approach to source coding based on mapping the continuous-time waveform to its discrete-time representation by sampling at increasingly high rates.
Berger did not resolved various technical difficulties that arise in this approach, such as the convergence of mutual information and error in vector quantization as the sampling rate increases. Pinsker showed that the mutual information between a pair of continuous-time Gaussian stationary processes can be approximated by the mutual information of their values over finite sets \cite{pinsker1964information}. Although this result settles some of the difficulties with the sampling approach to continuous-time source coding, Pinsker did not discuss it in the context of source coding theory. The sampling approach to continuous-time source coding was only recently settled in \cite{neuhoff2013information} by studying the behavior of a vector quantizer as the sampling frequency approaches infinity. In view of these papers, it is quite remarkable that SKP reverse waterfilling provides the minimal distortion theoretically achievable in any digital representation of a continuous-time source, regardless of the way the time index is discretized or the specific mapping of the analog waveform to a finite alphabet set.

Since in our model the encoder needs to deliver information about the source but cannot observe it directly, the problem of characterizing the DRF falls within the regime of \textit{indirect} or \textit{remote source coding} problems \cite[Section 3.5]{berger1971rate}. Indirect source coding problems were first introduced by Dobrushin and Tsybakov in \cite{1057738}, where a closed form expression was derived in the case where the observable process and the source are jointly Gaussian and stationary. We refer to this setting as the stationary Gaussian indirect source coding problem. This setting is a special case of our model when the sampled process can be fully reconstructed from its samples, which happens for example when $X(\cdot)$ is bandlimited and sampled above its Nyquist frequency. In their work, Dobrushin and Tsybakov implicitly showed that quadratic indirect source coding can be separated into two independent problems: minimal mean squared error (MMSE) estimation and standard (direct) source coding. A single shot version of this separation was investigated by Wolf and Ziv in \cite{1054469}. An additional analysis of this separation result was given by Witsenhausen in \cite{1056251}, who viewed it as a special case of a reduced distortion measure which holds in indirect source coding under any fidelity criterion. These results are discussed in detail in Section~\ref{sec:Indirect-Source-Coding}. \par

The other branch of the diagram in Fig.~\ref{fig:diagram_scheme} is obtained if we relax the rate constraint in the model in Fig. \ref{fig:operational_scheme}. The distortion at a given sampling frequency is then simply the MMSE in estimating $X\left(\cdot\right)$ from its noisy sub-Nyquist samples $\mathbf Y\left[\cdot\right]$. An expression for this MMSE as well as a description of the optimal pre-sampling filter that minimizes it were derived in \cite{1090615} for single branch sampling. See also \cite{815501} and \cite{4663942} for a simple derivation. In particular, the MMSE expression establishes the sufficiency of uniform sampling above the Nyquist frequency for perfect reconstruction of random stationary signals, a fact which was first noted in \cite{1057404}. A necessary and sufficient sampling rate for perfect stable reconstruction of a sampled signal is the Lebesgue measure of the support of its spectrum, or the \emph{spectral occupancy} of the signal. This condition was derived by Landau \cite{1447892}, which in fact considered the more general setting of non-uniform sampling \cite{Landau1967}, although without a pre-sampling operation. Nevertheless, it follows from \cite{YuxinNonUniform} that the spectral occupancy, now termed the \emph{Landau rate},\footnote{Although denoted the \emph{Nyquist rate} by Landau himself in \cite{1447892}.} is the minimal sampling frequency that allows zero MSE under uniform sampling even when linear pre-processing is allowed. One way to achieve zero error at the Landau rate is by employing \emph{multi-branch sampling}, in which the input is passed through $P$ independent branches of linear filters and uniform samplers. This sampling strategy was proposed by Papoulis in \cite{Papoulis1977}. The MMSE in multi-branch sampling, as well as the optimal pre-sampling filters that minimize it, were implicitly derived in \cite{ShannonMeetsNyquist}. It was shown there that the optimal pre-sampling filters that maximize the capacity of a channel with sampling at the receiver are the same filters that minimize the MMSE in sub-Nyquist sampling. These optimal pre-sampling filters are designed to select a set of frequency bands with maximal signal to noise ration (SNR) while preventing aliasing in each sampling branch. This is an extension of a characterization of the optimal pre-sampling filter in single branch sampling given in \cite{1090615}. These results on the MMSE in sub-Nyquist sampling will be discussed in more detail in Section~\ref{sec:mmse}. 

\subsection{Main contributions}
The main result of this paper is a closed form expression for the function $D\left(f_{s},R\right)$ which represents the minimal quadratic distortion achieved in the reconstruction of any continuous time Gaussian stationary processes from its rate $R$ uniform noisy samples at frequency $f_s$. This is shown to be given by a parametric \emph{reverse waterfilling} expression, which in the case of single branch sampling takes the form
\begin{subequations}
\label{eq:main_result_intro}
\begin{align}
R(f_s,\theta)= \frac{1}{2} \int_{-\frac{f_s}{2}}^\frac{f_s}{2} \log^+\left[ \widetilde{S}_{X|Y}(f) /\theta\right]df,
\end{align}
\begin{align}
D(f_s,R) =\sigma_X^2- \int_{-\frac{f_s}{2}}^\frac{f_s}{2}\left[ \widetilde{S}_{X|Y}(f)-\theta \right]^+df,
\end{align}
\end{subequations}
where the function $\widetilde{S}_{X|Y}(f)$ is defined in terms of the sampling frequency $f_s$, the pre-sampling filter $H(f)$, the PSD of the source $S_X(f)$ and the PSD of the noise $S_\eta(f)$. The proof of \eqref{eq:main_result_intro} relies on an extension of the stationary Gaussian indirect source coding problem considered by Dobrushin and Tsybakov to vector-valued processes, which is given by Theorem \ref{thm:indirect_vector_case}. 
\par
The result of Dobrushin and Tsybakov was obtained for the case where the source and the observable process are jointly Gaussian and stationary. In our setting the observable discrete time process $Y\left[\cdot\right]$ and the analog processes $X\left(\cdot\right)$ are still jointly Gaussian, but the optimal reconstruction process under quadratic distortion, $\left\{ \mathbb{E}\left[X\left(t\right)|Y\left[\cdot\right]\right],\,t\in \mathbb R\right\}$, is in general not a stationary process. An easy way to see this is to consider the estimation error at the sampling times $t\in\mathbb Z/f_s$ in the noiseless case $\eta(\cdot) \equiv 0$, which must vanish, while the estimation error at any $t\notin \mathbb Z/f_s$ is not necessarily zero. In Section \ref{sec:Frequency-Rate-Distortion-Functi} we present a way to overcome this difficulty. The idea is to use time discretization, after which we can identify a vector-valued process jointly stationary with the samples $Y\left[\cdot\right]$ which contains the same information as the discretized version of $X\left(\cdot\right)$. The result is the indirect DRF at any given sampling frequency in a discrete-time version of our problem, which converges to $D\left(f_{s},R\right)$ under mild conditions.\par

In practice, the system designer may choose the parameters of the sampling mechanism to achieve minimal reconstruction error for a given sampling frequency $f_s$ and source coding rate $R$. This suggests that for a given source statistic and a sampling frequency $f_{s}$, an optimal choice of the pre-sampling filters can further reduce the distortion for a given source coding rate. In the single branch setting, this optimization is carried out in Subsection \ref{subsec:Optimal-pre-Sampling-Filter} and leads to the function $D^{\star}\left(f_s,R\right)$, which gives a lower bound on $D(f_s,R)$ and is only a function of the source and noise PSDs. The optimal pre-sampling filter $H(f)$ is shown to pass only one frequency in each discrete aliasing set $f+f_s\mathbb Z$ and suppress the rest. In other words, minimal distortion is achieved by eliminating aliasing. \par
We later extend our results to systems with $P\in \mathbb N$ sampling branches where the samples are represented by a vector-valued process $\mathbf Y[\cdot]$. We derive expressions for $D(P,f_s,R)$ and $D^\star(P,f_s,R)$, which denote the DRF with average sampling frequency $f_s$ and the DRF under optimal pre-sampling filtering, respectively. As the number of sampling branches $P$ goes to infinity, $D^\star(P,f_s,R)$ is shown to converge (but not monotonically, see Fig.~\ref{fig:DistMulti}) to a smaller value $D^\dagger(f_s,R)$, which essentially describes the minimal distortion achievable under any uniform sampling scheme. The functions $D^\star(P,f_s,R)$ and $D^\dagger(f_s,R)$ depend only on the statistics of the source and the noise. In particular, if the noise is zero, then $D^{\star}\left(P,f_s,R\right)$ and $D^\dagger\left(f_s,R\right)$ describe a fundamental trade-off in signal processing and information theory associated with any Gaussian stationary source. \par
Our main result \eqref{eq:main_result_intro} shows that the function $D(f_s,R)$ is obtained by reverse waterfilling over the function $\widetilde{S}_{X|Y} (f)$ that was initially introduced to calculate the MMSE in sub-Nyquist sampling in \cite{1090615} and \cite{815501}, denoted by $\mmse_{X|Y}(f_s)$. As a result, the optimal pre-sampling filter that minimize $D(f_s,R)$ is the same optimal pre-sampling filter that minimizes $\mmse_{X|Y}(f_s)$. In Section~\ref{sec:mmse} we prove this result using an approach based on a decomposition of the signal to its \emph{polyphase components}. We also define the notion of an \emph{aliasing-free set}, and use it to describe the optimal pre-sampling filter. This approach allows us to derive the MMSE in sub-Nyquist sampling using multi-branch uniform sampling with the number of branches going to infinity. This polyphase approach to deriving the MMSE also inspires the derivation of our main result \eqref{eq:main_result_intro}. We note that the fact that the function $\widetilde{S}_{X|Y} (f)$ is used in computing both $\mmse_{X|Y}(f_s)$ and $D(f_s,R)$ is not related to recent results on the relation between mutual information and MMSE estimation \cite{guo2005mutual}. Indeed, no information measure over a Gaussian channel is explicitly considered in our setting.

\subsection{Organization}

The rest of the paper is organized as follows: the combined sampling and source-coding problem is presented in Section~\ref{sec:problem_statement}. An overview of the main results in a simplified version of the problem is given in  Section~\ref{sec:overview}. Sections~\ref{sec:mmse} and~\ref{sec:Indirect-Source-Coding} are dedicated to the special cases of sub-Nyquist sampling ($R\rightarrow \infty$) and indirect source coding ($f_s > f_{Nyq}$), as shown in the respective  branches in the diagram of Fig.~\ref{fig:diagram_scheme}. In Section~\ref{sec:Frequency-Rate-Distortion-Functi} we prove our main results for single branch sampling, which is extended to multi-branch sampling in Section~\ref{sec:main_multi}. Concluding remarks are given in Section~\ref{sec:Concluding-Remarks}. \\

Throughout this paper, we use round brackets and square brackets to distinguish between continuous-time and discrete-time processes. Vectors and matrices are denoted by bold letters. In addition, we use the word `rate' to indicate information rate rather than sampling rate, and use `sampling frequency' for the latter. In some cases it is more convenient to measure the information rate in bits per sample, which is given by $\bar{R} \triangleq R / f_s$.  \par

\section{Problem Statement} \label{sec:problem_statement}

The system model for our combined sampling and source coding problem is depicted in Fig.~\ref{fig:operational_scheme}. The source $X\left(\cdot\right)=\left\{ X\left(t\right),\, t\in\mathbb{R}\right\}$ is a real Gaussian stationary process with variance $\sigma_{X}^{2}\triangleq\int_{-\infty}^{\infty}S_{X}\left(f\right)df<\infty$,  and power spectral density (PSD)
\[
S_X\left(f\right) \triangleq \int_{-\infty}^\infty\mathbb E\left[X(t+\tau)X(t)  \right]e^{-2\pi i \tau f}d\tau.
\]
The noise $\eta(\cdot)$ is a real Gaussian stationary process independent of the source with PSD $S_\eta(f)$. The sampler receives the noisy source as an input, and produces a discrete time process $\mathbf Y[\cdot]$ at a rate of $f_s$ samples per time unit. The process $\mathbf Y[\cdot]$ is in general a complex vector-valued process since pre-sampling operations that result in a complex valued process are allowed in the sampler. The encoder represents the samples $\mathbf Y[\cdot]$ in an average rate of no more than $R$ bits per time unit. Assuming the noise is additive and independent poses no limitation on the generality. Indeed, for any jointly stationary and Gaussian process pairs $X(\cdot)$ and $Z(\cdot)$, this relationship can be created via a linear transformation, which can be seen as part of the sampler structure. When the optimal sampling structure in this case is considered, the results can be adjusted by a straightforward reweighing of the PSDs of $Z(\cdot)$ and $\eta(\cdot)$ according to the frequency response of this transformation.
\par
The main problem we consider is as follows: given a sampling scheme with sampling frequency $f_s$, what is the minimal expected quadratic distortion that can be attained between $X\left(\cdot\right)$ and $\hat{X}\left(\cdot\right)$ over all encoder-decoder pairs with code-rate that does not exceed $R$ bits per time unit, as $T$ goes to infinity? \par
Classical results in rate-distortion theory \cite{1057738,1055440,pinsker1964information} imply that this problem has the informational rate-distortion characterization depicted in Fig. \ref{fig:The-general-scheme}, where the optimization over all encoding-decoding pairs is replaced by an optimization over the test channel $\mathrm{P}_{\hat{X}|\mathbf Y}$ of limited information rate. \\

\begin{figure}
\begin{center}
\begin{tikzpicture}[node distance=2cm,auto,>=latex]
  \node at (0,0) (source) {$X(\cdot)$} ;
  \node [ssum, right of = source, node distance = 2cm] (sum) {+};
  \node [above of = sum,node distance = 1.5cm] (noise) {$\eta(\cdot)$};
  \node [coordinate, right of = sum,node distance = 2.5cm] (smp_in) {};
  \node [coordinate, right of = smp_in,node distance = 1cm] (smp_out){};
  \node [coordinate,above of = smp_out,node distance = 0.6cm] (tip) {};
\fill  (smp_out) circle [radius=2pt];
\fill  (smp_in) circle [radius=3pt];
\fill  (tip) circle [radius=2pt];
\node[left,left of = tip, node distance = 0.8 cm] (ltop) {$f_s$};
\draw[->,dashed,line width = 1pt] (ltop) to [out=0,in=90] (smp_out.north);
\node [right of = smp_out, node distance = 2cm] (right_edge) {};
\node [below of = right_edge] (right_b_edge) {};

           \node [right] (dest) [below of=source]{$\hat{X}(\cdot)$};
                \node [int] (enc) [right of = dec, node distance = 3cm]{$\mathrm P_{\hat{X}| \mathbf Y}$};
                \node [above of = enc, node distance = 0.8cm] {$I_T \left( \hat{X};\mathbf Y \right) \leq R$};

  \draw[-,line width=2pt] (smp_out) -- node[above, xshift = 0.5cm] {$\mathbf Y[\cdot]$} (right_edge);
  \draw[-,line width = 2]  (right_edge.west) -| (right_b_edge.east);
    \draw[->,line width = 2]  (right_b_edge.east) -- (enc.east);

   \draw[->,line width=2pt] (enc) -- (dest);
    \draw[->,line width=2pt] (source) -- (sum);
    \draw[->,line width=2pt] (noise) -- (sum);
    \draw[->, line width = 2pt] (sum) -- (smp_in);
    \draw[line width=2pt]  (smp_in) -- (tip);
    \draw[dashed, line width = 1pt] (smp_in)+(-0.5,-0.5) -- +(-0.5,1) -- node[above] {sampler} +(1.5,1) -- + (1.5,-0.5) -- +(-0.5,-0.5);
    
\end{tikzpicture}
\caption{\label{fig:The-general-scheme}  Rate-distortion representation.}
\end{center}
\end{figure}
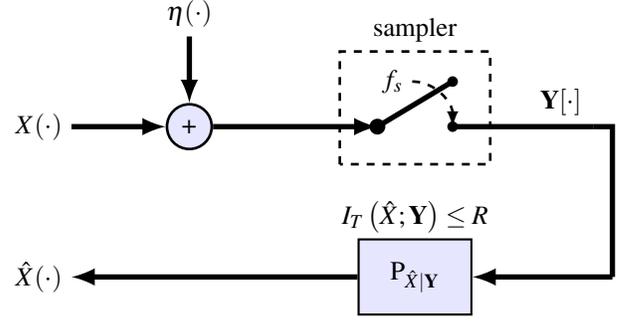

Specifically, for a finite $T>0$, denote by $X_T(\cdot)$ the restriction of the process $X(\cdot)$ to the interval $[-T,T]$. Similarly denote by $\mathbf Y_T[\cdot]$ the restriction of the process $\mathbf Y[\cdot]$ obtained by sampling $X_T(\cdot)$. The fidelity criterion is defined by the squared error between the original source and its reconstruction $\hat{X}(\cdot)=\left\{\hat{X}(t),\,t\in\mathbb R\right\}$, namely
\begin{align} 
d_T\left(\hat{x}\left(\cdot\right),x\left(\cdot\right)\right)& \triangleq \|\hat{x}(\cdot)-x(\cdot)\|_T^2  \label{eq:dist_def}
\end{align}
where $\left\|x(\cdot) \right\|_T$ is the $L_2$ norm of the signal $x(\cdot)$ over the interval $[-T,T]$, defined by
\[
\left\|{x}(\cdot)\right\|_T^2\triangleq \frac{1}{2T}\int_{-T}^{T} \left(x(t)\right)^{2} dt. 
\]
Define the function 
\begin{equation} \label{eq:d_T_def}
D_T  \triangleq \inf_{Y\overset{R}{\longrightarrow} \hat{X}} \mathbb E\,d_T \left(X(\cdot),\hat{X} \right)
\end{equation}
where the infimum is taken over all mappings from $\mathbf Y_T[\cdot]$ to $\hat{X}(\cdot)$ such that the mutual information rate
\[
I_T \left(\mathbf Y[\cdot]; \hat{X}(\cdot) \right) \triangleq \frac{1}{T} I\left(\mathbf Y_T[\cdot]; \hat{X}(\cdot) \right)
\]
is limited to $R$ bits per time unit. The indirect distortion-rate function (iDRF) of $X(\cdot)$ given $Y[\cdot]$, denoted by $D_{X|\mathbf Y}$, is defined by
\begin{equation} \label{eq:D_def}
D =  \liminf_{T\rightarrow \infty} D_T.
\end{equation}
Note that the number of samples in the interval $[-T,T]$ and consequently the number of bits per sample $\bar{R}$, is a function of the specific structure of the sampler which will be defined in the sequel. For example, for a uniform sampler with spacing $1/f_s$ between samples we have $\bar{R}=R/f_s$.  \\

Besides the sampling frequency $f_s$ and the source coding rate $R$, $D$ in \eqref{eq:D_def} depends on the sampling structure. In this work we restrict ourselves to samplers consisting of a pre-sampling filtering operation followed by a pointwise sampler. We focus on two basic structures:

\begin{figure}
\begin{center}
\begin{subfigure}{0.48\textwidth}
\centering

\begin{tikzpicture}[node distance=2cm,auto,>=latex]

\node at (0,0) (source) {$X(\cdot)$};
\node[coordinate] (source_up) [above of = source,node distance = 1cm]{};

\node[coordinate] (first_jnc) [right of = source, node distance=1.5cm] {};
  
\draw[-, line width=1pt] (source)--(first_jnc);   

\node[int1]  (pre_sampling2) [right of = first_jnc, node distance=0.8cm]{$H(f)$};  

\draw[->, line width=1pt] (first_jnc)--(pre_sampling2);   
  	  
\node [coordinate, right of = pre_sampling2,node distance = 1.7cm] (smp_in2) {};
  \node [coordinate, right of = smp_in2,node distance = 0.7cm] (smp_out2){};
	\node [coordinate,above of = smp_out2,node distance = 0.4cm] (tip2) {};
\fill  (smp_out2) circle [radius=2pt];
\fill  (smp_in2) circle [radius=2pt];
\fill  (tip2) circle [radius=2pt];
\node[left,left of = tip2, node distance = 0.5 cm] (ltop2) {$f_s$};

\node [right of = smp_out2, node distance=3cm]  (out) {$Y[\cdot]$};

\draw[->,densely dotted,line width = 1pt,thin] (ltop2) to [out=0,in=70] (smp_out2.north);
 \draw[line width=1pt]  (smp_in2) -- (tip2);
 \draw[-,line width=1pt]   (pre_sampling2)-- node[above] {\small $Z(\cdot)$}(smp_in2);

\draw[line width=1pt]  (smp_in2) -- (tip2);

\draw[->,line width = 1pt] (smp_out2) -- (out); 

\draw[line width=1pt, dashed] (0.9,0.8) rectangle (6,-0.7) ;
    
\end{tikzpicture}
\caption{single-branch sampler \vspace{10pt} }
\end{subfigure}

\begin{subfigure}{0.48\textwidth}
\centering
\begin{tikzpicture}[node distance=2cm,auto,>=latex]

\node at (0,0) (source) {$X(\cdot)$};
\node[coordinate] (source_up) [above of = source,node distance = 1cm]{};

\node[coordinate] (first_jnc) [right of = source, node distance=1.2cm] {};
\fill  (first_jnc) circle [radius=2pt];
  
\draw[->, line width=1pt] (source)--(first_jnc);   

\node[int1]  (pre_sampling2) [right of = first_jnc, node distance=1cm]{$H_2(f)$};  

\draw[->, line width=1pt] (first_jnc)--(pre_sampling2);   
  	  
\node [coordinate, right of = pre_sampling2,node distance = 1.7cm] (smp_in2) {};
  \node [coordinate, right of = smp_in2,node distance = 0.7cm] (smp_out2){};
	\node [coordinate,above of = smp_out2,node distance = 0.4cm] (tip2) {};
\fill  (smp_out2) circle [radius=2pt];
\fill  (smp_in2) circle [radius=2pt];
\fill  (tip2) circle [radius=2pt];
\node[left,left of = tip2, node distance = 0.5 cm] (ltop2) {$f_s/P$};

\node [coordinate] (enc_left) [right of = source, node distance = 6cm] {};
\node [right of = enc_left]  (out) {$\mathbf Y[\cdot]$};

\draw[->,densely dotted,line width = 1pt,thin] (ltop2) to [out=0,in=70] (smp_out2.north);
 \draw[line width=1pt]  (smp_in2) -- (tip2);
 \draw[-,line width=1pt]   (pre_sampling2)-- node[above] {\small $Z_2(\cdot)$}(smp_in2);
 \draw[->,line width=1pt] (smp_out2) -- node[above] {$Y_2[\cdot]$} (enc_left);
\draw[line width=1pt]  (smp_in2) -- (tip2);

\node[coordinate] (enc_bot) [below of = enc_left, node distance = 1.3cm] {};
\node[coordinate] (enc_top) [above of = enc_left, node distance = 0.8cm] {};

\fill (enc_bot) circle [radius=3pt];
\fill (enc_top) circle [radius=3pt];
\fill (enc_left) circle [radius=3pt];

\draw[line width = 3pt] (enc_top) -- (enc_bot); 
\draw[->,line width = 3pt] (enc_left) -- (out); 

\node [below of=enc_bot, node distance = 1cm] (right_down){};

\node[int1]  (pre_sampling3) [below of = pre_sampling2, node distance=1.3cm]{$H_P(f)$};  

\draw[->,line width=1pt] (first_jnc)|-(pre_sampling3);   

\draw[-,dotted, line width = 1pt] (pre_sampling2) -- (pre_sampling3) ;  	  
\node [coordinate, right of = pre_sampling3,node distance = 1.7cm] (smp_in3) {};
  \node [coordinate, right of = smp_in3,node distance = 0.7cm] (smp_out3){};
	\node [coordinate,above of = smp_out3,node distance = 0.4cm] (tip3) {};
\fill  (smp_out3) circle [radius=2pt];
\fill  (smp_in3) circle [radius=2pt];
\fill  (tip3) circle [radius=2pt];
\node[left,left of = tip3, node distance = 0.5 cm] (ltop3) {\small $f_s/P$};

\draw[->,dashed,densely dotted,line width = 1pt,thin] (ltop3) to [out=0,in=70] (smp_out3.north);
 \draw[line width=1pt]  (smp_in3) -- (tip3);
 \draw[-,line width=1pt]   (pre_sampling3)--node[above] {\small $Z_P(\cdot)$} (smp_in3);
 
\draw[->,line width=1pt] (smp_out3) -- node[above] {\small $ Y_P[\cdot]$} (enc_bot);
\draw[line width=1pt]  (smp_in3) -- (tip3);

\node[int1]  (pre_sampling1) [above of = pre_sampling2, node distance=0.8cm]{$H_1(f)$};  

\draw[->, line width=1pt] (first_jnc)|-(pre_sampling1);   
  	  
\node [coordinate, right of = pre_sampling1,node distance = 1.7cm] (smp_in1) {};
  \node [coordinate, right of = smp_in1,node distance = 0.7cm] (smp_out1){};
	\node [coordinate,above of = smp_out1,node distance = 0.4cm] (tip1) {};
\fill  (smp_out1) circle [radius=2pt];
\fill  (smp_in1) circle [radius=2pt];
\fill  (tip1) circle [radius=2pt];
\node[left,left of = tip1, node distance = 0.5 cm] (ltop1) {\small $f_s/P$};

\draw[->,dashed,densely dotted,line width = 1pt,thin] (ltop1) to [out=0,in=70] (smp_out1.north);
 \draw[line width=1pt]  (smp_in1) -- (tip1);
 \draw[-,line width=1pt]   (pre_sampling1)--node[above] {\small $Z_1(\cdot)$}(smp_in1);
 
\draw[->,line width=1pt] (smp_out1) -- node[above] {$Y_1[\cdot]$}  (enc_top);
\draw[line width=1pt]  (smp_in1) -- (tip1);

\draw[line width=1pt, dashed] (0.9,1.5) rectangle (7,-2) ;
 \end{tikzpicture}
\vspace{-10pt}
\caption{multi-branch sampler}
\end{subfigure}
\caption{Two sampling schemes. \label{fig:sampling_scheme}}
\end{center}
\end{figure}
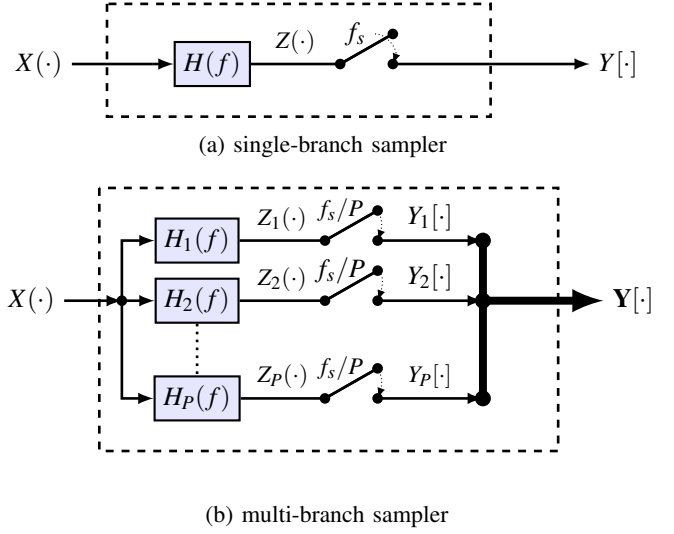

\subsubsection{Single-branch uniform sampling (Fig. \ref{fig:sampling_scheme}-a)} 
$H$ is an LTI system with frequency response $H(f)$ which serves as a pre-sampling filter. This means that the input to the pointwise sampler $Z(\cdot)=\left\{ Z(t),\, t\in\mathbb{R}\right\}$ and $X(\cdot)$ are jointly Gaussian and stationary with joint spectral density
\begin{align*}
S_{XZ}(f)&\triangleq \int_{-\infty}^\infty  \mathbb E\left[X(t+\tau)Z(t) \right]e^{-2\pi i \tau f}d\tau = S_X(f)H^*(f).
\end{align*} 
Although we allow an arbitrary noise PSD $S_{\eta}(f)$, in order for the uniform sampling operation to be well defined we require that
\begin{equation}
\label{eq:Sz_finite}
\int_{-\infty}^\infty
 S_{Z}\left(f\right)df=\int_{-\infty}^\infty S_{X+\eta}(f) \left|H(f)\right|^2 df< \infty.
\end{equation}
In \eqref{eq:Sz_finite} and henceforth we denote $S_{X+\eta}(f)\triangleq S_X(f)+S_\eta(f)$, which is justified since $X(\cdot)$ and $\eta(\cdot)$ are independent processes. 
We sample $Z(\cdot)$ uniformly at times $\frac{n}{f_s}$, resulting in the discrete time process 
\[
Y[n]=Z\left(\frac{n}{f_{s}}\right),\quad n\in\mathbb{Z}. 
\]
Recall that the spectral density of $Y[\cdot]$ is given by
\begin{align*}
S_Y\left(e^{2\pi i \phi}\right)&=\sum_{k\in \mathbb Z}\mathbb E\left[Y[n]Y[n+k] \right]e^{-2\pi i k \phi}\\
&=\sum_{k\in \mathbb Z}f_sS_Z\left(f_s(\phi-k)\right).
\end{align*}
We denote by $D(f_s,R)$ the iDRF \eqref{eq:D_def} using uniform single-branch sampling at frequency $f_s$. 
\subsubsection{Multi-branch or filter-bank uniform sampling (Fig. \ref{fig:sampling_scheme}-b)}  For each $p=1,\ldots,P$, $Z_p(\cdot)$ is the output of the LTI system $H_p$ whose input is the source $X(\cdot)$. The sequence $Y_p[\cdot]$ is obtained by uniformly sampling $Z_p(\cdot)$ at frequency $f_s/P$, i.e. 
\[
Y_p[n]=Z\left(\frac{nP}{f_s}\right),\quad p=1,\ldots,P.
\]
The output of the sampler is the vector $\mathbf Y[\cdot]=\left(Y_1[\cdot],\ldots,Y_p[\cdot]\right)$.  
Since each one of the $P$ branches produces samples at rate $f_s/P$, the sampling frequency of the system is $f_s$. The iDRF \eqref{eq:D_def} of $X(\cdot)$ given the vector process $\mathbf Y[\cdot]$ will be denoted $D(f_s,R)$. \\

The parameters of the two sampling schemes above are the average sampling frequency $f_s$ and the pre-sampling filters $H(f)$ or $H_1(f),\ldots, H_P(f)$. Given an average sampling frequency $f_s$ and a source coding rate $R$, we also consider the following question: what are the optimal pre-sampling filters that minimize $D(P,f_s,R)$? The value of $D(P,f_s,R)$  under an optimal choice of the pre-sampling filters is denoted by $D^\star(P,f_s,R)$, and is only a function of $f_s$, $R$, $P$ and the source and noise statistics. We also determine the behavior of $D^\star(P,f_s,R)$ as the number of branches $P$ goes to infinity. \\

\section{Overview of the Main Results \label{sec:overview}}
In this section we provide an overview of the main results under the simplified assumptions of a single branch sampler $P=1$ and no noise $\eta(\cdot) \equiv 0$.\par
Our first main result in Theorem~\ref{thm:main_result} implies that under these assumptions, 
the function $D(f_s,R)$ is given by the following parametric form
\begin{equation} \label{eq:pre_main_result}
\begin{split}
R\left(f_{s},\theta\right)  = & \frac{1}{2}\int_{-\frac{f_s}{2}}^{\frac{f_s}{2}}\log^{+}\left[\widetilde{S}_{X|Y}(f) /\theta\right]df, \\
D\left(f_{s},\theta\right)  = &  \mmse_{X|Y} (f_s)+\int_{-\frac{f_s}{2}}^{\frac{f_s}{2}}\min\left\{ \widetilde{S}_{X|Y}(f) ,\theta\right\} df ,
\end{split}
\end{equation}
where 
\begin{align} \label{eq:pre_J_def}
\widetilde{S}_{X|Y}(f) &= \frac{\sum_{k\in\mathbb{Z}} \left|H\left(f-f_s k\right)\right|^2  S_{X}\left(f-f_sk\right)^{2}}{\sum_{k\in\mathbb{Z}} \left|H\left(f-f_s k\right)\right| S_{X}\left(f-f_sk\right)},
\end{align}
and 
\begin{align}
\mmse_{X|Y} (f_s) \triangleq \int_{-\infty}^\infty \left[ S_X(f)-\widetilde{S}_{X|Y}(f) \right] df, \label{eq:pre_mmse_def}
\end{align}
is the MMSE in estimating $X(\cdot)$ from its uniform samples $Y[\cdot]$. The parametric solution \eqref{eq:pre_main_result} has the reverse waterfilling interpretation described in Fig.~\ref{fig:unimodal_waterfilling}. Note that the function $D(f_s,R)$ converges to $\mmse_{X|Y}(f_s)$ as $R\rightarrow \infty$, and to 
the DRF of $X(\cdot)$ as $f_s$ exceeds the Nyquist frequency of $X(\cdot)$. This agrees with the diagram in Fig.~\ref{fig:diagram_scheme}.  \par

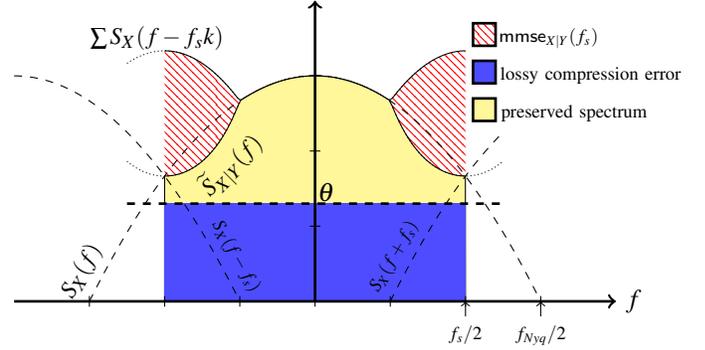
\begin{figure}

\begin{center}

\begin{tikzpicture}[scale=1]

\fill[fill=red!50, pattern=north west lines, pattern color=red] (-2,0) -- plot[domain=-2:2, samples=100]  (\x, {-(\x)*(\x)/3+3  +max(-(\x-4)*(\x-4)/3+3,0)+max(-(\x+4)*(\x+4)/3+3,0) }) -- (2,0);	

\draw  plot[domain=-2:2, samples=100] (\x, {-(\x)*(\x)/3+3  +max(-(\x-4)*(\x-4)/3+3,0)+max(-(\x+4)*(\x+4)/3+3,0)});	

\draw[densely dotted]  plot[domain=-2.5:1.6, samples=100] (\x, {-(\x)*(\x)/3+3  +max(-(\x-4)*(\x-4)/3+3,0)+max(-(\x+4)*(\x+4)/3+3,0)});	

\draw[dashed] plot[domain=-3:3, samples=100] (\x, {-\x*\x/3+3})  ;	

\fill[fill=yellow!50] (-2,0) -- plot[domain=-2:2, samples=100]  (\x, { ((-(\x)*(\x)/3+3)^(2) +(max(-(\x-4)*(\x-4)/3+3,0))^(2)+(max(-(\x+4)*(\x+4)/3+3,0) )^(2) + (max(-(\x+8)*(\x+8)/3+3,0) )^(2) +(max(-(\x-8)*(\x-8)/3+3,0))^(2) )/ (-(\x)*(\x)/3+3  +max(-(\x-4)*(\x-4)/3+3,0)+max(-(\x+4)*(\x+4)/3+3,0) +max(-(\x-8)*(\x-8)/3+3,0)+max(-(\x-8)*(\x-8)/3+3,0)  ) } ) -- (2,0);	

\draw (-2,0) -- plot[domain=-2:2, samples=100]  (\x, { ((-(\x)*(\x)/3+3)^(2) +(max(-(\x-4)*(\x-4)/3+3,0))^(2)+(max(-(\x+4)*(\x+4)/3+3,0) )^(2) + (max(-(\x+8)*(\x+8)/3+3,0) )^(2) +(max(-(\x-8)*(\x-8)/3+3,0))^(2) )/ (-(\x)*(\x)/3+3  +max(-(\x-4)*(\x-4)/3+3,0)+max(-(\x+4)*(\x+4)/3+3,0) +max(-(\x-8)*(\x-8)/3+3,0)+max(-(\x-8)*(\x-8)/3+3,0)  ) } ) -- (2,0) ;

\draw[densely dotted] plot[domain=-2.5:2.5, samples=100]  (\x, { ((-(\x)*(\x)/3+3)^(2) +(max(-(\x-4)*(\x-4)/3+3,0))^(2)+(max(-(\x+4)*(\x+4)/3+3,0) )^(2) + (max(-(\x+8)*(\x+8)/3+3,0) )^(2) +(max(-(\x-8)*(\x-8)/3+3,0))^(2) )/ (-(\x)*(\x)/3+3  +max(-(\x-4)*(\x-4)/3+3,0)+max(-(\x+4)*(\x+4)/3+3,0) +max(-(\x-8)*(\x-8)/3+3,0)+max(-(\x-8)*(\x-8)/3+3,0)  ) } ) ;

\fill[fill=blue!70] (-2,0) -- (-2,1.3) --(2,1.3) -- (2,0)  ;	
 
\draw[dashed] plot[domain=1:2.5, samples=100] (\x, {-(\x-4)*(\x-4)/3+3})  ;	
\draw[dashed] plot[domain=-4:-1, samples=100] (\x, {-(\x+4)*(\x+4)/3+3})  ;

\foreach \x in {-3,-2,-1,0,1,2,3}
 \draw[shift={(\x,0)}] (0pt,2pt) -- (0pt,-2pt);
  \foreach \y/\ytext in {1/,2/,3/}
  \draw[shift={(0,\y)}] (2pt,0pt) -- (-2pt,0pt) node[left] {$\ytext$};

\draw[->] (2,-0.2) node[below] {\scriptsize $f_s/2$} -- (2,0);
\draw[->] (3,-0.2) node[below] {\scriptsize $f_{Nyq}/2$} -- (3,0);


\draw [fill=red!50, line width=1pt, pattern=north west lines, pattern color=red] (2.1,3.4) rectangle  (2.4,3.7) node[left, xshift = 1.5cm, yshift = -0.2cm] {\scriptsize $\mmse_{X|Y}(f_s)$};

\draw [fill=blue!70, line width=1pt] (2.1,2.9) rectangle  (2.4,3.2) node[left, xshift = 2.6cm, yshift = -0.2cm, align = center] {\scriptsize lossy compression error};

\draw [fill=yellow!50, line width=1pt] (2.1,2.4) rectangle  (2.4,2.7) node[left, xshift=2.15cm, yshift = -0.2cm,align = left] {\scriptsize preserved spectrum};

\node at (0.15,1.45) {$\theta$};
\draw[dashed, line width=1pt] (-2.5,1.3) -- (2.5,1.3);
\draw[->,line width=1pt]  (-4,0)--(4,0) node[right] {$f$};
\draw[->,line width=1pt]  (0,0)--(0,4) node[right] {};
\node at (-3.1,0.4) [rotate=60] {\small $S_X(f)$};
\node at (1.05,0.6) [rotate=60] {\scriptsize $S_X(f+f_s)$};
\node at (-1.05,0.6) [rotate=-60] {\scriptsize $S_X(f-f_s)$};
\node at (-1.1,1.8) [rotate=45] {\small $\widetilde{S}_{X|Y}(f)$};
\node at (-2.1,3.5) [rotate=0] {\small $\sum S_X(f-f_sk)$};

\end{tikzpicture}
\end{center}
\vspace{-10pt}

\caption{ \label{fig:unimodal_waterfilling}
Reverse waterfilling interpretation of \eqref{eq:pre_main_result}: The function $D(f_s,R)$ is given by the sum of the sampling error $\mmse_{X|Y}(f_s)$ and the lossy compression error $\int_{-f_s/2}^{f_s/2}\min\left\{ \widetilde{S}_{X|Y}(f),\theta \right\}df $. The function $\sum_{k\in \mathbb Z} S_X(f-f_sk)$ is the   aliased PSD, which represents the full energy of the original signal within the band $(-f_s/2,f_s/2)$. The part of the energy recovered by the MMSE estimator is $\widetilde{S}_{X|Y}(f)$. 
}
\end{figure}

Next, we turn to find an expression for the optimal pre-sampling filter $H^\star(f)$ that minimizes the function $D(f_s,R)$. Since $H(f)$ appears in both nominator and denominator of \eqref{eq:pre_J_def}, its magnitude has no effect on the distortion and all that matters is whether $H(f)$ is zero (in which case we interpret \eqref{eq:pre_J_def} as zero) or not. Proposition~\ref{prop:opt_single} implies that $H^\star(f)$ is an anti-aliasing filter that passes the frequency bands with the highest SNR (and in the non-noisy case with highest energy), and suppresses the rest to prevent aliasing. We intuitively explain this result through Example~\ref{ex:joint_mmse}. In the special case where $S_X(f)$ is unimodal in the sense that it is non-increasing for $f>0$, $H^\star(f)$ is a simple low-pass filter with cut-off frequency $f_s/2$. The iDRF in this setting is given as:
\begin{equation} \label{eq:pre_main_optimal}
\begin{split}
R\left(f_{s},\theta\right)  = & \frac{1}{2}\int_{-\frac{f_s}{2}}^{\frac{f_s}{2}}\log^{+}\left[ S_{X}(f)/\theta  \right]df, \\
D^\star\left(f_{s},\theta\right)  = &  \mmse^\star_{X|Y} (f_s)+\int_{-\frac{f_s}{2}}^{\frac{f_s}{2}}\min\left\{ S_{X}(f),\theta\right\} df,
\end{split}
\end{equation}
where 
\[
\mmse^\star_{X|Y}(f_s) \triangleq \int_{-\infty}^\infty S_X(f) df - \int_{-\frac{f_s}{2}}^{\frac{f_s}{2}} S_{X}(f) df.
\]
Fig.~\ref{fig:unimodal_waterfilling_opt} provides an intuitive interpretation of \eqref{eq:pre_main_optimal} as a sum of two terms: the error due to sampling and the error due to lossy compression. The situation in the general case in which $S_X(f)$ is not unimodal is less intuitive: it is generally impossible to define a single pre-sampling filter that passes the frequencies with the maximal SNR and simultaneously eliminates aliasing. In such cases, it is useful to consider multi-branch sampling with a set of optimal pre-sampling filters. The expression for the corresponding iDRFs in multi-branch sampling and the iDRF under an optimal choice of such filters is given in Section~\ref{sec:main_multi}. \\

Since the iDRF \eqref{eq:D_def} is always bounded from below by the MMSE in estimating $X(\cdot)$ from $\mathbf Y[\cdot]$, we devote the following section to discuss the behavior the MMSE in sub-Nyquist sampling and the optimal choice of the pre-sampling filter that minimizes this error.

\begin{figure}
\begin{center}

\begin{tikzpicture}[scale=1]

\fill[fill=red!50, pattern=north west lines, pattern color=red]  plot[domain=-3:3, samples=100]  (\x, {-(\x)*(\x)/3+3 }) ;	

\draw  plot[domain=-3:3, samples=100]  (\x, {-(\x)*(\x)/3+3 }) ;	

\fill[fill=yellow!50] (-2,0) --  plot[domain=-2:2, samples=100]  (\x, {-(\x)*(\x)/3+3 }) --(2,0) ;	

\fill[fill=blue!70] (-2,0) -- (-2,1.1) --(2,1.1) -- (2,0)  ;	
 
\foreach \x in {-3,-2,-1,0,1,2,3}
 \draw[shift={(\x,0)}] (0pt,2pt) -- (0pt,-2pt);
  \foreach \y/\ytext in {1/,2/,3/}
  \draw[shift={(0,\y)}] (2pt,0pt) -- (-2pt,0pt) node[left] {$\ytext$};

\draw[->] (2,-0.2) node[below] {\scriptsize $f_s/2$} -- (2,0);
\draw[->] (3,-0.2) node[below] {\scriptsize $f_{Nyq}/2$} -- (3,0);


\draw [fill=red!50, line width=1pt, pattern=north west lines, pattern color=red] (2,3.4) rectangle  (2.3,3.7) node[left, xshift = 1.5cm, yshift = -0.2cm] {\scriptsize $\mmse^\star_{X|Y}(f_s)$};

\draw [fill=blue!70, line width=1pt] (2,2.9) rectangle  (2.3,3.2) node[left, xshift = 2.6cm, yshift = -0.2cm, align = center] {\scriptsize lossy compression error};

\draw [fill=yellow!50, line width=1pt] (2,2.4) rectangle  (2.3,2.7) node[left, xshift=2.1cm, yshift = -0.2cm,align = left] {\scriptsize preserved spectrum};

\node at (0.13,1.25) {$\theta$};
\draw[dashed, line width=1pt] (-2.6,1.1) -- (2.6,1.1);
\draw[->,line width=1pt]  (-4,0)--(4,0) node[right] {$f$};
\draw[->,line width=1pt]  (0,0)--(0,4) node[right] {};
\node at (-1.4,2.7) [rotate=40] {\small $S_X(f)$};

\end{tikzpicture}
\end{center}
\vspace{-10pt}

\caption{ \label{fig:unimodal_waterfilling_opt}
Reverse waterfilling interpretation of \eqref{eq:pre_main_optimal}: The function $D^\star(f_s,R)$ of a unimodal $S_X(f)$ and zero noise is given by the sum of the sampling error $\mmse^\star_{X|Y}(f_s)$ and the lossy compression error $\int_{f_s/2}^{f_s/2}\min\left\{ S_{X}(f),\theta \right\}df $. 
}
\end{figure}
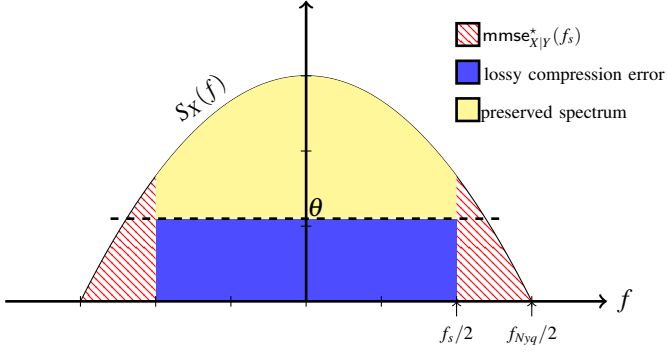


\section{MMSE in sub-Nyquist sampling of a Gaussian stationary process \label{sec:mmse}}

\begin{figure}
\begin{center}
\begin{tikzpicture}[node distance=2cm,auto,>=latex]
  \node at (0,0) (source) {$X(\cdot)$} ;
  \node [ssum, right of = source, node distance = 2cm] (sum) {+};
  \node [above of = sum,node distance = 1.5cm] (noise) {$\eta[\cdot]$};
  \node [coordinate, right of = sum,node distance = 2.5cm] (smp_in) {};
  \node [coordinate, right of = smp_in,node distance = 1cm] (smp_out){};
  \node [coordinate,above of = smp_out,node distance = 0.6cm] (tip) {};
\fill  (smp_out) circle [radius=2pt];
\fill  (smp_in) circle [radius=3pt];
\fill  (tip) circle [radius=2pt];
\node[left,left of = tip, node distance = 0.8 cm] (ltop) {$f_s$};
\draw[->,dashed,line width = 1pt] (ltop) to [out=0,in=90] (smp_out.north);
\node [right of = smp_out, node distance = 2cm] (right_edge) {};
\node [below of = right_edge] (right_b_edge) {};

           \node [right] (dest) [below of=source]{$\hat{X}(\cdot)$};
        \node [int] (dec) [right of=dest, node distance = 4cm,align = center] {MMSE \\ estimator};
               
  \draw[-,line width=2pt] (smp_out) -- node[above, xshift = 0.5cm] {$\mathbf Y[\cdot]$} (right_edge);
  \draw[-,line width = 2]  (right_edge.west) -| (right_b_edge.east);
    \draw[->,line width = 2]  (right_b_edge.east) -- (dec.east);

   \draw[->,line width=2pt] (dec) -- (dest);
    \draw[->,line width=2pt] (source) -- (sum);
    \draw[->,line width=2pt] (noise) -- (sum);
    \draw[->, line width = 2pt] (sum) -- (smp_in);
    \draw[line width=2pt]  (smp_in) -- (tip);
    \draw[dashed, line width = 1pt] (smp_in)+(-0.5,-0.5) -- +(-0.5,1) -- node[above] {sampler} +(1.5,1) -- + (1.5,-0.5) -- +(-0.5,-0.5);
\end{tikzpicture}

 \caption{\label{fig:mmse_model} System model for MMSE reconstruction under sub-Nyquist sampling.}
 \end{center}
\end{figure}
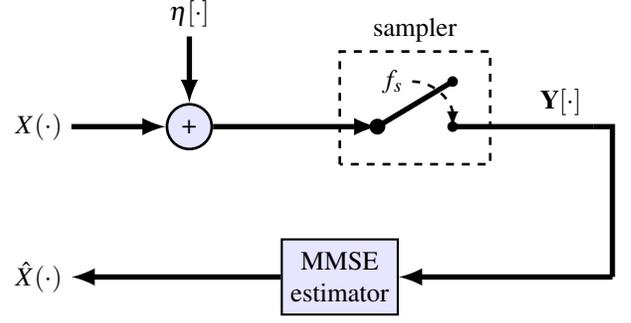

In this section we consider the right side of the diagram in Fig.~\ref{fig:diagram_scheme}, obtained from the general sampling and source coding problem of Fig. \ref{fig:The-general-scheme} with no rate constraint on the source encoder. This leaves us with the system model in Fig.~\ref{fig:mmse_model}, in which the problem we consider is to find the MMSE in estimating the source $X(\cdot)$ from samples $\mathbf Y[\cdot]$ which we denote as $\mmse_{X|\mathbf Y}(f_s)$:
\begin{equation} \label{eq:mmse_def}
 \mmse_{X|\mathbf Y}(f_s) \triangleq \lim_{T\rightarrow \infty} \mathbb E \| X(\cdot)- \tilde{X}(\cdot)\|_T^2,
\end{equation}
where $\tilde{X}(t)\triangleq\mathbb E\left[ X(t)| \mathbf Y[\cdot]\right]$, $t\in \mathbb R$, is the MMSE estimator of $X(\cdot)$ from its sub-Nyquist  samples $\mathbf Y[\cdot]$. In addition, we are interested in the optimal pre-sampling filter that minimizes the MMSE for a given input signal and sampling frequency. A single branch version of this problem without the pre-sampling filter can be found in \cite{815501}. The pre-sampling filter and its optimization is included in \cite{1090615}, where a similar setting was considered with applications in pulse-code modulation. The multi-branch case was solved in \cite[Prop. 3]{ShannonMeetsNyquist}, but the expression for the optimal pre-sampling filters was not explicit and relayed on a different proof. The main contribution of this section is a new way to prove these result, which is based on a polyphase decomposition of the source. The new proves of the above results provided here will be useful in proving our main results in Section \ref{sec:Frequency-Rate-Distortion-Functi}. \\

\subsection*{MMSE via polyphase components}

Since the instantaneous estimation error $X(t)-\tilde{X}(t)$ is periodic in $t$ hence uniformly bounded, \eqref{eq:mmse_def} can be written as
\begin{align} 
  &\mmse_{X|\mathbf Y}(f_s)  = \lim_{T\rightarrow \infty} \frac{1}{2T} \int_{-T}^T \mathbb E\left[\left(X(t)-\tilde{X}(t)\right)^2 \right]dt \nonumber  \\
&  =  \int_0^1 \lim_{N\rightarrow \infty}  \frac{1}{2N+1} \sum_{n=-N}^{N} \mathbb E \left[\left(X\left(\frac{n+\Delta}{f_s}\right)-\tilde{X}\left(\frac{n+\Delta}{f_s}\right) \right)^2 \right]  d\Delta  \nonumber \\
&  \quad \quad \quad \quad \quad  = \int_0^1 \mmse_{X|\mathbf Y} (f_s) d\Delta,
 \label{eq:mmse_average_equiv} 
\end{align}
where the process $X_\Delta[\cdot]$ is the \emph{$\Delta$ polyphase component} of $X(\cdot)$ \cite{1162788}, defined by
\begin{equation} \label{eq:X_delta_def}
X_\Delta[n]\triangleq X\left(\frac{n+\Delta}{f_s}\right),\quad n\in \mathbb Z,
\end{equation}
and $\tilde{X}_\Delta[n]\triangleq \mathbb E\left[ X_\Delta[n]|\mathbf Y[\cdot]\right]$. In \eqref{eq:mmse_average_equiv} we also denoted 
\[
\mmse_{X_\Delta|\mathbf Y} \triangleq  \frac{1}{2N+1} \sum_{n=-N}^{N} \mathbb E \left[\left(X\left(\frac{n+\Delta}{f_s}\right)-\tilde{X}\left(\frac{n+\Delta}{f_s}\right) \right)^2 \right] .
\]
Since $X_\Delta[\cdot]$ and $\mathbf Y[\cdot]$ are jointly Gaussian and stationary, $\mmse_{X_\Delta|\mathbf Y}$ can be evaluated using linear estimation techniques. For the single branch sampler of Fig.~\ref{fig:sampling_scheme}(a), this leads to 
\begin{prop} \label{prop:mmse_single} Consider the model of Fig.~\ref{fig:mmse_model} where we use the single branch sampler of Fig.~\ref{fig:sampling_scheme}(a). The MMSE in estimating $X(\cdot)$ from $Y[\cdot]$ is given by
 \begin{equation} \label{eq:mmse_single_theorem}
 \mmse_{X|Y}(f_s)=\sigma_X^2-\int_{-\frac{f_s}{2}}^{\frac{f_s}{2}} \widetilde{S}_{X|Y}(f)df,
 \end{equation}
 where $\sigma_X^2 = \mathbb E \left(X(t)\right)^2$ and
 \begin{align}
\widetilde{S}_{X|Y}(f)\triangleq \frac{\sum_{k\in \mathbb Z}S_X^2(f-f_sk)\left|H(f-f_sk)\right|^2}{\sum_{k\in \mathbb Z} S_{X+\eta}(f-f_sk)\left|H(f-f_sk)\right|^2}.
\label{eq:j_def}
\end{align} 
\end{prop}
\begin{proof} 
This result obtained by evaluating \eqref{eq:mmse_average_equiv}. Details can be found in Appendix~\ref{sec:proof_mmse}. 
\end{proof}

Note that since the denominator in \eqref{eq:j_def} is periodic in $f$ with period $f_s$, \eqref{eq:mmse_single_theorem} can be written as
\begin{align} 
\mmse_{X|Y}(f_s)&=\sigma_X^2-\int_{-\infty}^\infty \frac{S_X^2(f)|H(f)|^2}{\sum_{k\in \mathbb Z} S_{X+\eta}(f-f_sk)|H(f-f_sk)|^2}df \nonumber \\
=\int_{-\infty}^\infty S_X(f)&\left(1-\frac{S_X(f)|H(f)|^2}{\sum_{k\in \mathbb Z} S_{X+\eta}(f-f_sk)|H(f-f_sk)|^2} \right)df. \label{eq:J_modified}
\end{align}
This shows that the expression for $\mmse_{X|Y}(f_s)$ in Proposition~\ref{prop:mmse_single} is equivalent to the \cite[Eq. 10]{815501}. The alternate proof of this proposition given here using the new expression for the MMSE given in \eqref{eq:mmse_average_equiv} provides a new interpretation of the function $\widetilde{S}_{X|Y}(f)$ as the average of spectral densities of estimators of the stationary polyphase components of $X(\cdot)$, namely
\begin{equation} \label{eq:J_as_average}
\widetilde{S}_{X|Y}(f) = \int_0^1 {f_s}S_{X_\Delta|Y}(f/f_s) d\Delta.
\end{equation}

%
%
%

\subsection{An optimal pre-sampling filter \label{subsec:mmse_optimal}}
We now consider the pre-sampling filter $H$ as part of the system design and ask what is the optimal pre-sampling filter $H^\star$ that minimizes \eqref{eq:mmse_single_theorem}; as is apparent from \eqref{eq:mmse_single_theorem}, this problem is equivalent to finding the filter that maximizes $\widetilde{S}_{X|Y}(f)$ for every frequency $f\in \left(-f_s/2,f_s/2\right)$ independently, i.e. we are looking to determine
\begin{align}
\widetilde{S}_{X|Y}^\star(f) & \triangleq \sup_H \widetilde{S}_{X|Y}(f) \nonumber
\\ &= \sup_{H}\frac{\sum_{k\in\mathbb{Z}}S_{X}^{2}(f-f_sk)\left|H(f-f_sk)\right|^{2}}{\sum_{k\in\mathbb{Z}}S_{X+\eta}(f-f_sk)\left|H(f-f_sk)\right|^{2}}  \label{eq:sup_J}
\end{align}
in the domain $\left(-f_s/2,f_s/2\right)$. Note that scaling $H(f)$ has an equal effect on the nominator and denominator in \eqref{eq:sup_J} and hence the optimal $H(f)$ can only be specified by its support, i.e., those frequencies which are not blocked by the filter. \par
In what follows we will describe $H^\star(f)$ by defining a set of frequencies $F^\star$ of minimal Lebesgue measure such that
\begin{equation} \label{eq:maximal_set_pre}
\int_{F^\star} \frac{S_X^2(f)}{S_{X+\eta}(f)}df = \int_{-\frac{f_s}{2}}^\frac{f_s}{2} \sup_{k\in \mathbb Z} \frac{S_X^2(f-f_sk)}{S_{X+\eta}(f-f_sk)} df.
\end{equation}
Since the integrand in the right hand side (RHS) of \eqref{eq:maximal_set_pre} is periodic in $f$ with period $f_s$, excluding a set of Lebesgue measure zero, the set $F^\star$ will not contain two frequencies $f_1,f_2\in \mathbb R $ that differ by an integer multiple of $f_s$ due to its minimality. This property will be given the following name:
\begin{definition}[aliasing-free set] \label{def:aliasing_free} A measurable set $F \subset \mathbb R$ is said to be  \emph{aliasing-free} with respect to the sampling frequency $f_s$ if, for almost\footnote{By \emph{almost any} we mean \emph{for all but a set of Lebesgue measure zero}.} all pairs $f_1,f_2\in F$, it holds that $f_1-f_2 \notin f_s\mathbb Z=\left\{f_s k,\,k\in \mathbb Z \right\}$. \par
\end{definition} 
The aliasing-free property imposes the following restriction on the Lebesgue measure of a bounded set:

\begin{prop} \label{prop:aliasing_free}
Let $F$ be an aliasing-free set with respect to $f_s$. If $F$ is bounded, then the Lebesgue measure of $F$ does not exceed $f_s$.
\end{prop}

\begin{proof}
By the aliasing-free property, for any $n\in \mathbb Z \setminus \{0\}$ the intersection of $F$ and $F^\star + nf_s$ is empty. It follows that for all $N\in \mathbb N$, $\mu \left( \cup_{n=1}^N \left\{ F^\star+f_s n \right\} \right)=N\mu(F^\star)$. Now assume $F^\star$ is bounded by the interval $(-M,M)$ for some $M>0$. Then $\cup_{n=1}^N \left\{ F^\star+f_s n \right\}$ is bounded by the interval $(-M,M+Nf_s)$. It follows that
\[
\frac{\mu(F^\star)}{f_s}=\frac{N\mu(F^\star)}{Nf_s}=\frac{\mu(\cup_{n=1}^N\left\{F^\star+nf_s\right\})}{Nf_s} \leq \frac{2M+Nf_s}{Nf_s}.
\]
Letting $N\rightarrow \infty$, we conclude that $\mu(F^\star)\leq f_s$. \\
\end{proof}
We denote by $AF(f_s)$ the collection of all bounded aliasing free sets with respect to $f_s$. Note that a process whose spectrum's support is contained in $AF(f_s)$ admits no aliasing when uniformly sampled at frequency $f_s$, i.e, such a process can be reconstructed with probability one from its non-noisy uniform samples at frequency $f_s$ \cite{Lloyd1959}. As the following theorem shows, the optimal pre-sampling filter is characterized by an aliasing-free set with an additional maximality property.

\begin{theorem}
\label{thm:mmse_opt_single} For a fixed $f_s$, the optimal pre-sampling filter $H^\star(f)$ that maximizes $\widetilde{S}_{X|Y}(f)$, $f\in(-f_s/2,f_s/2)$ and minimizes $ \mmse_{X|Y}(f_s)$ is given by
\begin{equation} \label{eq:Hstar_def}
H^{\star}\left(f\right)=\begin{cases}
1 & f\in F^\star,\\
0 & \text{otherwise},
\end{cases}
\end{equation}
where $F^\star = F^\star\left(f_s, \frac{S_X^2(f)}{S_{X+\eta}(f)}\right)\in AF(f_s)$ satisfies
\begin{equation}
\label{eq:maximal_aliasing_free}
\int_{F^\star}  \frac{S_X^2(f)}{S_{X+\eta}(f)} df=\sup_{F\in AF(f_s)}
\int_{F}  \frac{S_X^2(f)}{S_{X+\eta}(f)} df.
\end{equation}
The optimal MMSE for sampling at frequency $f_s$ is 
\begin{align}
\mmse^\star_{X|Y}(f_s) & =  \sigma_{X}^{2}-\int_{F^\star} \frac{S_X^2(f)}{S_{X+\eta}(f)} df,
\label{eq:minimal_mse}
\end{align}
where $\sigma_X^2 = \mathbb E \left(X(t)\right)^2$.
\end{theorem}

\begin{proof}
See Appendix~\ref{sec:proof_mmse_opt_single}. \\
\end{proof}
\subsubsection*{Remarks}
\begin{enumerate}
\item[(i)]
The proof also shows that 
\[
\int_{F^\star}  \frac{S_X^2(f)}{S_{X+\eta}(f)} df=
\int_{-\frac{f_s}{2}}^\frac{f_s}{2}  \widetilde{S}^\star_{X|Y}(f) df,
\]
where
\[
\widetilde{S}^\star_{X|Y}(f)  \triangleq  \sup_k \frac{S_X^2(f-f_sk)}{S_{X+\eta}(f-f_sk)},
\]
i.e. 
\[
\mmse^\star_{X|Y} (f_s)=\sigma_X^2-\int_{-\frac{f_s}{2}}^\frac{f_s}{2}  \widetilde{S}^\star_{X|Y}(f) df.
\]
\item[(ii)]
Since the SNR at each spectral line $f$ cannot be changed by $H(f)$, the filter $H^\star(f)$ can be specified only in terms of its support,
i.e. in \eqref{eq:Hstar_def} we may replace $1$
by any non-zero value, which can vary with $f$. \\
\end{enumerate}

Theorem \ref{thm:mmse_opt_single} motivates the following definition:
\begin{definition}
\label{def:maximal_aliasing_free} For a given spectral density $S(f)$ and a sampling frequency $f_s$,
an aliasing free set $F^\star \in AF(f_s)$ that satisfies
\[
\int_{F^\star} S(f) df= \sup_{F\in AF(f_s)} \int_{F} S(f) df
 \]
is called a \emph{maximal aliasing-free set} with respect to $f_s$ and the spectral density $S(f)$. Such a set will be denoted by $F^\star \left(f_s,S \right)$.
\end{definition}

Roughly speaking, the maximal aliasing free set $F^\star\left(f_s,S\right)$ can be constructed by going over all frequencies $f\in \left(-f_s/2,f_s/2\right)$, and including in $F^\star\left(f_s,S\right)$  the frequency $f^\star \in \mathbb R$ such that $S(f^\star)$ is maximal among all $S(f)$, $f\in f^\star-f_s \mathbb Z$. Since the estimator is aware of the PSD of the source, in order to estimate $X(\cdot)$ it needs only collect energy and avoid aliasing so that the signal can be uniquely identified. The question is whether there is an interesting interplay between collecting energy and preventing aliasing. Theorem~\ref{thm:mmse_opt_single} says that the optimal pre-sampling filter prefers to eliminate aliasing on the price of completely suppressing the energy of weaker bands. An intuition for this result is given through the following example.

\begin{figure}
\begin{center}
\begin{tikzpicture}[node distance=2cm,auto,>=latex]
  \node at (0,0.7) (s1) {$U_1$} ;
  \node at (0,0) (center) {};
  \node at (0,-0.7) (s2) {$U_2$} ;
  \node [ssum, right of = s1, node distance = 1.5cm] (sum1) {+};
  \node [ssum, right of = s2, node distance = 1.5cm] (sum2) {+};
  \node [above of = sum1,node distance = 1.3cm] (n1) {$\xi_1$};
    \node [below of = sum2,node distance = 1.3cm] (n2) {$\xi_2$};
    \node[int1, right of = sum1, node distance = 1.2cm] (h1) {$h_1$};
        \node[int1, right of = sum2, node distance = 1.2cm] (h2) {$h_2$};
        \node[ssum, right of = center, node distance = 3.6cm] (fv) {+};
	\node[int1, right of = fv, node distance = 2cm, align = center] (est) {$\mmse$ \\ estimator};
	\node[right of = est, node distance = 2cm] (dest) {$\hat{U}_1,\hat{U}_2$};

       \draw[->,line width=2pt] (sum1) -- (h1);
	  \draw[->,line width=2pt] (sum2) -- (h2);
    \draw[->,line width=2pt] (h1) -| (fv);
     \draw[->,line width=2pt] (h2) -| (fv);
   \draw[->,line width=2pt] (fv) -- node[above] {$V$} (est);
   \draw[->,line width=2pt] (s1) -- (sum1);
    \draw[->,line width=2pt] (s2) -- (sum2);
    \draw[->,line width=2pt] (n1) -- (sum1);
        \draw[->,line width=2pt] (n2) -- (sum2);
        \draw[->,line width=2pt] (est) -- (dest);
    
\end{tikzpicture}
\caption{\label{fig:exampe_joint_mmse}  Joint MMSE estimation from a linear combination.}
\end{center}
\end{figure}
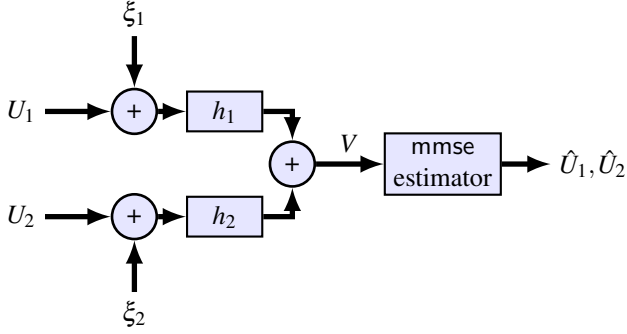

\begin{example}[joint MMSE estimation] \label{ex:joint_mmse}
 Consider the setting in Fig.~\ref{fig:exampe_joint_mmse}, where $U_1$ and $U_2$ be two independent Gaussian random variables with variances $C_{U_1}$ and $C_{U_2}$ respectively. We are interested in MMSE estimation of $\mathbf U=\left(U_1,U_2\right)$ from a noisy linear combination of their sum $V=h_1(U_1+\xi_1)+h_2(U_2+\xi_2)$, where $h_1,h_2 \in \mathbb R$ and $\xi_1,\xi_2$ are another two Gaussian random variables with variances $C_{\xi_1}$ and $C_{\xi_2}$ respectively, independent of $U_1$ and $U_2$ and independent of each other.  We have
\begin{align}
\label{eq:mmse_twoRV}
\mmse_{\mathbf U|V} & =\frac{1}{2}\left(\mmse_{U_1|V}+\mmse_{U_2|V}\right)\\
 = \frac{1}{2}\left(C_{U_1}\right.&+\left.C_{U_2}-\frac{h_1^2C_{U_1}^2+h_2^2 C_{U_2}^2}{h_1^2(C_{U_1}+C_{\xi_1})+h_2^2(C_{U_2}+C_{\xi_2})}\right). \nonumber 
\end{align}
The optimal choice of the coefficients vector $\mathbf h=\left(h_1,h_2\right)$ that minimizes \eqref{eq:mmse_twoRV} is 
\[
\mathbf h= 
\begin{cases}
\left(c,0\right) & \frac{C_{U_1}^2}{C_{U_1}+C_{\xi_1}} >  \frac{C_{U_2}^2}{C_{U_2}+C_{\xi_2}} \\
\left(0,c\right) & \frac{C_{U_1}^2}{C_{U_1}+C_{\xi_1}} <  \frac{C_{U_2}^2}{C_{U_2}+C_{\xi_2}},
\end{cases}
\]
where $c$ is any constant different from zero. If $\frac{C_{U_1}^2}{C_{U_1}+C_{\xi_1}} =  \frac{C_{U_2}^2}{C_2+C_{\xi_2}}$, then any non-trivial linear combination results in the same estimation error. 
\end{example}
This example can be generalized to a countable number of random variables $\mathbf U=\left(U_1,U_2,\ldots\right)$ and respective noise sequence $\mathbf \xi=\left(\xi_1,\xi_2,\ldots\right)$ such that $V=\sum_{i=1}^\infty h_i(U_i+\xi_i)<\infty$ with probability one. The optimal coefficient vector $\mathbf h=\left(h_1,h_2,\ldots \right)$ that minimizes $\mmse_{\mathbf U|\mathbf V}$ is the indicator for the maximum among
\[
 \left\{\frac{C_{U_i}^2}{C_{U_i}+C_{\xi_i}},~i=1,2,\ldots\right\}.
 \]

In the context of the expression for the optimal pre-sampling filter \eqref{eq:Hstar_def}, each frequency $f$ in the support of $S_X(f)$ can be seen as an independent component of the process $X(\cdot)$ with spectrum $\approx \mathbf 1_{[f,f+\Delta f)}S_X(f)$ (see for example the derivation of the SKP reverse waterfilling in \cite{Berger1998}).  
For a given $f\in\left( -f_s/2,f_s/2 \right)$, the analogue for the vector $\mathbf U$ in our case are the components of the source process that corresponds to the frequencies $f-f_s\mathbb Z$, which are folded and summed together due to aliasing: each set of the form $f-f_s\mathbb Z$ corresponds to a linear combination of a countable number of independent Gaussian random variables attenuated by the coefficients $\left\{H(f-f_sk),~k\in \mathbb Z\right\}$. The optimal choice of coefficients that minimizes the MMSE in joint estimation of all source components are those that pass only the spectral component with maximal $\frac{S_X^2(f')}{S_{X+\eta}(f')}$ among all $f'\in f-f_s\mathbb Z$, and suppress the rest. This means that under the MSE criterion, the optimal choice is to eliminate aliasing at the price of losing all information contained in spectral components other than the maximal one. \par 
An example of a maximal aliasing-free set for a specific PSD appears in Fig.~\ref{fig:aliasing_free_various}. The MMSE with the optimal pre-sampling filter and with an all-pass filter are shown in Fig.~\ref{fig:mmse_fig}.\\

\begin{figure}
\centering
\begin{tikzpicture}
\node[int] at (0,0) {\includegraphics[trim=1cm 2.6cm 1cm 2cm, clip=true, scale=0.45]{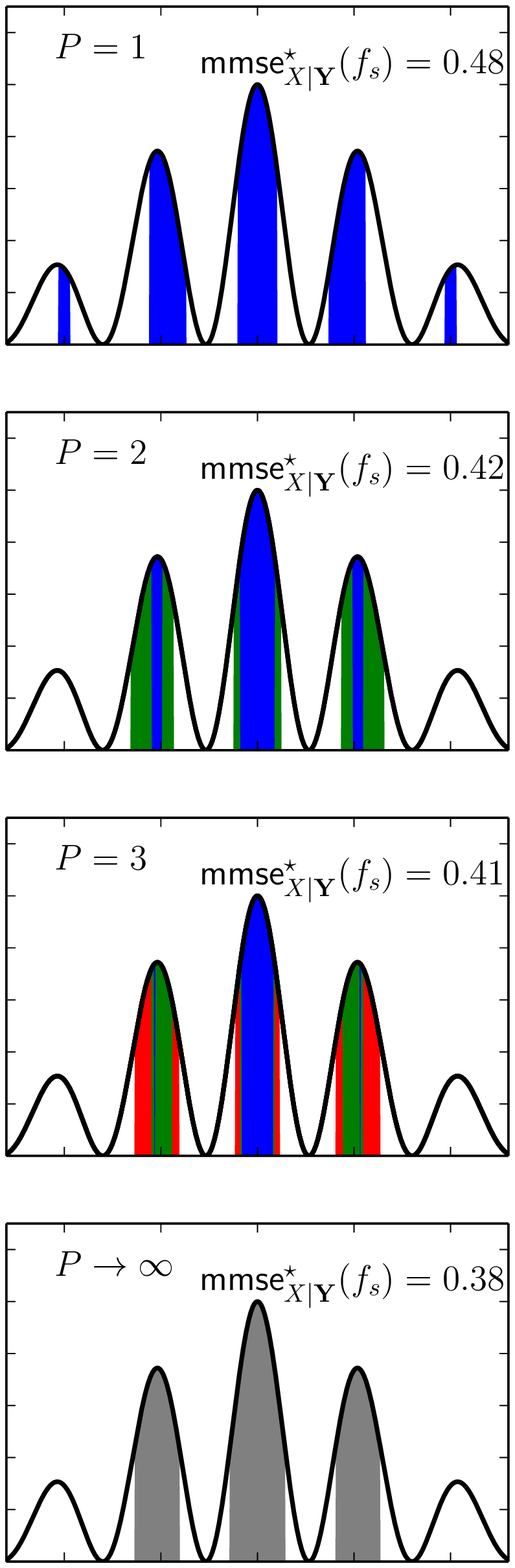}};
\node[int] at (4.5,0) {\includegraphics[trim=1cm 2.6cm 1cm 2cm, clip=true, scale=0.45] {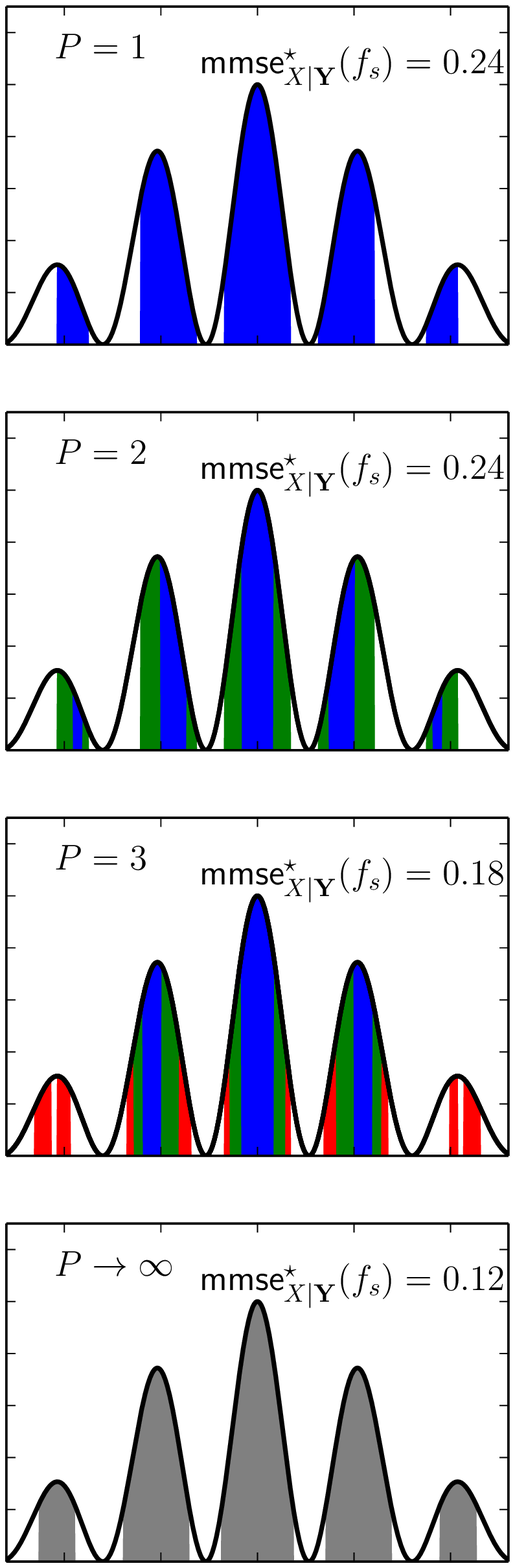}};
\end{tikzpicture}
\caption{\label{fig:aliasing_free_various}
Maximal aliasing-free sets with respect to the PSD $S_X^2(f)/S_{X+\eta}(f)$ and sampling frequencies $f_s = f_{Nyq}/4$ (left) and $f_s =  f_{Nyq}/2$ (right), for $1,2$ and $3$ sampling branches. The first, second and third maximal aliasing-free set is given by the frequencies below the blue, green and red areas, respectively. The sets below the total colored area all have Lebesgue measure $f_s$. Assuming $S_\eta(f)\equiv 0$, the white area bounded by the PSD equals $\mmse^\star_{X|\mathbf Y}(f_s)$. The ratio of this area to the total area bounded by the PSD is specified in each case. From Theorem~\ref{thm:mmse_Landau}, the case $P\rightarrow \infty$ corresponds to the set $\mathcal F^\star$ that achieves the RHS of \eqref{eq:mmse_bound_landau}.
}
\end{figure}

\begin{figure}
\begin{center}
\begin{tikzpicture}
\node at (0,0) {\includegraphics[trim=0cm 4cm 1.5cm 3cm, clip=true, scale=0.4]{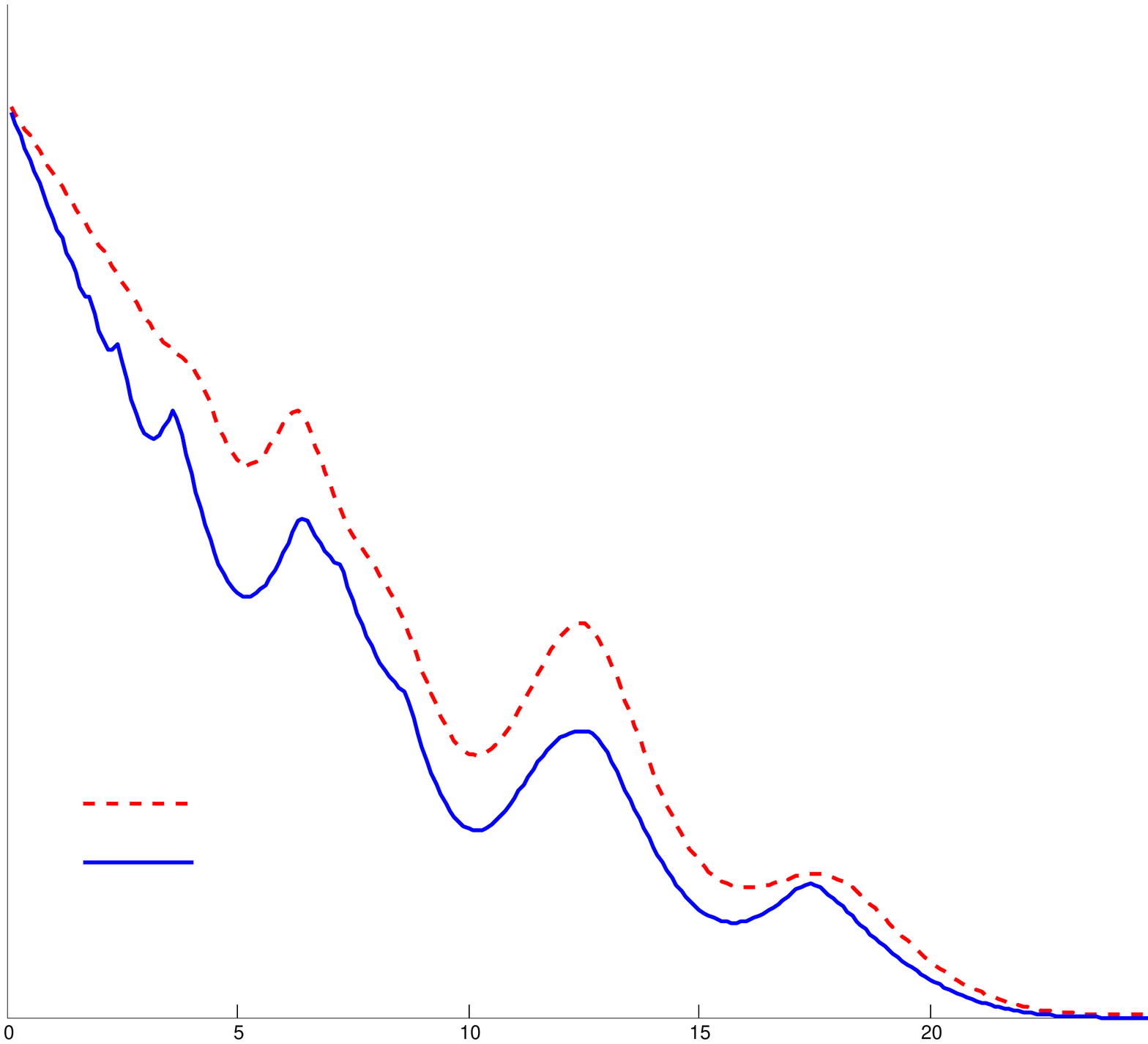}}; 
\node at (2,2)   {\includegraphics[trim=4cm 9.5cm 4cm 9cm, clip=true, scale=0.3,frame]{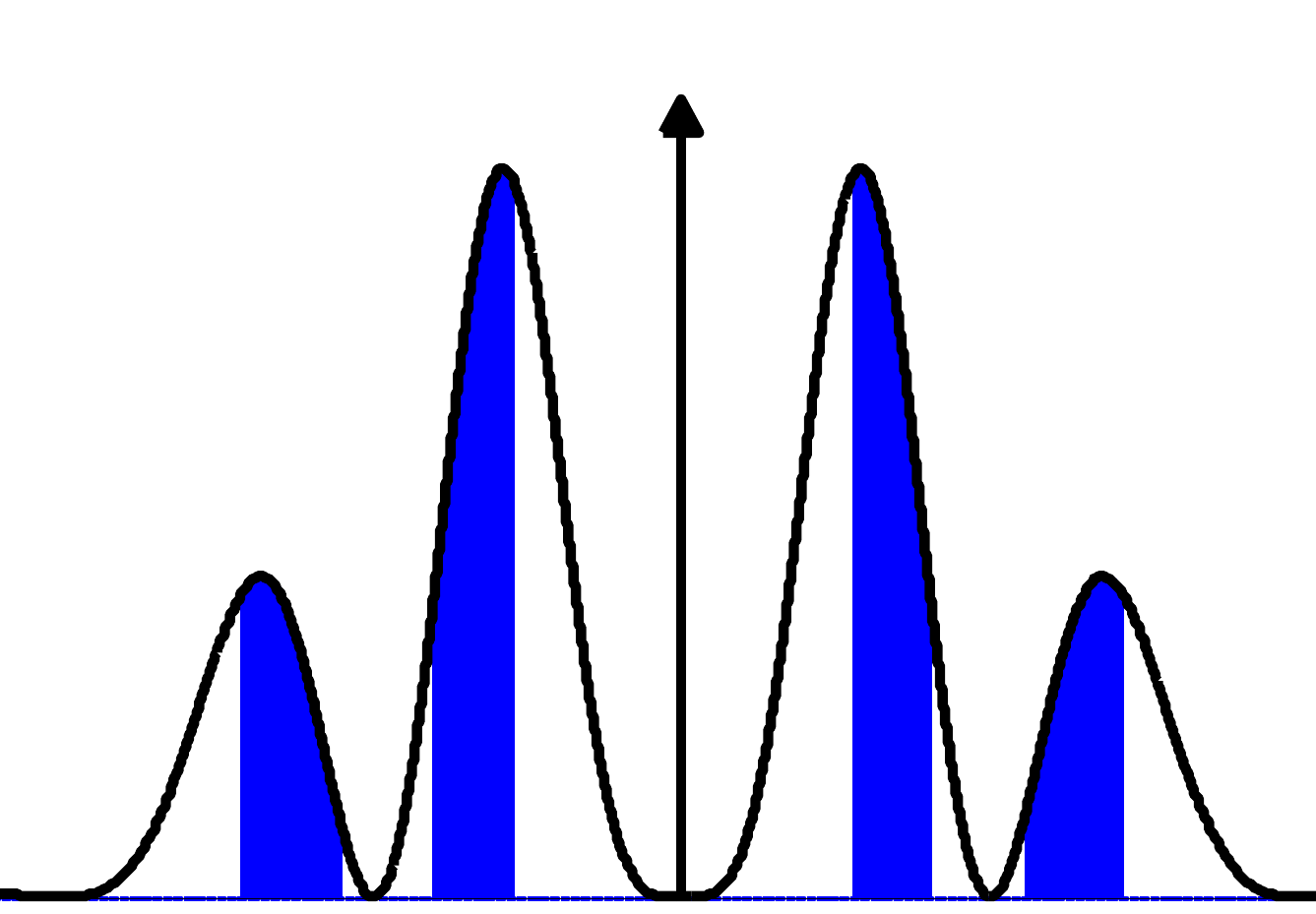}};
\node at (3.3,3) {\small $\frac{S_X(f)}{S_{X+\eta}(f)}$};
\draw[line width = 1pt] (-3.8,3.4) node [left] {$\sigma_X^2$}-- (-3.6,3.4);
\draw[->,line width = 2pt] (-3.7,-3.8) --  (4.1,-3.8) node[right] {\small $f_s$};
\draw[->,line width = 2pt] (-3.75,-3.8) -- node[above,rotate=90] {\small $MSE$} (-3.75,3.8);
\node at (-1.5,1.2) {\color{red} \small $\left|H(f)\right|\equiv 1$};
\node at (-2,-1) {\color{blue} \small $H^\star(f)$};
\fill[white] (-3.5,-2.7) rectangle (-2,-1.5);
\end{tikzpicture}
\caption{\label{fig:mmse_fig} The MMSE as a function of the sampling frequency $f_s$ in single branch sampling, with an optimal pre-sampling filter and an all-pass filter. The function $S_X^2(f)/S_{X+\eta}(f)$ is given in the small frame. For the case $f_s = f_{Nyq}/3$, the support of the optimal pre-sampling filter associated with this spectral density and sampling frequency is given by the shaded area.}
\end{center}
\end{figure}

It also follows from Theorem~\ref{thm:mmse_opt_single} and Proposition~\ref{prop:aliasing_free} that a lower bound on $\mmse^\star_{X|Y}(f_s)$ can be obtained by integrating over a set of Lebesgue measure $f_s$ with maximal $\frac{S_X^2(f)} {S_{X+\eta}(f)}$ (that is, without the aliasing-free property). This leads to 
\begin{equation} \label{eq:mmse_lower_bound}
\mmse_{X|Y}^\star(f_s) \geq \sigma_X^2 - \sup_{\mu(F)\leq f_s} \int_{F} \frac{S_X^2(f)}{S_{X+\eta}} df
\end{equation}
(the supremum is taken over all measurable subsets of $\mathbb R$ with Lebesgue measure not exceeding $f_s$). A special case in which the bound \eqref{eq:mmse_lower_bound} is achieved is described in the following example. 

\begin{example}[unimodal PSD] \label{ex:unimodal}
In the special case where the function $\frac{S_X^2(f)}{S_{X+\eta}(f)}$ is unimodal in the sense that it is non-increasing for $f>0$, the associated maximal aliasing-free set is the interval $(-f_s/2,f_s/2)$ and the optimal pre-sampling filter is a lowpass with cutoff frequency $f_s/2$. Theorem \ref{thm:mmse_opt_single} then implies that
\begin{equation} \label{eq:mmse_unimodal}
\mmse_{X|Y}^\star(f_s) = \sigma_X^2 - \int_{-\frac{f_s}{2}}^\frac{f_s}{2} \frac{S_X^2(f)}{S_{X+\eta}(f)} df.
\end{equation}
Since $S_X(f)$ is symmetric and non-increasing for $f>0$, $\mmse_{X|Y}^\star(f_s)$ in \eqref{eq:mmse_average_equiv} achieves the bound \eqref{eq:mmse_lower_bound}. 
\end{example}

In contrast to the case described in Example~\ref{ex:unimodal}, the bound in \eqref{eq:mmse_lower_bound} cannot be achieved by a single sampling branch in general. It can, however, be approached by increasing the number of sampling branches, as will be discussed in the following two subsections.

\subsection{Multi-branch sampling}
We now extend Propositions~\ref{prop:mmse_single} and \ref{thm:mmse_opt_single} to the case of multi-branch sampling. The system model is given by Fig.~\ref{fig:mmse_model} with the sampler of Fig.~\ref{fig:sampling_scheme}(b). 

\begin{theorem}[MMSE in multi-branch sampling]\label{thm:mmse_multi} For each $p=1,\ldots,P$, let $Z_p(\cdot)$ be the process obtained by passing a Gaussian stationary source $X(\cdot)$ corrupted by a stationary Gaussian noise $\eta(\cdot)$ through an LTI system $H_p(f)$. Let
$Y_p[\cdot]$, be the samples of the process $Z_p(\cdot)$ at frequency $f_s/P$, namely
\[ 
Y_p[n]=Z(nP/f_s)=h_p(\cdot)\star \left(X(\cdot)+\eta(\cdot)\right)(nP/f_s),\quad n\in \mathbb Z.
\]
The MMSE in estimating $X(\cdot)$ from the samples $\mathbf Y[\cdot]=\left(Y_1[\cdot],\ldots,Y_P[\cdot]\right)$, is given by
 
\begin{align}
\label{eq:mmse_multi}
\mmse_{X|\mathbf Y}(f_s)& = \sigma_{X}^{2}-\int_{-\frac{f_{s}}{2}}^{\frac{f_{s}}{2}}\mathrm{Tr}\left( \widetilde{\mathbf S}_{X|\mathbf Y}(f)\right)df.
\end{align}
Here $\sigma_X^2 = \mathbb E \left(X(t)\right)^2$, $\widetilde{\mathbf S}_{X|\mathbf Y}(f)$ is the $P\times P$ matrix defined by
\begin{equation} \label{eq:J_mat_def}
 \widetilde{\mathbf S}_{X|Y}(f)\triangleq \tilde{\mathbf  S}_{\mathbf Y}^{-\frac{1}{2}*}(f){\mathbf K}(f)\tilde{\mathbf  S}_{\mathbf Y}^{-\frac{1}{2}}(f),
\end{equation}
where the matrices $\tilde{\mathbf S}_{\mathbf Y}(f),\mathbf K(f) \in \mathbb C^{P\times P}$  are given by
\begin{align*}
\left(\tilde{\mathbf  S}_{\mathbf Y}(f)\right)_{i,j} & =\sum_{k\in\mathbb{Z}}\left\{ S_{X+\eta}H_{i}^*H_{j}\right\} \left(f-f_{s}k\right),
\end{align*}
 and 
\[
\left({\mathbf K}(f)\right)_{i,j}=\sum_{k\in\mathbb{Z}}\left\{ S_{X}^{2}H_{i}^*H_{j}\right\} \left(f-f_{s}k\right).
\]
\end{theorem}

\begin{proof}
See Appendix \ref{sec:proof_mmse}.
\end{proof}

\subsection{Optimal pre-sampling filter bank}

It follows from Theorem~\ref{thm:mmse_multi} that minimizing the MMSE in multi-branch sampling is equivalent to maximizing the sum of the eigenvalues of $\widetilde{\mathbf S}_{X| \mathbf Y}(f)$ of \eqref{eq:J_mat_def} for every $f \in (-f_s/2,f_s/2)$. A characterization of the set of pre-sampling filter that maximizes this sum is given in \cite[Prop. 3]{ShannonMeetsNyquist}. We will provide here a different proof, which will be useful in proving a similar result for the general combined sampling and source coding problem.

\begin{theorem}
\label{thm:mmse_opt_filters_bank}
The optimal pre-sampling filters $H^\star_1(f),\ldots,H^\star_P(f)$ that maximize $\mathrm {Tr}\,\widetilde{\mathbf S}_{X|\mathbf Y}(f)$ of \eqref{eq:J_mat_def} and minimize $\mmse_{X|\mathbf Y}(f_s)$ of \eqref{eq:mmse_multi} are given by 
\begin{equation}
H_{p}^{\star}(f)=\mathbf 1_{F_p^\star}(f) \triangleq  \begin{cases}
1 & f\in F_p^\star,\\
0 & f\notin F_p^\star,
\end{cases}\label{eq:filterbanks_discrete_Hdef},\quad p=1,\ldots,P,
\end{equation}
where the sets $F_1^\star,\ldots,F_P^\star \subset \mathbb R$ satisfy:
\begin{enumerate}
\item[(i)] $F^\star_p \in AF(f_s/P)$ for all $p=1,\ldots,P$.
\item[(ii)] For $p=1$,
\[
\int_{F_1^\star} \frac{S_X^2(f)}{S_{X+\eta}(f)} df = \int_{-\frac{f_s}{2}}^\frac{f_s}{2} J_1^\star(f) df, 
\]
where 
\[
J^\star_1(f) \triangleq \sup_{k\in \mathbb Z} \frac{S_X^2(f-kf_s/P)}{S_{X+\eta}(f-kf_s/P)},
\]
and for $p=2,\ldots,P$,
\[
\int_{F_p^\star} \frac{S_X^2(f)}{S_{X+\eta}(f)} df = \int_{-\frac{f_s}{2}}^\frac{f_s}{2} J_p^\star(f) df, 
\]
where 
\[
J^\star_p(f) \triangleq \sup_{k\in \mathbb Z} \frac{S_X^2(f-kf_s/P)}{S_{X+\eta}(f-kf_s/P)}\mathbf 1_{\mathbb R\setminus \left\{F^\star_1\cup\cdots \cup F_{p-1}^\star \right\}}.
\]
\end{enumerate}
The resulting MMSE is 
\begin{align} \label{eq:mmse_opt_multi_def}
\mmse_{X|\mathbf Y}^\star\left(f_s\right) & \triangleq \sigma_X^2 -\sum_{p=1}^{P}\int_{F^\star_p}\frac{S_X^2(f)}{S_{X+\eta}(f)}df \\ & =  \sigma_{X}^{2} -\sum_{p=1}^{P}\int_{-\frac{f_s}{2}}^\frac{f_s}{2} J_p^\star(f) df \nonumber .
\end{align}
\end{theorem}
\begin{proof}
See Appendix~\ref{sec:proof_mmse}.\\
\end{proof}

\subsubsection*{Remarks}
\begin{enumerate}
\item[(i)] The proof implies an even stronger statement than Theorem~\ref{thm:mmse_opt_filters_bank}: the filters $H_1^\star(f),\ldots,H_P^\star(f)$ yield a set of eigenvalues of $\widetilde{\mathbf S}_{X| \mathbf Y}(f)$ which are uniformly maximal, in the sense that the $i^\mathrm{th}$ eigenvalue of $\widetilde{\mathbf S}_{X| \mathbf Y}(f)$ is always smaller than the $i^\mathrm{th}$ eigenvalue of $\widetilde{\mathbf S}^\star_{X| \mathbf Y}(f)$. This is an important fact that will be used in proving Theorem~\ref{thm:opt_filters_bank} below. 
\item[(ii)] As in the single-branch case in Theorem~\ref{thm:mmse_opt_single}, the filters $H^\star_{1}(f),\ldots,H^\star_{P}(f)$ are specified only in terms of their support and
in \eqref{eq:filterbanks_discrete_Hdef} we can replace $1$ by any non-zero value which may vary with $p$ and $f$. 
\item[(iii)]
Condition $(ii)$ for the sets $F_1^\star,\ldots,F_P^\star$ can be relaxed in the following sense: if $F_1^\star,\ldots,F_P^\star$ satisfy condition $(i)$ and $(ii)$, then $\mmse^\star_{X|\mathbf Y}(f_s)$ is achieved by any pre-sampling filters defined as the indicators of the sets $F_1',\ldots,F_P'$ in $AF(f_s/P)$ for which 
\[
\sum_{p=1}^P \int_{F'_p} \frac{S_X^2(f)}{S_{X+\eta}(f)} df = \sum_{p=1}^P \int_{F_p^\star} \frac{S_X^2(f)}{S_{X+\eta}(f)} df.
\]
\item[(iv)]
One possible construction for $F_1^\star,\ldots, F_P^\star$ is as follows: over all frequencies $-\frac{f_s}{2P}\leq f < \frac{f_s}{2P}$, for each $f$ denote by $f^\star_1(f),\ldots,f^\star_P(f)$ the $P$ frequencies that correspond to the largest values among $\left\{\frac{S_X^2(f-f_sk)}{S_{X+\eta}(f-f_sk)},\,k\in \mathbb Z\right\}$. Then assign each $f^\star_p(f)$ to $F^\star_p$. Under this construction, the set $F^\star_p$ can be seen as the $p^\textrm{th}$ \emph{maximal aliasing free set} with respect to $f_s/P$ and $\frac{S_{X}^{2}(f)}{S_{X+\eta}(f)}$. This is the construction that was used in Fig.~\ref{fig:aliasing_free_various}.
\end{enumerate} 

Fig.~\ref{fig:mmse_differentsPs} illustrates $\mmse^\star_{X|\mathbf Y} (f_s)$ as a function of $f_s$ for a specific PSD and $P=1,2$ and $3$. As this figure shows, increasing the number of sampling branches does not necessarily decrease $\mmse_{X|\mathbf Y}^\star(f_s)$ and may even increase it for some\footnote{Note that if $P_1$ and $P_2$ are co-primes, then even if $P_2>P_1$, there is no way to choose the pre-sampling filters for the system with $P_2$ branches and sampling frequency $f_s/P_2$ at each branch to produce the same output as the system with $P_1$ sampling branches and sampling frequency $f_s/P_1$.}  $f_s$. However, we will see below that $\mmse_{X|\mathbf Y}^\star(f_s)$ converges to a fixed number as $P$ increases.

\begin{figure}
\begin{center}

\begin{tikzpicture}
\node at (0,0) {\includegraphics[trim=2.2cm 8.8cm 7cm 5cm, clip=true, scale=0.6]{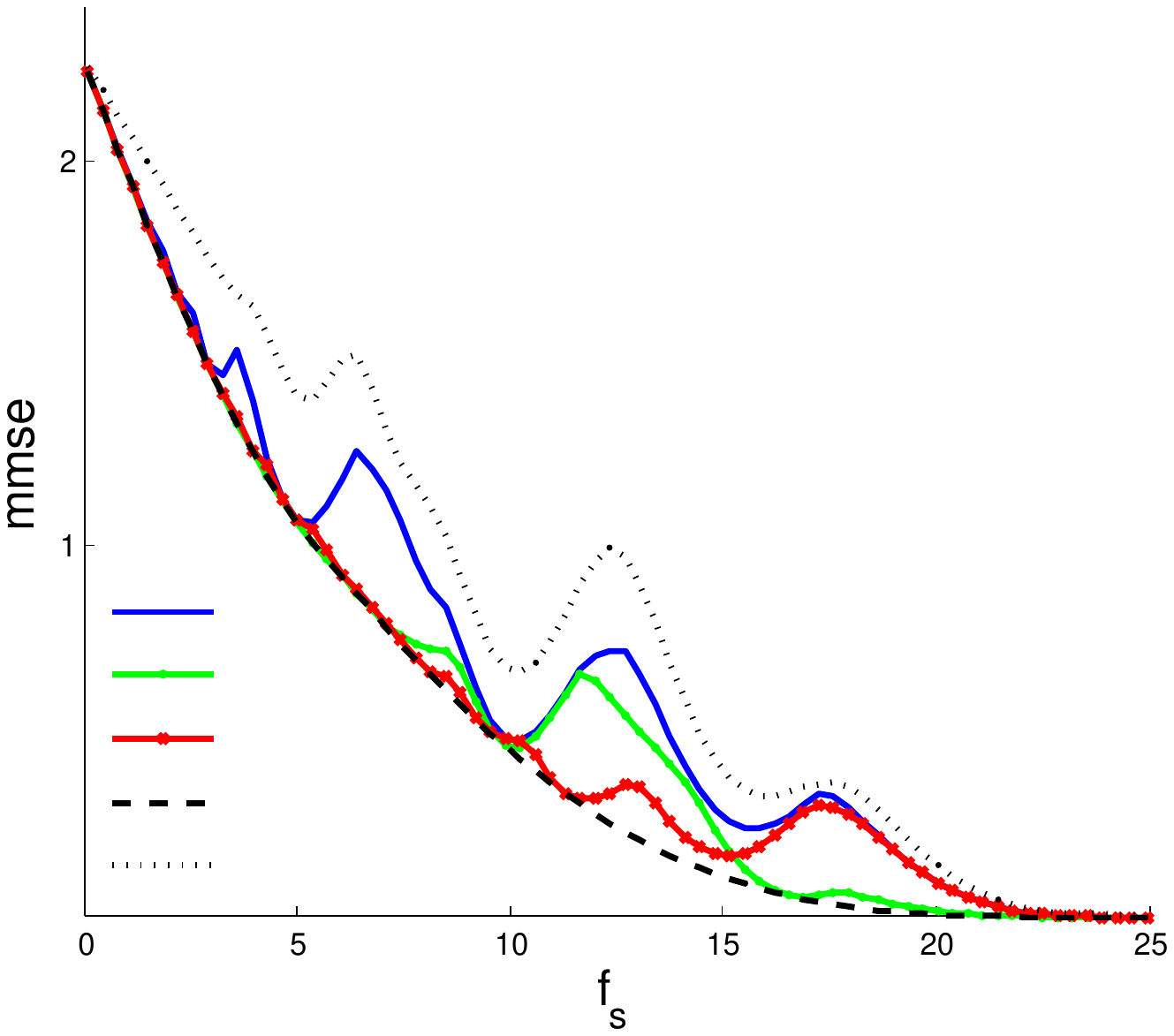} };
\node at (2,1) { \includegraphics[trim=2.5cm 8.5cm 1.7cm 7cm, clip=true, scale=0.29,frame]{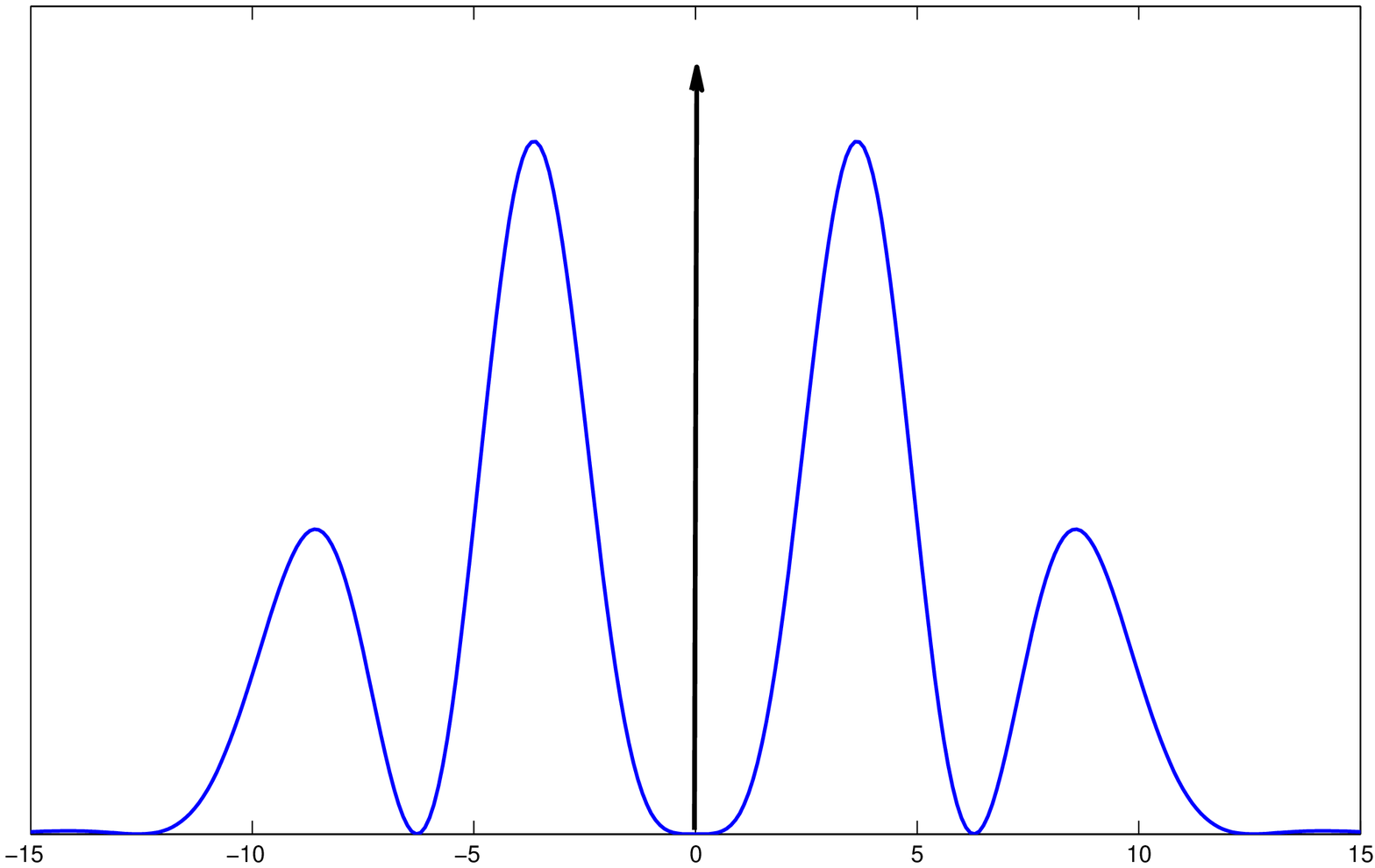} };

\node at (-2,-1.8) {\small $P=1$};
\node at (-2,-2.3) {\small $P=2$};
\node at (-2,-2.7) {\small $P=3$};
\node at (-2,-3.1) {\small $P\rightarrow \infty$};
\node at (-1.8,-3.6) {\small $H(f) \equiv 1$};

\draw[->,line width = 2] (-3.4,-3.95) -- (4,-3.95) node[right] {$f_s$};
\draw[->,line width = 2] (-3.4,-3.95) -- node[above, rotate = 90] {MSE} (-3.4,3);

\node at (3.7,2) {\small $S_X(f)$};
\draw  (-3.5,2.1) node[left] {$\sigma_X^2$} -- (-3.3,2.1);

\draw[dashed]  (3.2,-4) node[below] {$f_{Nyq}$} -- (3.2,-3);
\end{tikzpicture}

\caption{\label{fig:mmse_differentsPs}
MMSE under multi-branch sampling and optimal pre-sampling filter-bank, for $P=1,2$ and $3$ and $P$ large enough such that the bound \eqref{eq:mmse_lower_bound} is attained. The upper dashed line corresponds to $P=1$ and no pre-sampling filtering. The PSD $S_X(f)$ is given in the small frame, where we assumed $S_\eta(f) \equiv 0$.}
\end{center}
\end{figure}

\subsection{Increasing the number of sampling branches} \label{subsec: optimal_sampling}
The set $F^\star$ defined in Theorem~\ref{thm:mmse_opt_single} to describe $ \mmse_{X|Y}^\star(f_s)$ was obtained by imposing two constraints: (1) a measure constraint $\mu(F^\star)\leq f_s$, which is associated only with the sampling frequency, and (2) an aliasing-free constraint, imposed by the sampling structure. Theorem~\ref{thm:mmse_opt_filters_bank} says that in the case of multi-branch sampling, the aliasing-free constraint can be relaxed to $F^\star = \bigcup_{p=1}^P F^\star_p$, where now we only require that each $F^\star_p$ is aliasing-free with respect to $f_s/P$. This means that $F^\star$ need not be aliasing-free but its Lebesgue measure must still not exceed $f_s$. This implies that the lower bound \eqref{eq:mmse_lower_bound} still holds with multiple sampling branches. 
Increasing the number of branches $P$ allows more freedom in choosing an optimal frequency set $F^\star = \bigcup_{p=1}^P F^\star_p$, which eventually converges to a set $\mathcal F^\star$ that achieves the RHS of \eqref{eq:mmse_lower_bound} as shown in Fig.~\ref{fig:aliasing_free_various}. 
This means that the RHS of \eqref{eq:mmse_lower_bound} is achieved with an arbitrarily small gap if we increase the number of sampling branches, as stated in the next theorem.
\begin{theorem}\label{thm:mmse_Landau} For any $f_s>0$ and $\epsilon>0$, there exists $P\in \mathbb N$ and a set of LTI filters $H^\star_1(f),\ldots,H^\star_P(f)$ such that 
\[
\mmse_{X|\mathbf Y}^\star(f_s)-\epsilon <\sigma_X^2-\sup_{\mu(F)\leq f_s} \int_{F} \frac{S_X^2(f)}{S_{X+\eta}(f)} df.
\]
\end{theorem}
\begin{proof}
Since any interval of length $f_s$ is in $AF(f_s)$, it is enough to show that the set $\mathcal F^\star$ of measure $f_s$ that satisfies
\[
\int_{{\mathcal F}^\star}\frac{S_X^2(f)}{S_{X+\eta}(f)} df = \sup_{\mu( F)\leq f_s} \int_{F} \frac{S_X^2(f)}{S_{X+\eta}(f)} df,
\]
can be approximated by $P$ intervals of measure $f_s/P$ (The set $\mathcal F^\star$ corresponds to the frequencies below the gray area in the case $P\rightarrow \infty$ in Fig.~\ref{fig:aliasing_free_various}). By the regularity of the Lebesgue measure, a tight cover of $\mathcal F^\star$ by a finite number of intervals is possible. These intervals may be split and made arbitrarily small so that after a finite number of steps we can eventually have all of them at approximately the same length. Denote the resulting intervals $I_1,\ldots,I_P$. For $p=1,\ldots,P$, define the $p^\textrm{th}$ filter $H^\star_p(f)$ as the indicator function of the $p^\textrm{th}$ interval: $H^\star_p(f) = \mathbf 1_{I_p}(f) \triangleq \mathbf 1\left(f\in I_p \right)$. \par
By extending the argument above it is possible to show that we can pick each one of the sets $F^\star_p$ to be symmetric about the $y$ axis, so that the corresponding filter $H_p^\star(f)$ has a real impulse response. In fact, the construction described in Remark (iii) of Theorem~\ref{thm:mmse_opt_filters_bank} yields a symmetric maximal aliasing free set.
\end{proof}

Theorem~\ref{thm:mmse_Landau} implies that $\mmse_{X|\mathbf Y}^\star(f_s)$ converges to  
\[
\sigma_X^2-\sup_{\mu(F)\leq f_s} \int_{F} \frac{S_X^2(f)}{S_{X+\eta}(f)} df,
\]
as the number of sampling branches $P$ goes to infinity.\\

\subsubsection*{Remark} If we take the noise process $\eta(\cdot)$ to be zero in \eqref{eq:mmse_lower_bound}, we get
\begin{align}
\mmse_{X|\mathbf Y}^\star(f_s) & \geq \sigma_X^2-\sup_{\mu({F})\leq f_s}\int_{{F}} S_X(f)df 
\nonumber \\ & = \inf_{\mu({F})\leq f_s}\int_{\mathbb R \setminus F} S_X(f)df. \label{eq:mmse_bound_landau}
\end{align}
This shows that perfect reconstruction (in $\mathbf L_2$ norm) of $X(\cdot)$ is not possible unless the support of $S_X(f)$ is contained in a set of Lebesgue measure not exceeding $f_s$, a fact which agrees with the well-known condition of Landau for a stable sampling set of an analytic function \cite{Landau1967}. See also \cite{nitzan2012revisiting} for a recent and a much simpler proof of Landau's theorem.\\

Theorem \ref{thm:mmse_Landau} shows that uniform multi-branch sampling can be used to sample at the Landau rate and achieves an arbitrarily small MMSE by taking enough sampling branches and a careful selection of the pre-sampling filters. Fig.~\ref{fig:mmse_differentsPs} shows the optimal MMSE under uniform multi-branch sampling as a function of the sampling frequency for a specific PSD, which corresponds to the case $P\rightarrow \infty$.

\section{Indirect Source Coding\label{sec:Indirect-Source-Coding}}
The main goal of this section is to extend the indirect source coding problem solved in \cite{1057738}, in which the source and the observation are two jointly stationary Gaussian processes, to the case where the source and the observation are two jointly Gaussian vector-valued processes. By doing so we introduce concepts and notions which will be useful in the rest of the paper. We begin by reviewing the problem of indirect source coding.\\

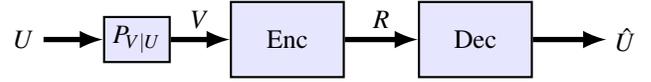
\begin{figure}
\begin{center}
\begin{tikzpicture}[node distance=2cm,auto,>=latex]
  \node at (0,0) (source) {$U$};
  \node [sint] (channel) [right of = source, node distance = 1.5cm] {$P_{V|U}$};
    \node [int] (enc) [right of = channel]{$\mathrm{Enc}$};
    \node [int] (dec) [right of=enc, node distance = 2.5cm] {$\mathrm{Dec}$};
    \node [right] (dest) [right of=dec]{$\hat{U}$};
   \draw[->,line width=2pt] (enc) --node[above] {$R$} (dec) ;
   \draw[->,line width=2pt] (dec) -- (dest);
   \draw[->,line width=2pt] (channel) -- node[above] {$V$} (enc);
      \draw[->,line width=2pt] (source) -- (channel);
 \end{tikzpicture}
\end{center}
\caption{\label{fig:general_indirect_RD}Indirect source coding model}
\end{figure}

In an indirect source coding problem, the encoder needs to describe the source by observing a different process, statistically related to the source, rather than the source itself \cite[Sec. 4.5.4]{berger1971rate}. A general model for indirect source coding is described by the diagram in Fig. \ref{fig:general_indirect_RD}, where $P_{V|U}$ represents a general conditional distribution, or a `channel', between the source $U$ and the observable process $V$. It follows from \cite{1057738} that the minimal averaged distortion attainable between  a stationary source process $U$ and its reconstruction using any code of rate $R$ bits per source symbol applied to $V$, is described by the indirect distortion- rate function (iDRF) of $U$ given $V$. This function is defined as the minimal expected averaged distortion between any pair of stationary processes $U$ and $\hat{U}$, such that the mutual information rate between $V$ and $\hat{U}$ is limited to $R$ bits per source symbol\footnote{This is in agreement with the notation introduced in Section~\ref{sec:problem_statement}, i.e., $R$ is measured by \emph{bits per time unit} when the source is in continuous-time .} \cite{gray2011entropy}. \\

We will begin by reviewing basic properties of the iDRF for the special case of jointly Gaussian source and observations. 

\subsection{Quadratic Gaussian indirect distortion-rate}
The general quadratic Gaussian settings of the iDRF corresponds to the case where the source $U$ and the observable process $V$ are jointly Gaussian. The definition of the iDRF in this case is similar to the definition of the function $D(f_s,R)$ in Section~\ref{sec:problem_statement}, where the minimization in \eqref{eq:d_T_def} is over all mappings from $U$ to $\hat{V}$ such that the mutual information rate between $U$ and $V$ does not exceed $R$ bits per source symbol. In the discrete-time version of this problem, we replace the distortion in \eqref{eq:dist_def} with the distortion 
\begin{equation}
 d_N\left(x[\cdot],\hat{x}[\cdot] \right) \triangleq  \frac{1}{2N+1} \sum_{n=-N}^N \left( x[n] - \hat{x}[n] \right)^2, \label{eq:dist_disc}
\end{equation}

In order to get some insight into the nature of indirect source coding in Gaussian settings and the corresponding iDRF, it is instructive to start with a simple example.

\begin{example} \label{ex:simple_scalar_case}
Consider two correlated sequences of i.i.d zero mean and jointly Gaussian random variables $U$ and $V$ with variances $C_U$, $C_{V}$, and covariance $C_{UV}$. As shown in \cite{1054469}, the iDRF function of the source sequence $U$ given the observed sequence $V$, under a quadratic distortion measure, is given by
%
\begin{align}
D_{U|V}\left(R\right) & =  C_{U}-\left(1-2^{-2R}\right)\frac{C_{UV}^{2}}{C_{V}}\label{eq:D_simple}\\
 & = \mmse_{U|V}+2^{-2R}C_{U|V},\nonumber 
\end{align}
where $\mmse_{U|V}=C_{U}-C_{U|V}$ is the MMSE in estimating $U$ from $V$ and $C_{U|V}\triangleq \frac{C_{UV}^{2}}{C_{V}}$ is the variance of the estimator $\mathbb E\left[U|V\right]$.
The equivalent rate-distortion function is 
\begin{equation}
R_{U|V}\left(D\right)=\begin{cases}
\frac{1}{2}\log\frac{C_{U|V}}{D-\mmse_{U|V}} & C_{U}> D>\mmse_{U|V},\\
0 & D\geq C_{U}.
\end{cases}\label{eq:RD_simple}
\end{equation}

Equation \eqref{eq:D_simple} can be intuitively interpreted as an extension of the regular DRF of a Gaussian i.i.d source $D_U(R)=2^{-2R}C_U$ to the case where the information about the source is not entirely available at the encoder. Since $C_{U|V}\leq C_U$, the slope of $D_{U|V}(R)$ is more moderate than that of $D_U(R)$, which confirms the intuition that an increment in the bit-rate when describing noisy measurements is less effective in reducing distortion as the intensity of the noise increases. \par
If we denote $\theta = D-\mmse_{U|V}$, then \eqref{eq:D_simple} and \eqref{eq:RD_simple} are equivalent to 
\begin{subequations}
\label{eq:example_iid_parametric}
\begin{align}\label{eq:example_iid_D}
D(\theta)& =\mmse_{U|V}+ \min\left\{C_{U|V},\theta \right\} \\
& = C_U-\left[C_{U|V}-\theta \right]^+, \nonumber
\end{align}
where 
\begin{equation} 
R (\theta)=\frac{1}{2} \log^+ \left( \frac{C_{U|V}}{\theta}\right), \label{eq:example_iid_R}
\end{equation}
\end{subequations}
i.e, we express the iDRF $D_{U|V}(R)$, or the equivalent rate-distortion $R_{U|V}(D)$, through a joint dependency of $D$ and $R$ on the parameter $\theta$.\par
The relation between $V$ and $\hat{V}$ under optimal quantization can also be described by the `backward' Gaussian channel
\begin{equation} \label{eq:backward_channel}
 \frac{C_{VU}}{C_V}V= \hat{V}+\xi,
\end{equation}
where $\xi$ is a zero-mean Gaussian random variable independent of $\hat{V}$ with variance $\min\left\{C_{U|V},\theta\right\}$. The random variable $\xi$ can be understood as the part of the observable process that is lost due to lossy compression discounted by the factor $C_{VU}/C_V$. Since $\frac{C_{VU}}{C_V} V=\mathbb E\left[U|V\right]$, it also suggests that an optimal source code can be achieved by two separate steps:
\begin{enumerate}
\item[(i)]
Obtain an MMSE estimate of the source $U$ given the observed variable $V$.
\item[(ii)]
Apply an optimal direct source code to the MMSE estimator $\mathbb E\left[U|V \right]$.
\end{enumerate}
\end{example}
Although in this example the parametric representation \eqref{eq:example_iid_parametric} is redundant, we shall see below that this representation, the backward Gaussian channel \eqref{eq:backward_channel} and the decomposition of distortion into an MMSE part plus the regular DRF of the estimator are repeating motifs in quadratic Gaussian indirect source coding problems. \\

\begin{figure}

\begin{center}
\begin{tikzpicture}[scale=1]
 \fill[fill=red!50, pattern=north west lines, pattern color=red] (-4,0) -- (-3.5,0.22)--plot[domain=-3.5:3.5, samples=100] (\x, {(0.7+3*cos(24*3.14*\x)^2)*exp(-\x*\x/10)}) -- (3.5,0.22)--(4,0);	

 \fill[fill=blue!50] (-4,0) -- (-3.5,0.13)--plot[domain=-3.5:3.5, samples=100] (\x, {(0.6+2.7*cos(24*3.14*\x)^2)*exp(-\x*\x/8)}) -- (3.5,0.13)--(4,0);	

 \fill[fill=yellow!50] (-2.6,1.3) -- plot[domain=-2.6:-1.8, samples=100] (\x, {(0.6+2.7*cos(24*3.14*\x)^2)*exp(-\x*\x/8)}) --(-1.8,1.3);	

 \fill[fill=yellow!50] (1.8,1.3) -- plot[domain=1.8:2.6, samples=100] (\x, {(0.6+2.7*cos(24*3.14*\x)^2)*exp(-\x*\x/8)}) --(2.6,1.3);	

 \fill[fill=yellow!50] (-0.75,1.3) -- plot[domain=-0.75:0.75, samples=100] (\x, {(0.6+2.7*cos(24*3.14*\x)^2)*exp(-\x*\x/8)}) --(0.75,1.3);	

 \draw (-4,0) -- (-3.5,0.13)--plot[domain=-3.5:3.5, samples=100] (\x, {(0.6+2.7*cos(24*3.14*\x)^2)*exp(-\x*\x/8)}) -- (3.5,0.13)--(4,0);	

 \draw (-4,0) -- (-3.5,0.22)--plot[domain=-3.5:3.5, samples=100] (\x, {(0.7+3*cos(24*3.14*\x)^2)*exp(-\x*\x/10)}) -- (3.5,0.22)--(4,0);	

\foreach \x/\xtext in {-3,-2,-1,0,1,2,3}
 \draw[shift={(\x,0)}] (0pt,2pt) -- (0pt,-2pt) node[below] {$\xtext$};
  \foreach \y/\ytext in {1/,2/,3/}
    \draw[shift={(0,\y)}] (2pt,0pt) -- (-2pt,0pt) node[left] {$\ytext$};


\draw [fill=red!50, line width=1pt, pattern=north west lines, pattern color=red] (1,3.4) rectangle  (1.3,3.7) node[left, xshift = 1.9cm, yshift = -0.2cm] {\scriptsize estimation error};

\draw [fill=blue!70, line width=1pt] (1,2.9) rectangle  (1.3,3.2) node[left, xshift = 2.7cm, yshift = -0.2cm, align = center] {\scriptsize lossy compression error};

\draw [fill=yellow!50, line width=1pt] (1,2.4) rectangle  (1.3,2.7) node[left, xshift=2.2cm, yshift = -0.2cm,align = left] {\scriptsize preserved spectrum};

\node at (0.15,1.45) {$\theta$};
\draw[dashed, line width=1pt] (-4,1.3) -- (4,1.3);
\draw[->,line width=1pt]  (-4,0)--(4,0) node[right] {$f$};
\draw[->,line width=1pt]  (0,0)--(0,4) node[right] {};
\node at (-2.9,1.7) [rotate=60] {\small $S_X(f)$};
\node at (-2.7,0.65) [rotate=60] {\small ${S}_{X|Z}(f)$};
\end{tikzpicture}
\end{center}
\vspace{-10pt}

\caption{ \label{fig:waterfilling}
Reverse waterfilling interpretation for \eqref{eq:DnT}: Water is being poured into the area bounded by the graph of $S_{X|Z}(f)$ up to level $\theta$. The rate is determined by integration over the preserved part through \eqref{eq:DnT_R}. The distortion in \eqref{eq:DnT_D} is the result of integrating over two parts: (i) $S_X(f)-S_{X|Z}(f)$, which results in the MMSE of estimating $X(\cdot)$ from $Z(\cdot)$, (ii) $\min\left\{S_{X|Z}(f),\theta\right\}$, which is the error due to lossy compression. 
}
\end{figure}
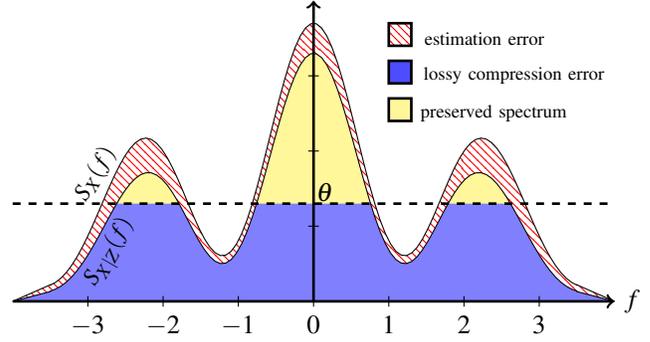

Next, we consider an indirect source coding problem in the more general case where the source and the observable process are jointly Gaussian and stationary. This problem can be obtained from our general model in Fig. \ref{fig:The-general-scheme} if we assume that the process $Z(\cdot)$ in Fig.~\ref{fig:sampling_scheme} can be recovered from its samples $\mathbf Y[\cdot]$ with zero error. In this case, the iDRF of $X(\cdot)$ given $\mathbf Y[\cdot]$ reduces to the iDRF of $X(\cdot)$ given $Z(\cdot)$, which we denote as $D_{X|Z}(R)$. An expression for $D_{X|Z}(R)$ was found by Dubroshin and Tsybakov in \cite{1057738}.
\begin{theorem}[Dobrushin and Tsybakov \cite{1057738}]
\label{thm:[Dobrushin-and-Tsybakov]}
Let~$X(\cdot)$~and $Z(\cdot)$ be two jointly stationary Gaussian stochastic processes with spectral densities $S_{X}(f)$,
$S_Z(f)$, and joint spectral density $S_{XZ}(f)$.
The indirect distortion-rate function of $X(\cdot)$ given $Z(\cdot)$ is
\begin{subequations}
\label{eq:DnT}
\begin{align}
\label{eq:DnT_R}
R\left(\theta\right)& = \frac{1}{2} \int_{-\infty}^\infty \log^+\left[ {S_{X|Z}(f)}/\theta \right] df, \\ \label{eq:DnT_D}
D_{X|Z}\left(\theta\right)&  = \mmse_{X|Z}+\int_{-\infty}^\infty \min\left\{S_{X|Z}(f),\theta \right\}df  \\
 & = \sigma_X^2-\int_{-\infty}^\infty\left[S_{X|Z}(f)-\theta\right]^+df,  \nonumber
\end{align}
\end{subequations}
where 
\begin{equation}
\label{eq:mmse_x_z}
S_{X|Z}(f)\triangleq \frac{\left|S_{XZ}(f) \right|^2}{S_Z(f)}=\frac{S_X^2(f)\left|H(f) \right|^2}{S_{X+\eta}(f)\left|H(f) \right|^2}
\end{equation}
is the spectral density of the MMSE estimator of $X(\cdot)$ from $Z(\cdot)$, $[x]^+=\max\left\{x,0\right\}$,  and 
\begin{align*}
\mmse_{X|Z}  = & \int_{-\infty}^\infty \mathbb E \left( X(t)-\mathbb E \left[X(t)|Z(\cdot)\right] \right)^2 dt \\
 = & \int_{-\infty}^\infty \left(S_X(f)-S_{X|Z}(f)\right)df
\end{align*}
is the MMSE.\\
\end{theorem}

\subsubsection*{Remarks}
\begin{enumerate}
\item[(i)] In \eqref{eq:mmse_x_z} and in similar expressions henceforth, we interpret fractions as zero if both numerator and denominator are zero, i.e. \eqref{eq:mmse_x_z} can be read as
\[
\frac{S_X^2(f)\left|H(f) \right|^2}{S_{X+\eta}(f)\left|H(f) \right|^2}=\frac{S_X^2(f)}{S_{X+\eta}(f)} \mathbf 1_{\textrm{supp}H}(f),
\]
where $\mathbf 1_{\textrm{supp}H}(f)$ is the indicator function of the support of $H(f)$.
\item[(ii)]
Expressions of the form \eqref{eq:DnT} are still correct if the spectral density $S_{X|Z}(f)$ includes Dirac delta functions. This is because the Lebesgue integral is not affected by infinite values on a set of measure zero. This is in accordance with the fact that periodic components can be determined for all times by specifying only their magnitude and phase, which requires zero information rate.
\item[(iii)] The discrete-time counterpart of \eqref{eq:DnT} is
\begin{subequations}
\label{eq:DnT_discrete}
\begin{align}
\label{eq:DnT_discrete_R}
R\left(\theta\right)& = \frac{1}{2}\int_{-\frac{1}{2}}^{\frac{1}{2}}\log^{+}\left[S_{X|Z}\left(e^{2\pi i\phi}\right)\theta^{-1}\right]d\phi,
\end{align}
and
\begin{align}
\label{eq:DnT_discrete_D}
D\left(\theta\right) & = \mmse_{X|Z}+\int_{-\frac{1}{2}}^{\frac{1}{2}}\min\left\{ \theta,S_{X|Z}\left(e^{2\pi i\phi}\right)\right\} d\phi\nonumber \\
 & =  \sigma_{X}^{2}-\int_{-\frac{1}{2}}^{\frac{1}{2}}\left[S_{X|Z}\left(e^{2\pi i\phi}\right)-\theta\right]^{+}d\phi,
\end{align}
\end{subequations}
\end{enumerate}
where the distortion between $X[\cdot]$ and its reconstruction sequence $\hat{X}[\cdot]$ is defined by the limit over \label{eq:dist_disc}.

Equation \eqref{eq:DnT} defines the function $D_{X|Z}(R)$ through a joint dependency of $D_{X|Z}$ and $R$ on the parameter $\theta$. The distortion is the sum of the MMSE in estimating $X(\cdot)$ from $Z(\cdot)$, plus a second term which has a water-filling interpretation. This is illustrated in Fig. \ref{fig:waterfilling}. This expression generalizes the celebrated Shannon-Kolmogorov-Pinsker (SKP) \textit{reverse waterfilling} expression for a single stationary Gaussian source \cite{Pinsker1954,1056823,Berger1998}.\par
In analogy with \eqref{eq:backward_channel}, the solution \eqref{eq:DnT} implies the following backward
Gaussian channel to describe the relation between the observable process $Z(\cdot)$ to its quantized version $\hat{Z}(\cdot)$ under the optimal lossy compression scheme:
\begin{equation} \label{eq:backward_channel2}
\left\{ q_{X|Z}*Z\right\}(t)=\hat{Z}(t)+\xi(t),\quad t\in \mathbb R,
\end{equation}
where $\xi(\cdot)$ is a noise process independent of $\hat{Z}(\cdot)$ with spectral density $S_\xi(f)=\min\{S_{X|Z}(f),\theta\}$, and $q_{X|Z}(t)$ is the impulse response of the Wiener filter in estimating $X(\cdot)$ from $Z(\cdot)$ with corresponding frequency response $Q(f)=\frac{S_{XZ}^*(f)}{S_Z(f)}$. The spectral counterpart of \eqref{eq:backward_channel2} is
\[
S_{X|Z}(f)=S_{\hat{Z}}(f)+\min\left\{S_{X|Z}(f),\theta \right\}.
\]
This decomposition of $S_{X|Z}(f)$ is seen in Fig. \ref{fig:waterfilling}, where $S_{\hat{Z}}(f)$ corresponds to the preserved part of the spectrum, and $\min\left\{S_{X|Z}(f),\theta \right\}$ corresponds to the error due to lossy compression.

\subsection{Separation principle}
Example~\ref{ex:simple_scalar_case} and Theorem~\ref{thm:[Dobrushin-and-Tsybakov]} suggest a general structure for the solution of indirect source coding problems under quadratic distortion. The following proposition follows from the proof in \cite{1057738}, where a separate proof is given in \cite{1054469}. 
\begin{prop}
\label{prop:For-an-indirect} Let $\mathbf U$ and $\mathbf V$ be any pair of vector valued random  processes. The iDRF of $\mathbf U$ given $\mathbf V$ under quadratic distortion can be written as 
\begin{equation}
D_{\mathbf U|\mathbf V}\left(R \right)=\mmse_{\mathbf U|\mathbf V}+D_{\mathbb{E}\left[\mathbf U|\mathbf V\right]}\left(R\right),\label{eq:prop_separation}
\end{equation}
where $D_{\mathbb{E}\left[\mathbf U|\mathbf V\right]}(R)$ is the (direct) distortion-rate function of the estimator $\mathbb{E}\left[\mathbf U|\mathbf V\right]$. 
\end{prop}
This proposition is valid under both discrete and continuous time indices, and therefore the time index was suppressed. \\

We can now revisit Example~\ref{ex:simple_scalar_case} and Theorem \ref{thm:[Dobrushin-and-Tsybakov]}
to observe that both are consequences of Proposition \ref{prop:For-an-indirect}.\par
Going back to our general sampling model of Fig.~\ref{fig:The-general-scheme}, we can use Proposition~\ref{prop:For-an-indirect} to write
\[
D\left(f_s,R\right)=\mmse_{ X|\mathbf Y}(f_s)+D_{\tilde{ X}}(R),
\]
where the process $\tilde{X}(\cdot)$ is defined by
\[
\tilde{ X}\left(\cdot\right)\triangleq\left\{ \mathbb{E}\left[ X\left(t\right)|\mathbf Y\left[\cdot\right]\right],\, t\in\mathbb{R}\right\}. 
\]
This shows that the solution to our combined sampling and source coding problem is a sum of two terms. The first term is the MMSE in sub-Nyquist sampling already given by Theorem~\ref{thm:mmse_multi}. The second term is the DRF of the process $\tilde{X}(\cdot)$. Since $\tilde{X}(\cdot)$ is not stationary, we currently do not have the means to find its DRF. In fact, the DRF of $\tilde{X}(\cdot)$ is obtained as a special case of our main result in Section~\ref{sec:Frequency-Rate-Distortion-Functi} below.

\subsection{Vector-valued sources}
We now derive the counterpart of \eqref{eq:DnT_discrete} for vector-valued processes. We recall that for the Gaussian stationary source $X[\cdot]$, the counterpart
of the SKP reverse water-filling was given in \cite[Eq. (20) and (21)]{1628751}, as
\begin{subequations}
\label{eq:vector_RD}
\begin{align}
R\left(\theta\right)&=\sum_{i=1}^{M}\int_{-\frac{1}{2}}^{\frac{1}{2}}\frac{1}{2}\log^{+}\left[\lambda_{i}\left(\mathbf {S}_{\mathbf X}\right)\theta^{-1}\right]d\phi,\label{eq:vector_RD_R}\\
D_{\mathbf X}\left(\theta\right)&=\frac{1}{M}\sum_{i=1}^{M}\int_{-\frac{1}{2}}^{\frac{1}{2}}\min\left\{ \lambda_{i}\left(\mathbf {S}_{\mathbf X}\right),\theta\right\} d\phi,\label{eq:vector_RD_D}
\end{align}
\end{subequations}
where $\lambda_{1}\left(\mathbf {S}_{\mathbf X}\right),..., \lambda_{M}\left(\mathbf {S}_{\mathbf X}\right)$
are the eigenvalues of the spectral density matrix $\mathbf{S}_{\mathbf X}\left(e^{2\pi i \phi} \right)$ at frequency $\phi$. Combining \eqref{eq:vector_RD} with the separation principle allows us to extend Theorem~\ref{thm:[Dobrushin-and-Tsybakov]} to Gaussian vector processes. 
\begin{theorem}
\label{thm:indirect_vector_case}Let $\mathbf{X}\left[\cdot\right]=\left(X_1[\cdot],\ldots,X_M[\cdot] \right)$
be an $M$ dimensional vector-valued Gaussian stationary 
process, and let $\mathbf{Y}\left[\cdot\right]$ be another vector
valued process such that $\mathbf{X}\left[\cdot\right]$ and $\mathbf{Y}\left[\cdot\right]$
are jointly Gaussian and stationary. The indirect distortion-rate function
of $\mathbf{X}\left[\cdot\right]$ given $\mathbf{Y}\left[\cdot\right]$
under quadratic distortion is given by
\begin{align*}
R\left(\theta\right)=\sum_{i=1}^{M} \int_{-\frac{1}{2}}^{\frac{1}{2}}\frac{1}{2}\log^{+}\left[\lambda_i\left(\mathbf {S}_{\mathbf X|\mathbf Y}\right)\theta^{-1}\right]d\phi,
\end{align*}
\begin{align*}
D\left(\theta\right)  = & \mmse_{\mathbf{X}|\mathbf{Y}}
+\frac{1}{M}\sum_{i=1}^{M}\int_{-\frac{1}{2}}^{\frac{1}{2}}\min\left\{ \lambda_i\left(\mathbf {S}_{\mathbf X|\mathbf Y}\right),\theta\right\} d\phi\\
 = & \frac{1}{M}\int_{-\frac{1}{2}}^{\frac{1}{2}}\mathrm{Tr}~\mathbf S_{\mathbf{X}}\left(e^{2\pi i \phi}\right)d\phi 
 \\& -\frac{1}{M}\sum_{i=1}^{M}\int_{-\frac{1}{2}}^{\frac{1}{2}}\left[\lambda_i\left(\mathbf {S}_{\mathbf X|\mathbf Y}\right)-\theta\right]^{+}d\phi,
\end{align*}
where $\lambda_1\left(\mathbf {S}_{\mathbf X|\mathbf Y}\right),...,\lambda_M\left(\mathbf {S}_{\mathbf X|\mathbf Y}\right)$ are the eigenvalues of 
\[
\mathbf {S}_{\mathbf{X|Y}}\left(e^{2\pi i\phi}\right)\triangleq \mathbf {S}_{\mathbf{XY}}\left(e^{2\pi i\phi}\right) \mathbf{S}_{\mathbf{Y}}^{-1}\left(e^{2\pi i\phi}\right)\mathbf{S}_{\mathbf{XY}}^* \left(e^{2\pi i\phi}\right),
\]
which is the spectral density matrix of the MMSE estimator of $\mathbf{X}\left[\cdot\right]$ from $\mathbf{Y}\left[\cdot\right]$. Here $\mmse_{\mathbf{X}|\mathbf{Y}}$
is defined as
\begin{align*}
\mmse_{\mathbf{X}|\mathbf{Y}} & = \frac{1}{M}  \int_{-\frac{1}{2}}^{\frac{1}{2}}\mathrm{Tr}\left(\mathbf S_{\mathbf{X}}(e^{2\pi i\phi})-\mathbf S_{\mathbf{X|Y}}(e^{2\pi i\phi})\right)d\phi\\
 & = \frac{1}{M}  \sum_{i=1}^{M}\mmse_{{X_i}|\mathbf{Y}},
\end{align*}
where $\mmse_{{X_i}|\mathbf{Y}}$ is the MMSE in estimating the $i^\mathrm{th}$ coordinate of $X[\cdot]$ from $\mathbf Y[\cdot]$. 
\end{theorem}

\begin{proof}
This is an immediate consequence of Proposition~\ref{prop:For-an-indirect}
and Equations (\ref{eq:vector_RD_R}) and (\ref{eq:vector_RD_D}). 
\end{proof}

\subsection{Lower bound on the DRF of vector sources\label{subsec:lower_bound}} 

Throughout this subsection we suppress the time index and allow the processes considered to have either continuous or discrete time indices. We also use $R$ to represents either bits per time unit or bits per symbol, according to the time index.\\

In Section~\ref{sec:mmse} we exploited the fact that the polyphase components $X_\Delta[\cdot]$ defined in \eqref{eq:X_delta_def} and the process $\mathbf Y[\cdot]$ are jointly Gaussian to compute the MMSE of $X(\cdot)$ given $\mathbf Y[\cdot]$. This was possible since the overall MMSE is given by averaging the MMSE in estimating each one of the polyphase components $X_\Delta[\cdot]$ over $0\leq \Delta <1$, as expressed by \eqref{eq:mmse_average_equiv}. Unfortunately, the iDRF does not satisfy such an averaging property in general. Instead, we have the following proposition, which holds for any source distribution and distortion measure. 

\begin{prop}\label{prop:lower_bound} \emph{(average distortion bound)} Let $\mathbf{X}$ and $\mathbf{Y}$
be two vector-valued processes. The iDRF
of $\mathbf{X}$ given $\mathbf{Y}$, under a single-letter distortion measure $\tilde{d}$,
satisfies
\begin{equation}
D_{\mathbf{X}|\mathbf{Y}}\left(R\right)\geq\frac{1}{M}\sum_{i=1}^{M}D_{X_{i}|\mathbf{Y}}\left(R\right),\label{eq:lower_bound}
\end{equation}
where $\mathbf{X}=\left(X_{1},\ldots X_{M}\right)$, and $D_{\mathbf{X}|\mathbf{Y}}\left(R\right)$ is defined using the distortion measure 
$\tilde{d}\left(\mathbf{X},\mathbf{\hat{X}}\right)=\frac{1}{M}\sum_{i=1}^{M}\tilde{d}\left(X_{i},\hat{X}_{i}\right)$.
\end{prop}
\begin{proof}
The distortion at each coordinate $m=1,\ldots,M$ obtained by an optimal code of rate ${R}$ that was designed to minimize the distortion averaged over all coordinates cannot be smaller than the distortion of the optimal rate-$R$ code designed specifically for the $m^\textrm{th}$ coordinate.
\end{proof}

Note that $D_{\mathbf{X}|\mathbf{Y}}\left(\infty\right)=\mmse_{\mathbf{X|Y}}=\frac{1}{M}\sum_{i=1}^{M}D_{X_{i}|\mathbf{Y}}(\infty)$
and $D_{\mathbf{X}|\mathbf{Y}}(0)=\sigma_{\mathbf X}^2=\frac{1}{M}\sum_{i=1}^{M}D_{X_{i}|\mathbf{Y}}(0)$,
i.e. the bound is always tight for $R=0$ and $R\rightarrow\infty$. \\ 
 
The proof of Proposition~\ref{prop:lower_bound} implies that equality in the bound \eqref{eq:lower_bound} is achieved when the optimal indirect rate-$R$ code for the vector process $\mathbf X$ induces an indirect optimal rate-$R$ code for each one of the coordinates. This is the case if the $M$ optimal indirect rate-$R$ codes for each coordinate are all functions of a single indirect rate-$R$ code. Indeed, the bound is  tight when $R\rightarrow\infty$ since any code essentially describes $\mathbb E\left[\mathbf X(\cdot)|\mathbf Y[\cdot] \right]$, which is a sufficient statistic for the MMSE reconstruction problem. Another case of equality is described in the following example.

 \begin{example}[i.i.d vector source] \label{ex:AID2}
Let $\mathbf{U}=\left(U_{1},\ldots,U_{M}\right)$ and $\mathbf{V}=\left(V_{1},\ldots,V_{P}\right)$
be two i.i.d jointly Gaussian vector sources with covariance matrices
$\mathbf{C_{\mathbf{U}}}$, $\mathbf{C_{V}}$, and $\mathbf{C_{UV}}$.
In order to find the iDRF of $U_{m}$
given $\mathbf{V}$, for $m=1,\ldots,M$, we use Proposition \ref{prop:For-an-indirect} to obtain 
\[
D_{Um|\mathbf{V}}(R)=\mmse_{U_{m}|\mathbf{V}}+2^{-2R}C_{U_{m}|\mathbf{V}}.
\]
Here we relied on the fact that the distortion-rate function of the Gaussian random i.i.d source $\mathbb E[U_m|\mathbf V]$ is $2^{-2R}C_{U_m|\mathbf V}=2^{-2R} \mathbf C_{U_m\mathbf V} \mathbf \mathbf C^{-1}_{\mathbf V} C_{U_m\mathbf V}^*$. The bound \eqref{eq:lower_bound} implies 
\begin{align}
D_{\mathbf{U}|\mathbf{V}}(R) & \geq\frac{1}{M}\sum_{m=1}^{M}\left(\mmse_{U_{m}|\mathbf{V}}+2^{-2R}C_{U_{m}|\mathbf{V}}\right)\nonumber \\
 & =\mmse_{\mathbf{U}|\mathbf{V}}+\frac{1}{M}2^{-2R}\mathrm{Tr}\,\mathbf{C_{U|V}}\nonumber \\
 & =\mmse_{\mathbf{U}|\mathbf{V}}+\frac{1}{M}2^{-2R}\sum_{m=1}^{P\wedge M}\lambda_{i}\left(\mathbf{C_{U|V}}\right)\label{eq:lower_bound_vector_example},
\end{align}
where $P\wedge M=\min\left\{ P,M\right\} $ is the maximal rank of the matrix $\mathbf{C_{U|V}}$. \par
We now compare \eqref{eq:lower_bound_vector_example} to the true value of the iDRF of $\mathbf{U}$ given $\mathbf{V}$, which is obtained using Theorem~\ref{thm:indirect_vector_case},
\begin{align}
R(\theta) & =\frac{1}{2}\sum_{i=1}^{P\wedge M}\log^{+}\left(\lambda_{i}\left(\mathbf{C_{U|V}}\right)/\theta\right),\nonumber \\
D_{\mathbf{U|V}}(\theta) & =\mmse_{\mathbf{U|V}}+\frac{1}{M}\sum_{i=1}^{P\wedge M}\min\left\{ \lambda_{i}\left(\mathbf{C_{U|V}}\right),\theta\right\} \label{eq:lower_bound_vector_example2}.
\end{align} 
From \eqref{eq:lower_bound_vector_example}
and \eqref{eq:lower_bound_vector_example2} we conclude the following eigenvalue inequality, valid for any $R\geq 0$:
\begin{equation}
\sum_{i=1}^{P\wedge M}\min\left\{ \lambda_{i}\left(\mathbf{C_{U|V}}\right),\theta\right\} \geq2^{-2R}\sum_{m=1}^{P\wedge M}\lambda_i\left(\mathbf{C_{U|V}}\right)\label{eq:lower_bound_vector_example3},
\end{equation}
where 
\[
R(\theta)=\frac{1}{2}\sum_{i=1}^{P\wedge M}\log^{+}\left(\lambda_i \left(\mathbf{C_{U|V}}\right)/\theta\right).
\] 
\end{example}
If $P=1$, then $\mathbf{C_{U|V}}$ has a single non-zero eigenvalue and equality holds in \eqref{eq:lower_bound_vector_example3}. 
As will be seen by the next example, equality in \eqref{eq:lower_bound} when the observable process is one-dimensional is indeed special to the i.i.d case. The next example will also be used later  to prove a lower bound for the combined sampling and source coding problem in Theorem~\ref{thm:lower_bound}.
\begin{example}[vector stationary source]
Let $\mathbf X[n]=\left( X_1[n],\ldots,X_M[n] \right)$, $n\in \mathbb Z$, be a Gaussian stationary  vector source and let $Y[\cdot]$ be a one-dimensional process jointly Gaussian and stationary with $\mathbf X[\cdot]$. From Theorem~\ref{thm:indirect_vector_case}, it follows that the iDRF of $\mathbf X[\cdot]$ given $Y[\cdot]$ is 
\begin{align*}
D_{\mathbf X|Y}(\theta)=\mmse_{\mathbf X|Y}
+\frac{1}{M}\int_{-\frac{1}{2}}^\frac{1}{2} \min \left\{\sum_{m=1}^M S_{X_m|Y}\left(e^{2\pi i\phi}\right),\theta\right\} d\phi,
\end{align*}
\[
R(\theta)=\frac{1}{2} \int_{-\frac{1}{2}}^{\frac{1}{2}} \log^+\left[ \sum_{m=1}^M  S_{ X_m|Y}\left(e^{2\pi i\phi}\right)/\theta\right]d\phi.
\]
Here we used the fact that the rank of $\mathbf S_{\mathbf X|Y}\left(e^{2\pi i\phi}\right)$ is at most one, and thus the sum of the eigenvalues of $\mathbf S_{\mathbf X|Y}\left(e^{2\pi i\phi}\right)$ equals its trace, which is given by $\sum_{m=1}^M S_{X_m|Y}\left( e^{2\pi i \phi}\right)$. Considering the $m^\textrm{th}$ coordinate of $\mathbf X[\cdot]$ separately, the iDRF of $X_m[\cdot]$ given $Y[\cdot]$ is 
\begin{align*}
D_{X_m|Y}(\theta)= &\mmse_{ X_m|Y}+\int_{-\frac{1}{2}}^\frac{1}{2} \min \left\{ S_{X_m|Y}\left(e^{2\pi i\phi}\right),\theta\right\} d\phi,
\end{align*}
where
\[
R(\theta_m)=\frac{1}{2} \int_{-\frac{1}{2}}^{\frac{1}{2}} \log^+\left[  S_{X_m|Y}\left(e^{2\pi i\phi}\right)/\theta_m\right]d\phi.
\]
Since $\mmse_{\mathbf X|Y}=\frac{1}{M}\sum_{m=1}^M\mmse_{X_m|Y}$, the bound \eqref{eq:lower_bound} implies
\begin{align*}
& \int_{-\frac{1}{2}}^\frac{1}{2} \sum_{m=1}^M \min\left\{ S_{X_m|Y}\left(e^{2\pi i\phi}\right),\theta_m(R) \right\} d\phi \\
&\quad \quad \leq \int_{-\frac{1}{2}}^\frac{1}{2} \min\left\{\sum_{m=1}^M S_{X_m|Y}\left( e^{2\pi i \phi}\right),\theta(R) \right\} d\phi.
\end{align*}
\end{example}

\section{Indirect DRF under sub-Nyquist sampling \label{sec:Frequency-Rate-Distortion-Functi}}
In this section we solve our main source coding problem for the case of single branch sampling. Specifically, we derive a closed form expression for the function $D\left(f_s,R\right)$ defined in 
\eqref{eq:D_def} and for its minimal value over all pre-sampling filters $H(f)$. \par
From the definition of $D\left(f_s,R\right)$ in Section~\ref{sec:problem_statement} we can already deduce the following facts about $D(f_s,R)$:
\begin{prop}
\label{prop:basic_properties_of_DSFR}
Consider the combined sampling and source coding problem of Section~\ref{sec:problem_statement}. The function $D(f_s,R)$ satisfies:
\begin{enumerate}
\item[(i)]
For all $f_s>0$ and $R\geq0$,
\[
 D(f_s,R)\geq D_{X|Z}(R),
 \]
 where $D_{X|Z}(R)$ was defined in \eqref{eq:DnT}. In addition,
 \[
 D(f_s,R)\geq\mmse_{X|Y}(f_s),
 \]
 where $\mmse_{X|Y}(f_s)$ is the MMSE in reconstructing $X(\cdot)$ from the uniform samples $Y[\cdot]$ given in Proposition \ref{prop:mmse_single}.

\item [(ii)]
If the process $Z(\cdot)$ has almost surely Riemann integrable realizations, then the reconstruction error of $Z(\cdot)$ from $Y[\cdot]$ can be made arbitrarily small by sampling at a high enough frequency\footnote{Note that $Z(\cdot)$ does not need to be bandlimited. The only assumption on $Z(\cdot)$ is finite variance, i.e. $S_Z(f)$ is in $L_1$.}. It follows that as $f_s$ goes to infinity, $D\left(f_{s},R\right)$ converges to $D_{X|Z}\left(R\right)$. In particular, if $Z\left(\cdot\right)$ is bandlimited, then $D\left(f_{s},R\right)=D_{X|Z}\left(R\right)$ for any $f_{s}$ above the Nyquist frequency of $Z\left(\cdot\right)$.
\item [(iii)] For a fixed $f_s>0$, $D(f_s,R)$ is a monotone non-increasing function of $R$ which converges to $ \mmse_{X|Y}(f_s)$ as $R$ goes to infinity. It is not necessarily non-increasing in $f_s$ since $ \mmse_{X|Y}(f_s)$ is not necessarily non-increasing in $f_s$. \\
\end{enumerate}
{\,}
\end{prop}
\par

\subsection{Lower bound}
Note that $(i)$ in Proposition~\ref{prop:basic_properties_of_DSFR} implies that the manifold defined by $D\left(f_s,R\right)$ in the three dimensional space $\left(f_s,R,D\right)$ is bounded from below by the two cylinders $ \mmse_{X|Y}(f_s)$ and $D_{X|Z}(R)$ (and from above by the plane $D=\sigma_X^2$). A tighter lower bound is obtained using Proposition~\ref{prop:lower_bound}.
\begin{theorem} \label{thm:lower_bound}
Consider the combined sampling and source coding problem of Fig.~\ref{fig:operational_scheme} with the single branch sampler of Fig.~\ref{fig:sampling_scheme}-(a). We have the following bound of the indirect distortion-rate of $X(\cdot)$ given $Y[\cdot]$:
\begin{align} \label{eq:lower_bound_to_show}
D\left(f_s,R\right) \geq &\mmse_{X| Y}(f_s)\\ \nonumber
& + \int_{0}^\Delta \int_{-\frac{1}{2}}^\frac{1}{2} \min \left\{S_{X_\Delta| Y}\left(e^{2\pi i\phi} \right) ,\theta_\Delta\right\}d\phi d\Delta,
\end{align}
where 
\[
 S_{X_\Delta|Y}(e^{2\pi i \phi})=\frac {\sum_{k,l\in \mathbb Z} S_{XZ}\left(f_s(\phi-k)\right) S_{XZ}^*\left(f_s(\phi-l)\right) e^{2\pi i (k-l) \Delta} }{\sum_{k\in \mathbb Z}S_Z\left(f_s(\phi-k)\right)} \nonumber,
\]
and for each $0\leq\Delta\leq 1$, $\theta_\Delta$ satisfies
\[
\bar{R}=R/f_s=\frac{1}{2} \int_{-\frac{1}{2}}^\frac{1}{2} \log^+\left[ S_{X_\Delta| Y}\left(e^{2\pi i\phi} \right)/\theta_\Delta\right]d\phi.
\]
\end{theorem}
\begin{proof}
For a given finite set of points $\Delta_1,\ldots \Delta_M$ in $[0,1)$ define the vector valued process 
\[
{\mathbf X}^M[n] = \left( X_{\Delta_1}[n],\ldots,X_{\Delta_M}[n]\right),\quad n\in \mathbb Z,
\]
where for $m=1,\ldots,M$, the discrete-time process $X_{\Delta_m}[\cdot]$ is defined in \eqref{eq:X_delta_def}. By Proposition~\ref{prop:lower_bound} we have  
\begin{equation} \label{eq:lower_bound_proof}
D_{{\mathbf X}^M| Y}(\bar{R}) \geq \frac{1}{M}\sum_{m=1}^M D_{X_{\Delta_m}| Y}(\bar{R}).
\end{equation}
It follows from the proof of Proposition~\ref{prop:mmse_single} that for all $m=1,\ldots,M$, $X_{\Delta_m}[\cdot]$ and $Y[\cdot]$ are jointly Gaussian and stationary, with $S_{X_{\Delta_m}|Y}\left( e^{2\pi i\phi}\right)$ given by \eqref{eq:mmse_single_proof}. Applying the discrete-time version of Theorem~\ref{thm:[Dobrushin-and-Tsybakov]}, i.e. \eqref{eq:DnT_discrete}, the iDRF of $X_{\Delta_m}[\cdot]$ given $Y[\cdot]$ is 
\begin{align} \label{eq:lower_bound_proof_2}
D_{X_{\Delta_m}|Y}(\bar{R}) = \mmse_{X_{\Delta}|Y} +\int_{-\frac{1}{2}}^\frac{1}{2} \min \left\{S_{X_{\Delta_m}|Y}\left( e^{2\pi i\phi}\right), \theta_{\Delta_m}\right\}d\phi,
\end{align}
where for a fixed $\bar{R}$, $\theta_{\Delta_m}$ satisfies
\[
\bar{R}(\theta_{\Delta_m})=\int_{-\frac{1}{2}}^\frac{1}{2} \log^+\left[ S_{X_{\Delta_m}|Y}\left( e^{2\pi i\phi}\right)/\theta_{\Delta_m}\right]d\phi.
\] \par
$D_{X_{\Delta}}(\bar{R})$ is a continuous function of $\Delta$ and hence integrable with respect to it. As the number of points $M$ goes to infinity with vanishing division parameter $\max_{m_1\neq m_2}\left|\Delta_{m_1}-\Delta_{m_2}\right|$, the RHS of \eqref{eq:lower_bound_proof} converges to the integral of \eqref{eq:lower_bound_proof_2} over the interval $(0,1)$. The RHS of \eqref{eq:lower_bound_proof} converges to $D_{X|Y}(\bar{R})=D(f_s,R)$ by a similar argument that is used in the proof of Theorem~\ref{thm:main_result} to follows. Using \eqref{eq:mmse_average_equiv}, \eqref{eq:lower_bound_to_show} follows. 
\end{proof}

\subsection{Discrete-time sampling}
We first solve the discrete-time counterpart of our main source coding problem. Here the underlying process is $X\left[\cdot\right]$ and we observe a factor $M$ down-sampled version of the discrete time process $Z\left[\cdot\right]$, which is jointly Gaussian and jointly stationary with $X\left[\cdot\right]$. Note that unlike what was discussed in Section \ref{sec:Indirect-Source-Coding}, the source process and the observable process are no longer jointly stationary. 
\begin{theorem}[single branch decimation]
\label{thm:discrete_decimation_rate_distortion} Let $X\left[\cdot\right]$
and $Z\left[\cdot\right]$ be two jointly Gaussian stationary processes. Given $M\in \mathbb N$, define the process $Y[\cdot]$ by
 $Y\left[n\right]=Z\left[Mn\right]$, for all $n\in\mathbb Z$. The indirect distortion-rate
function of $X\left[\cdot\right]$ given $Y\left[\cdot\right]$, under the quadratic distortion \eqref{eq:dist_disc}, is given by
\begin{align*}
R \left(\theta\right)& = \frac{1}{2}\int_{-\frac{1}{2}}^{\frac{1}{2}}\log^{+}\left[J_{M}\left(e^{2\pi i\phi}\right)\theta^{-1}\right]d\phi,
\end{align*}
\begin{align*}
D\left(\theta\right) & = 
\mmse_{X|Y}(M)  +\int_{-\frac{1}{2}}^{\frac{1}{2}}\min\left\{ J_{M}\left(e^{2\pi i\phi}\right),\theta\right\} d\phi\\
 & =  \sigma_X^2 -\int_{-\frac{1}{2}}^{\frac{1}{2}}\left[J_{M}\left(e^{2\pi i\phi}\right)-\theta\right]^{+}d\phi,
\end{align*}
where 
\[
J_M \expphi \triangleq \frac{1}{M}\frac{\sum_{m=0}^{M-1}\left|S_{XZ}\left(e^{2\pi i\frac{\phi-m}{M}}\right)\right|^{2}}{\sum_{m=0}^{M-1}S_{Z}\left(e^{2\pi i\frac{\phi-m}{M}}\right)},
\]
and $\mmse_{X|Y}(M)$
is defined by
\begin{align*}
\lim_{N\rightarrow\infty}\frac{1}{2N+1}\sum_{n=-N}^{N}\mathbb{E}\left[X\left[n\right]-\mathbb E\left({X}[n]|Y[\cdot]\right)\right]^{2}\\
=\frac{1}{M}\sum_{n=0}^{M-1}\mathbb{E}\left[X\left[n\right]-\mathbb E\left({X}[n]|Y[\cdot]\right)\right]^{2}.
\end{align*}
\textup{ }\end{theorem}

\begin{proof}
While the details can be found in Appendix~\ref{sec:Proof-of-discrete}, an outline of the proof is as follows:
given $M\in \mathbb N$, define the vector-valued process $\mathbf X^M[\cdot]$ by
\[
\mathbf X^M[n] \triangleq \left(X[Mn],X[Mn+1],\ldots,X[Mn+M-1] \right),\quad n\in\mathbb Z.
\]
The process $\mathbf X^M[\cdot]$ is a stacked version of $X[\cdot]$ over $M$-length blocks, and hence shares the same iDRF given $Y[\cdot]$. Since $\mathbf X[\cdot]$ and $Y[\cdot]$ are jointly Gaussian and stationary, the result follows by applying Theorem \ref{thm:indirect_vector_case}. 
\end{proof}

\subsection{Single branch sampling}
We are now ready to solve our combined sampling and source coding problem introduced in Section~\ref{sec:problem_statement}. Note that here we go back to the model of Fig.~\ref{fig:The-general-scheme} with the single branch sampler of Fig.~\ref{fig:sampling_scheme}(a).
\begin{theorem}[single branch sampling]
\label{thm:main_result}Let $X\left(\cdot\right)$ and $Z\left(\cdot\right)$
be two jointly Gaussian stationary stochastic processes with almost surely Riemann integrable realizations and $L_1$ PSDs $S_{X}\left(f\right)$,
$S_{Z}\left(f\right)$ and $S_{XZ}\left(f\right)$. Let $Y\left[\cdot\right]$
be the discrete time process defined by $Y\left[n\right]=Z\left(n/f_{s}\right)$,
 where $f_{s}>0$. The indirect
distortion-rate function of $X\left(\cdot\right)$ given $Y\left[\cdot\right]$, is given by
\begin{subequations}
\label{eq:main_theorem}
\begin{align}
R\left(f_s,\theta\right)  = & \frac{1}{2}\int_{-\frac{f_s}{2}}^{\frac{f_s}{2}}\log^{+}\left[\widetilde{S}_{X|Y}(f)\theta^{-1}\right]df,\label{eq:main_rate}
\end{align}
\begin{align}
\label{eq:main_dist}
D\left(f_{s},\theta\right)  = & \mmse_{X|Y}(f_s)+\int_{-\frac{f_s}{2}}^{\frac{f_s}{2}}\min\left\{ \widetilde{S}_{X|Y}(f),\theta\right\} df\\
  =  &\sigma_{X}^{2}-\int_{-\frac{f_s}{2}}^{\frac{f_s}{2}}\left[\widetilde{S}_{X|Y}(f)-\theta\right]^{+}df, \nonumber
\end{align}
\end{subequations}
where $\sigma_X^2 = \mathbb E \left(X(t)\right)^2$,
\begin{align} \label{eq:J_def_single}
\widetilde{S}_{X|Y}(f) &= \frac{\sum_{k\in\mathbb{Z}}\left|S_{XZ}\left(f-f_sk\right)\right|^{2}}{\sum_{k\in\mathbb{Z}}S_{Z}\left(f-f_sk\right)},
\end{align}
and 
\[
\mmse_{X|Y}(f_s) =  \sigma_{X}^{2}-\int_{-\frac{f_s}{2}}^{\frac{f_s}{2}}\widetilde{S}_{X|Y}(f)df.
\]
\end{theorem}
\begin{proof}
see Appendix \ref{sec:proof:main_result}. The basic idea of the proof is to 
approximate the continuous time processes $X\left(\cdot\right)$ and $Z\left(\cdot\right)$ by discrete time processes, and take the limit in the solution to the discrete problem given by Theorem \ref{thm:discrete_decimation_rate_distortion}.
\end{proof}

\subsection{Discussion \label{sub:Discussion:}}
We see that for a given sampling frequency $f_{s}$, the optimal solution has a similar form as in the stationary case \eqref{eq:DnT} and Theorem \ref{thm:[Dobrushin-and-Tsybakov]}, where the function $\widetilde{S}_{X|Y}(f)$ takes the role of $S_{X|Z}\left(f\right)$. That is, the minimal distortion is obtained by a MMSE term plus a term determined by reverse waterfilling over the function $\widetilde{S}_{X|Y}(f)$. By writing 
\begin{align*}
\mmse_{X|Y}(f_s) & = \sigma_X^2 - \int_{-\frac{f_s}{2}}^\frac{f_s}{2} \widetilde{S}_{X|Y}(f)df  \\
& = \int_{-\frac{f_s}{2}}^\frac{f_s}{2} \left( \sum_{k\in \mathbb Z} S_X(f-f_sk) - \widetilde{S}_{X|Y}(f) \right) df,
\end{align*}
we see that \eqref{eq:main_theorem} has a waterfilling interpretation similar to Fig.~\ref{fig:waterfilling}, which is given by Fig.~\ref{fig:main_result_waterfilling}.

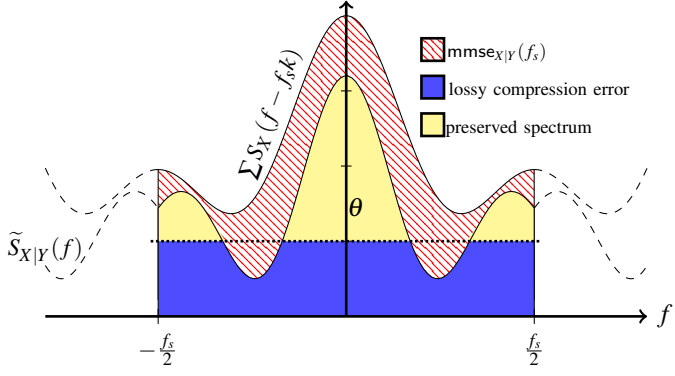
\begin{figure}
\begin{center}
\begin{tikzpicture}[scale=1]
 \fill[fill=red!50, pattern=north west lines, pattern color=red] (-2.5,0.22)--plot[domain=-2.5:2.5, samples=100] (\x, {(1.7+2.3*cos(20*3.14*\x)^2)*exp(-\x*\x/10)}) -- (2.5,0.22);	
 \draw (-2.5,0.22)--plot[domain=-2.5:2.5, samples=100] (\x, {(1.7+2.3*cos(20*3.14*\x)^2)*exp(-\x*\x/10)}) -- (2.5,0.22);	
 
  \draw[dashed] plot[domain=2.5:4, samples=100] (\x, {(1.7+2.3*cos(20*3.14*(\x-5))^2)*exp(-(\x-5)*(\x-5)/10)});	

  \draw[dashed] plot[domain=-4:-2.5, samples=100] (\x, {(1.7+2.3*cos(20*3.14*(\x+5))^2)*exp(-(\x+5)*(\x+5)/10)});

 \fill[fill=blue!70]  (-2.5,0)--plot[domain=-2.5:2.5, samples=100] (\x, {(0.6+2.6*cos(24*3.14*\x)^2)*exp(-\x*\x/8)}) -- (2.5,0);	

 \fill[fill=yellow!50] (-2.5,1) -- plot[domain=-2.5:-1.8, samples=100] (\x, {(0.6+2.6*cos(24*3.14*\x)^2)*exp(-\x*\x/8)}) --(-1.6,1);	

 \fill[fill=yellow!50] (1.6,1) -- plot[domain=1.8:2.5, samples=100] (\x, {(0.6+2.6*cos(24*3.14*\x)^2)*exp(-\x*\x/8)}) --(2.5,1);	

 \fill[fill=yellow!50] (-0.85,1) -- plot[domain=-0.75:0.75, samples=100] (\x, {(0.6+2.6*cos(24*3.14*\x)^2)*exp(-\x*\x/8)}) --(0.85,1);	

 \draw (-2.5,0)--plot[domain=-2.5:2.5, samples=100] (\x, {(0.6+2.6*cos(24*3.14*\x)^2)*exp(-\x*\x/8)})  -- (2.5,0);	
 
 \draw[dashed] plot[domain=-4:-2.5, samples=100] (\x, {(0.6+2.6*cos(24*3.14*(\x+5))^2)*exp(-(\x+5)*(\x+5)/8)});	

 \draw[dashed] plot[domain=2.5:4, samples=100] (\x, {(0.6+2.6*cos(24*3.14*(\x-5))^2)*exp(-(\x-5)*(\x-5)/8)});

  \foreach \y/\ytext in {1/,2/,3/}
    \draw[shift={(0,\y)}] (2pt,0pt) -- (-2pt,0pt) node[left] {$\ytext$};

\draw (2.5,-0.1)node[below] {\small $\frac{f_s}{2}$} -- (2.5,0.1) ;
\draw (-2.5,-0.1)node[below] {\small $-\frac{f_s}{2}$} -- (-2.5,0.1) ;

\draw [fill=red!50, line width=1pt, pattern=north west lines, pattern color=red] (1,3.4) rectangle  (1.3,3.7) node[left, xshift = 1.5cm, yshift = -0.2cm] {\scriptsize $\mmse_{X|Y}(f_s)$};

\draw [fill=blue!70, line width=1pt] (1,2.9) rectangle  (1.3,3.2) node[left, xshift = 2.6cm, yshift = -0.2cm, align = center] {\scriptsize lossy compression error};

\draw [fill=yellow!50, line width=1pt] (1,2.4) rectangle  (1.3,2.7) node[left, xshift=2.1cm, yshift = -0.2cm,align = left] {\scriptsize preserved spectrum};

\node at (0.13,1.45) {$\theta$};
\draw[densely dotted, line width=1pt] (-2.6,1) -- (2.6,1);
\draw[line width=1pt,->]  (-4,0)--(4,0) node[right] {$f$};
\draw[line width=1pt,->]  (0,0)--(0,4.2) node[right] {};
\node at (-1,2.7) [rotate=70] {\small $\sum S_X \left(f-f_sk\right)$};
\node at (-4,0.9) [rotate=0] {\small $\widetilde{S}_{X|Y}(f)$};
\end{tikzpicture}

\caption{Waterfilling interpretation of \eqref{eq:main_theorem}. The function $D(f_s,R)$ is the sum of the MMSE and the lossy compression error. \label{fig:main_result_waterfilling}}

\end{center}
\end{figure}

Comparing equations \eqref{eq:main_dist} and \eqref{eq:prop_separation}, we have the following interpretation of the second term in \eqref{eq:main_dist}:
\begin{prop} \label{prop:cyclo}
The (direct) distortion-rate function of the non-stationary process $\tilde{X}\left(\cdot\right)=\left\{ \mathbb{E}\left[X(t)|Y\left[\cdot\right]\right],\, t\in\mathbb{R}\right\}$ is given by
\begin{align*}
R\left(\theta\right) & =  \frac{1}{2}\int_{-\frac{f_{s}}{2}}^{\frac{f_{s}}{2}}\log^{+}\left[ \widetilde{S}_{X|Y}(f)  \theta^{-1}\right]df, \\
D_{\tilde{X}}\left(\theta\right) &  =  \int_{-\frac{f_{s}}{2}}^{\frac{f_{s}}{2}}\min\left\{ \widetilde{S}_{X|Y}(f),\theta\right\} df,
\end{align*}
where $\widetilde{S}_{X|Y}(f)$ is defined by \eqref{eq:J_def_single}.
\end{prop}
The process $\tilde{X}(\cdot)$ is in fact a \emph{cyclo-stationary} process. A deeper treatment of the DRF of such processes is provided in \cite{KipnisCyclo}, where the idea behind the proof of Theorem~\ref{thm:main_result} is extended to derive a general form for the DRF of such processes. 


The function $\widetilde{S}_{X|Y}(f)$ depends on the sampling frequency, the filter $H(f)$ and the spectral densities $S_X(f)$ and $S_\eta(f)$, but is independent of $R$. If we fix $R$ and consider a change in $\widetilde{S}_{X|Y}(f)$ such that 
\[
\widetilde{C}_{X|Y} \triangleq \int_{-\frac{f_s}{2}}^\frac{f_s}{2} \widetilde{S}_{X|Y}(f)df
\]
is increased, then from \eqref{eq:main_rate} we see that $\theta$ also increases to maintain the same fixed rate $R$. On the other hand, the expression for $D\left(f_s,R\right)$ in \eqref{eq:main_dist} exhibits a negative linear dependency on $\widetilde{C}_{X|Y}$. In this interplay between the two terms in \eqref{eq:main_dist}, the negative linear dependency in $\widetilde{S}_{X|Y}(f)$ is stronger then a logarithmic dependency of $\theta$ in $\widetilde{C}_{X|Y}$ and the distortion reduces with an increment in $\widetilde{C}_{X|Y}$. The exact behavior is obtained by taking the functional derivative of $D\left(f_s,R\right)$ with respect to $\widetilde{S}_{X|Y}(f)$ at the point $f\in \left(-f_s/2,f_s/2\right)$, which is non-positive. A simple analogue for that dependency can be seen in Example~\ref{ex:simple_scalar_case}, where the distortion in \eqref{eq:D_simple} is a non-increasing function of $C_{U|V}$. \par
We summarize the above in the following proposition:
\begin{prop} \label{prop:min_max}
For a fixed $R\geq 0$, minimizing $D\left(f_s,R\right)$ is equivalent to maximizing 
\[
\int_{-\frac{f_s}{2}}^{\frac{f_s}{2}} \widetilde{S}_{X|Y}(f) df,
\]
where $\widetilde{S}_{X|Y}(f)$ is defined by \eqref{eq:J_def_single}.
\end{prop} 

This says that a larger $\widetilde{S}_{X|Y}(f)$ accounts for more information available about the source through the samples $Y[\cdot]$, and motivates us to bound $\widetilde{S}_{X|Y}(f)$. Since $\widetilde{S}_{X|Y}(f)$ can be written as
\[
\widetilde{S}_{X|Y}(f) =\frac{\sum_{k\in\mathbb{Z}}S_{X|Z}\left(f-f_{s}k\right)S_{Z}\left(f-f_{s}k\right)}{\sum_{k\in\mathbb{Z}}S_{Z}\left(f-f_{s}k\right)},
\]
the following holds for almost every $f\in\left(-\frac{f_{s}}{2},\frac{f_{s}}{2}\right)$,
\begin{align}
\widetilde{S}_{X|Y}(f) & \leq\sup_{k}S_{X|Z}\left(f-f_{s}k\right) \nonumber \\
&=\sup_k \frac{S_X^2(f-f_sk)|H(f-f_sk)|^2}{S_{X+\eta}(f-f_sk)|H(f-f_sk)|^2} \nonumber \\
&=\sup_k \frac{S_X^2(f-f_sk)}{S_{X+\eta}(f-f_sk)} ,\label{eq:holder_ineq}
\end{align}
with equality if and only if for each $k\in\mathbb{Z}$, either $S_{X|Z}\left(f-f_{s}k\right)=\sup_{k}S_{X|Z}\left(f-f_{s}k\right)$ or $S_{Z}\left(f-f_{s}k\right)=0$. 
Thus, we have the following proposition.
\begin{prop}\label{prop:less_than_sup}
 For all $f_{s}>0$ and $R\geq0$, the indirect distortion-rate function of $X(\cdot)$ given $Y[\cdot]$ satisfies
  \[
D\left(f_{s},R\right)\geq D^*\left(f_s,R\right),
 \]
where $D^*\left(f_s,R\right)$ is the distortion-rate function of the Gaussian stationary process with PSD 
\[
\widetilde{S}^*(f)= \begin{cases} 
\sup_k \frac{S_X^2(f-f_sk)}{S_{X+\eta}(f-f_sk)}, 
& f\in \left(-\frac{f_s}{2},  \frac{f_s}{2}\right), \\
0, & \textrm{otherwise}.
\end{cases}
\]
\end{prop}
Note that the last expression is independent of the pre-sampling filter $H(f)$. Therefore, Proposition \ref{prop:less_than_sup} describes a lower bound which depends only on the statistics of the source and the noise. We will see in Theorem~\ref{thm:filterbanks_main} below that $D^*\left(f_s,R\right)$ is attainable for any given $f_s$ if we are allowed to choose the pre-sampling filter $H(f)$.\\

It is interesting to observe how Theorem \ref{thm:main_result} agrees with
the properties of $D(f_s,R)$ in the two special cases illustrated in Fig.~\ref{fig:diagram_scheme}.
\begin{enumerate}
\item[(i)] For $f_{s}$ above the Nyquist frequency of $Z(\cdot)$, $S_{Z}\left(f-f_{s}k\right)=0$ for any $k\neq0$. In this case the conditions for equality in \eqref{eq:holder_ineq} hold and 
\[
\widetilde{S}_{X|Y}(f)=\sup_k \frac{S_X^2(f-f_sk)}{S_{X+\eta}(f-f_sk)}=\frac{S_X^2(f)}{S_{X+\eta}(f)},
\] 
which means that \eqref{eq:main_theorem} is equivalent to \eqref{eq:DnT}. 
\item[(ii)]
If we take $R$ to infinity, then $\theta$ goes to zero and  \eqref{eq:main_dist} reduces to \eqref{eq:mmse_single_theorem}.
\end{enumerate}

In view of the above we see that Theorem \ref{thm:main_result} subsumes the two classical problems of finding $\mmse_{X|Y}(f_s)$ and $D_{X|Z}(R)$.

\subsection{Examples}
In Examples \ref{ex:Rectangular-spectrum} and \ref{ex:Non-monotonicity-in-the} below we derive a single letter expression for the function $D\left(f_s,R\right)$ under a given PSD $S_X(f)$, zero noise $S_\eta(f) \equiv 0$ and unit pre-sampling filter $|H(f)|\equiv 1$, i.e. when $S_X(f)=S_Z(f) $.
\setlength{\parindent}{0mm}
\begin{example}[rectangular spectrum] \label{ex:Rectangular-spectrum}

\setlength{\parindent}{3mm}
Let the spectrum of the source $X\left(\cdot\right)$ be
\[
S_{X}\left(f\right)=\begin{cases}
\frac{\sigma^{2}}{2W} & |f|\leq W,\\
0 & \text{otherwise},
\end{cases}
\]
for some $W>0$ and $\sigma>0$. In addition, assume that the noise is constant over the band $|f|\leq W$ with intensity $\sigma_\eta^2 = \gamma^{-1} \sigma_X^2 =\gamma^{-1} \sigma^2/(2W)$, where $\gamma>0$ can be seen as the SNR. For all frequencies $f\in\left(-f_s/2,f_s/2\right)$,
\begin{align*}
\widetilde{S}_{X|Y}(f) & = & \frac{\sum_{k\in\mathbb{Z}}S_{X}^{2}\left(f-f_s k\right)}{\sum_{k\in\mathbb{Z}}S_{X+\eta}\left(f-f_{s}k\right)}
  =  \frac{\sigma^{2}}{2W}\begin{cases}
\frac{\gamma}{1+\gamma} & \left|f\right|<W,\\
0 & \left|f\right|\geq W.
\end{cases}
\end{align*}
By Theorem \ref{thm:main_result} we have 
\begin{align*}
R\left(f_s,\theta\right) & =
 \begin{cases}
\frac{f_s}{2} \log \left( \frac{\sigma^2 \gamma } {2W\theta(1+\gamma)}\right) & 0\leq\frac{\theta}{\sigma^2 }(1+\gamma^{-1})<\frac{f_s}{2W}<1,\\
W  \log\left(\frac{\sigma^2 \gamma} {2W \theta(1+\gamma)}\right) & 0\leq\frac{\theta}{\sigma^2} (1+\gamma^{-1})<1\leq\frac{f_s}{2W},\\
0 & \text{otherwise},
\end{cases}
\end{align*}
and
\begin{align*}
D\left(f_s,\theta\right) & = \sigma^2  \begin{cases}
\left[1-\frac{f_{s}}{2W}\right]^{+}+\frac{\theta f_s}{\sigma^{2}} & \frac{\theta}{\sigma^{2}}\leq\min\left\{ \frac{f_s \gamma}{2W(1+\gamma)},1\right\} ,\\
1 & \text{otherwise}.
\end{cases}
\end{align*}
This can be written in a single expression as 
\begin{align}
D\left(f_{s},R\right) & = \sigma^{2}  \begin{cases}
1-\frac{f_{s}}{2W} + \frac{f_s}{2W} \frac{\gamma}{1+\gamma } 2^{\frac{-2R}{f_s}}    & \frac{f_{s}}{2W}<1,\\
\frac{1}{1+\gamma} + \frac{\gamma }{1+\gamma} 2^{-\frac{R}{W}}  & \frac{f_s}{2W}\geq1.
\end{cases}\label{eq:DistRate_square_spectrum-1}
\end{align}
\end{example} 
\par
Expression \eqref{eq:DistRate_square_spectrum-1} has a very intuitive structure: for frequencies
below the Nyquist frequency of the signal, the distortion as a function of the rate increases by a constant factor due to the error as a result of non-optimal sampling. This factor completely vanishes for $f_s$ greater than the Nyquist frequency of the signal, in which case $D\left(f_s,R\right)$
equals the iDRF of the process $X(\cdot)$
given $Z(\cdot) = X(\cdot)+\eta(\cdot)$, which by Theorem~\ref{thm:[Dobrushin-and-Tsybakov]} equals
\[
D_{X|Z}(R)=\mmse_{X|Z} + C_{X|Z}2^{-2R} = \frac{ \sigma^2}{1+\gamma}+  \frac{ \sigma^2  \gamma} {1+\gamma} 2^{-R/W}.
\]
This is depicted in Fig. \ref{fig:rect_spectrum} for $\gamma = 5$ and $\gamma\rightarrow \infty$.

\begin{figure}
\begin{center}
\begin{tikzpicture}
\node at (0,0) { \includegraphics[trim=0cm 0cm 0cm 0cm, clip=true, scale=0.47]{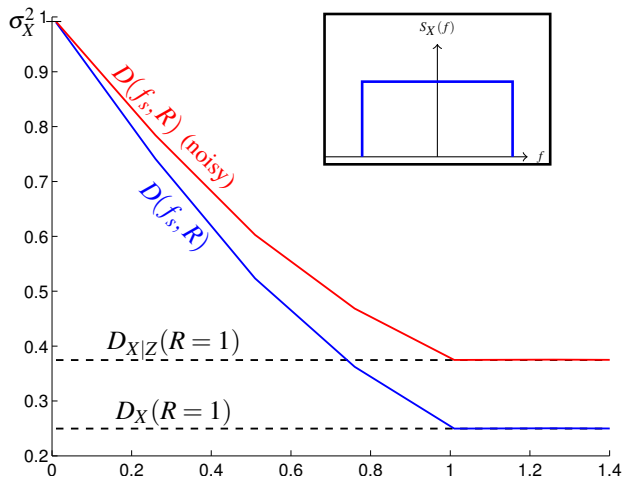}};
\draw[line width=1pt] (0,1) rectangle (3,3);
\draw[->] (0,1.1) -- (2.7,1.1) node[right] {\tiny $f$};
\draw[->] (1.5,1.1) -- (1.5,2.6) node[above] {\tiny $S_X(f)$};
\draw[line width=1pt, color=blue] (0.5,1.1) -- (0.5,2.1) -- (2.5,2.1) -- (2.5,1.1);
\draw (-3.7,2.9) node[left] {\small $\sigma_X^2$} --  (-3.5,2.9) ;
\node[rotate=-45] at (-2,1.5) {\color{red} $D(f_s,R)$ {\small (noisy)} };
\node[rotate=-45] at (-2,0.3) {\color{blue} $D(f_s,R)$ };
\node at (-2,-2.3) { $D_X(R=1)$ };
\node at (-2,-1.4) { $D_{X|Z}(R=1)$ };
\end{tikzpicture}
\caption{\label{fig:rect_spectrum} Distortion as a function of sampling frequency
$f_{s}$ and source coding rate $R=1\left[bit/sec\right]$ for a process with rectangular PSD and bandwidth $0.5$. The lower curve corresponds to zero noise and the upper curve corresponds to $S_\eta(f) = 0.2 S_X(f)$, where $\left|H(f)\right|\equiv 1$ in both cases. The dashed line represents the iDRF of the source given the pre-sampled process $Z(\cdot)$, which coincides with $D(f_s,R)$ for $f_s$ above the Nyquist frequency.}

\end{center}
\end{figure}

\begin{example}
\label{ex:Non-monotonicity-in-the}

\setlength{\parindent}{3mm}
The following example shows that the distortion-rate function is not necessarily monotonically decreasing in the sampling frequency. Here $S_{X}\left(f\right)$ has the band-pass structure 
\begin{equation}
S_{X}\left(f\right)=\begin{cases}
\frac{\sigma^{2}}{2} & 1\leq|f|\leq 2,\\
0 & \text{otherwise},
\end{cases}\label{eq:example_non_monotone}
\end{equation}
and we assume zero noise, i.e. $S_X(f)=S_Z(f)$.
We again obtain that for any $f\in\left(-f_s/2,f_s/2\right)$, $\widetilde{S}_{X|Y}(f)$ is either $\frac{\sigma^{2}}{2}$ or $0$. Thus, in order to find $D\left(f_s,R\right)$, all we need to know are for which values of $f\in\left(-f_s/2,f_s/2\right)$ the function $\widetilde{S}_{X|Y}(f)$ vanishes. This leads to
\[
D\left(f_{s},R\right)=\sigma^{2}\begin{cases}
2^{-R} & 4 \leq f_s,\\
1-\frac{f_s-2}{2} \left(1-2^{-\frac{2R}{f_s-2}}\right) & 3\leq f_s<4,\\
1-\frac{4-f_s}{2}\left(1-2^{-\frac{2R}{4-f_s}} \right) & 2\leq f_s<3,\\
1-(f_s-1)\left(1-2^{-\frac{R}{f_s-1}}\right) & 1.5\leq f_{s}<2, \\
1-(2-f_s) \left(1-2^{-\frac{R}{2-f_s}}\right) & 4/3 \leq f_{s}<1.5, \\
1-\frac{f_s}{2}\left(1-2^{-\frac{2R}{f_s}}\right) & 0\leq f_{s}<4/3, 
\end{cases}
\]
which is depicted in Fig. \ref{fig:non_monotone} for two different values of $R$.
\end{example}

\begin{figure}

\begin{center}

\begin{tikzpicture}

\node at (0,0) { \includegraphics[trim=0cm 0cm 0cm 0cm, clip=true, scale=0.47]{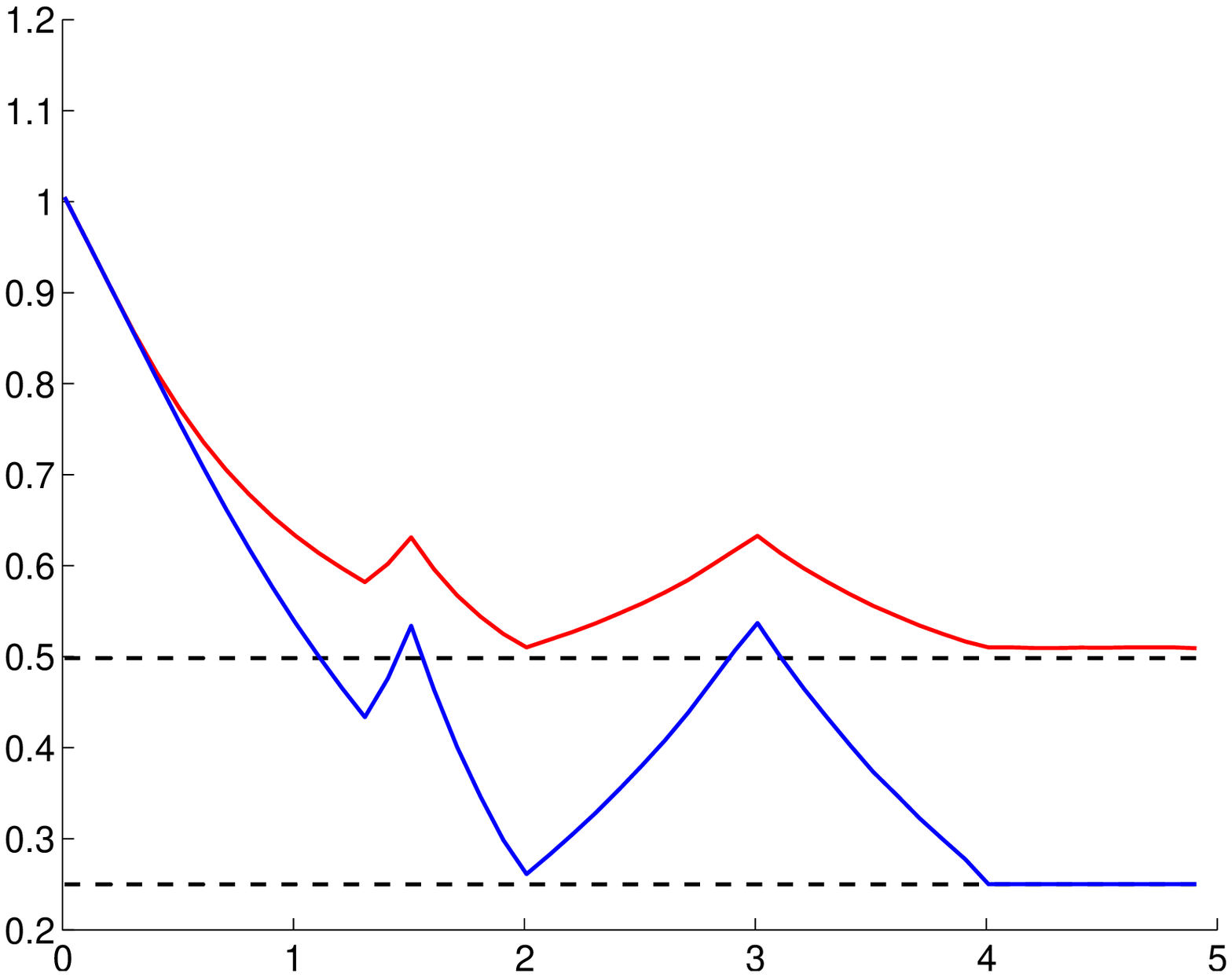}};
\draw[line width=1pt] (0,1) rectangle (3,3);
\draw[->] (0,1.1) -- (2.7,1.1) node[right] {\tiny $f$};
\draw[->] (1.5,1.1) -- (1.5,2.6) node[above] {\tiny $S_X(f)$};
\draw[line width=1pt, color=blue] (0.5,1.1) -- (0.5,2.1) -- (1,2.1)--(1,1.1) -- (2,1.1)--(2,2.1)--(2.5,2.1)--(2.5,1.1);
\draw (-3.7,1.8) node[left] {\small $\sigma_X^2$} --  (-3.5,1.8) ;

\draw[dotted] (2.35,-3)  node[below] {\small $f_{Nyq}$} -- (2.35,-1.1);
\draw[dotted]  (-0.6,-3) node[below] {\small $f_\mathrm{Landau}$} -- (-0.6,-1.1);

\node[rotate=0] at (3,-0.7) {\color{red} \small $D(f_s,R=1)$  };

\node[rotate=0] at (3.05,-2.3) {\color{blue} \small $D(f_s,R=2)$ };

\node at (-2.8,-2.4) {\small $D_X(R=2)$ };
\node at (-2.8,-0.95) {\small $D_X(R=1)$ };

\end{tikzpicture}
\caption{The function $D\left(f_{s},R\right)$ at two values of $R$ for the process with spectrum given in the small frame. Unlike in this example, single branch uniform sampling in general does not achieve $D(R)$ for $f_s \leq f_{Nyq}$. 
\label{fig:non_monotone}}

\end{center}
\end{figure}
 
\subsection{Optimal pre-sampling filter\label{subsec:Optimal-pre-Sampling-Filter}}
An optimization similar to the one carried out in Subsection~\ref{subsec:mmse_optimal} over the pre-sampling filter $H(f)$ can be performed over the function $\widetilde{S}_{X|Y}(f)$ in order to minimize the function $D(f_s,R)$. By Proposition~\ref{prop:min_max}, minimizing distortion for a given $f_{s}$ and $R$ is equivalent to maximizing $\widetilde{S}_{X|Y}(f)$ for every $f\in\left(-f_s/2,f_s/2\right)$ separately. But recall that the optimal pre-sampling filter $H^\star(f)$ that maximizes $\widetilde{S}_{X|Y}(f)$ was already given in Theorem \ref{thm:mmse_opt_single} in terms of the maximal aliasing free set associated with $\frac{S_X^2(f)}{S_{X+\eta}(f)}$. This leads us to the following conclusion:
\begin{prop}
\label{prop:opt_single} Given $f_s>0$, the optimal pre-sampling filter $H^\star(f)$ that minimizes $D(f_s,R)$, for all $R\geq 0$, is given by 
\[
H^{\star}\left(f\right)=\begin{cases}
1 & f\in F^\star,\\
0 & \text{otherwise},
\end{cases}
\]
where $F^\star \in AF(f_s)$ and satisfies 
\[
\int_{F^\star} \frac{S_X^2(f)}{S_{X+\eta}(f)}df= \int_{-\frac{1}{2}}^\frac{1}{2} \sup_k \frac{S_X^2(f-f_sk)}{S_{X+\eta}(f-f_sk)}df.
\]
The maximal value of $\widetilde{S}_{X|Y}(f)$ obtained this way is
\[
\widetilde{S}_{X|Y}^{\star} \left(f\right)  =  \sup_k \frac{S_X^2(f-f_sk)}{S_{X+\eta}(f-f_sk)},
\]
and the distortion-rate function at a given sampling frequency is given by
\begin{subequations}
\label{eq:optimal}
\begin{align}
R^\star \left(\theta\right) & =   \frac{1}{2}\int_{-\frac{f_s}{2}}^{\frac{f_s}{2}} \log^{+}\left[\widetilde{S}_{X|Y}^\star(f)/ \theta\right]df\\
& =  \frac{1}{2}\int_{F^\star}\log^{+}\left[\frac{S_X^2(f)}{S_{X+\eta}(f)} \theta^{-1}\right]df, \nonumber
\label{eq:filter_rate}
\end{align}
\begin{align}
D^{\star}\left(f_{s},\theta\right) & =  \sigma_{X}^{2}-\int_{-\frac{f_s}{2}}^{\frac{f_s}{2}}\left[\widetilde{S}_{X|Y}^\star(f)-\theta\right]^{+}df\\
& =  \sigma_{X}^{2}-\int_{F^\star}\left[\frac{S_X^2(f)}{S_{X+\eta}(f)}-\theta\right]^{+}df \nonumber.\label{eq:filter_dist}
\end{align}
\end{subequations}
\end{prop}
\begin{proof} 
From Theorem~\ref{thm:mmse_opt_single} we conclude that the filter $H^\star(f)$ that maximizes $\widetilde{S}_{X|Y}(f)$ is given by the indicator function of the maximal aliasing free set $F^\star$. Moreover, with this optimal filter, \eqref{eq:main_theorem} reduces to \eqref{eq:optimal}.
\end{proof}
We emphasize that even in the absence of noise, the filter $H^\star(f)$ still plays a crucial role in reducing distortion by preventing aliasing as described in Subsection~\ref{subsec:mmse_optimal}.   Fig.~\ref{fig:OPSF} illustrates the effect of the optimal pre-sampling filter on the function $D(f_s,R)$. 
 
 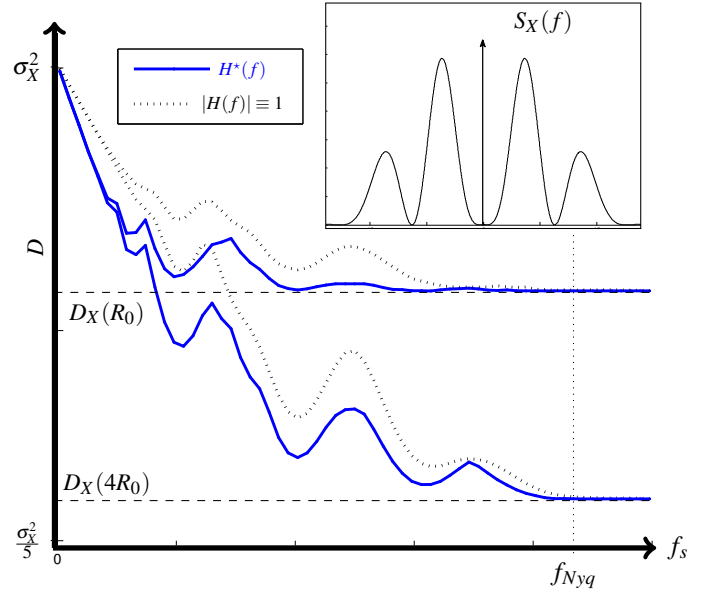
\begin{figure}
\begin{center}
\begin{tikzpicture}
\node at (0,0.1) {\includegraphics[trim=0cm 0cm 0cm 0cm, clip=true, scale=0.55] {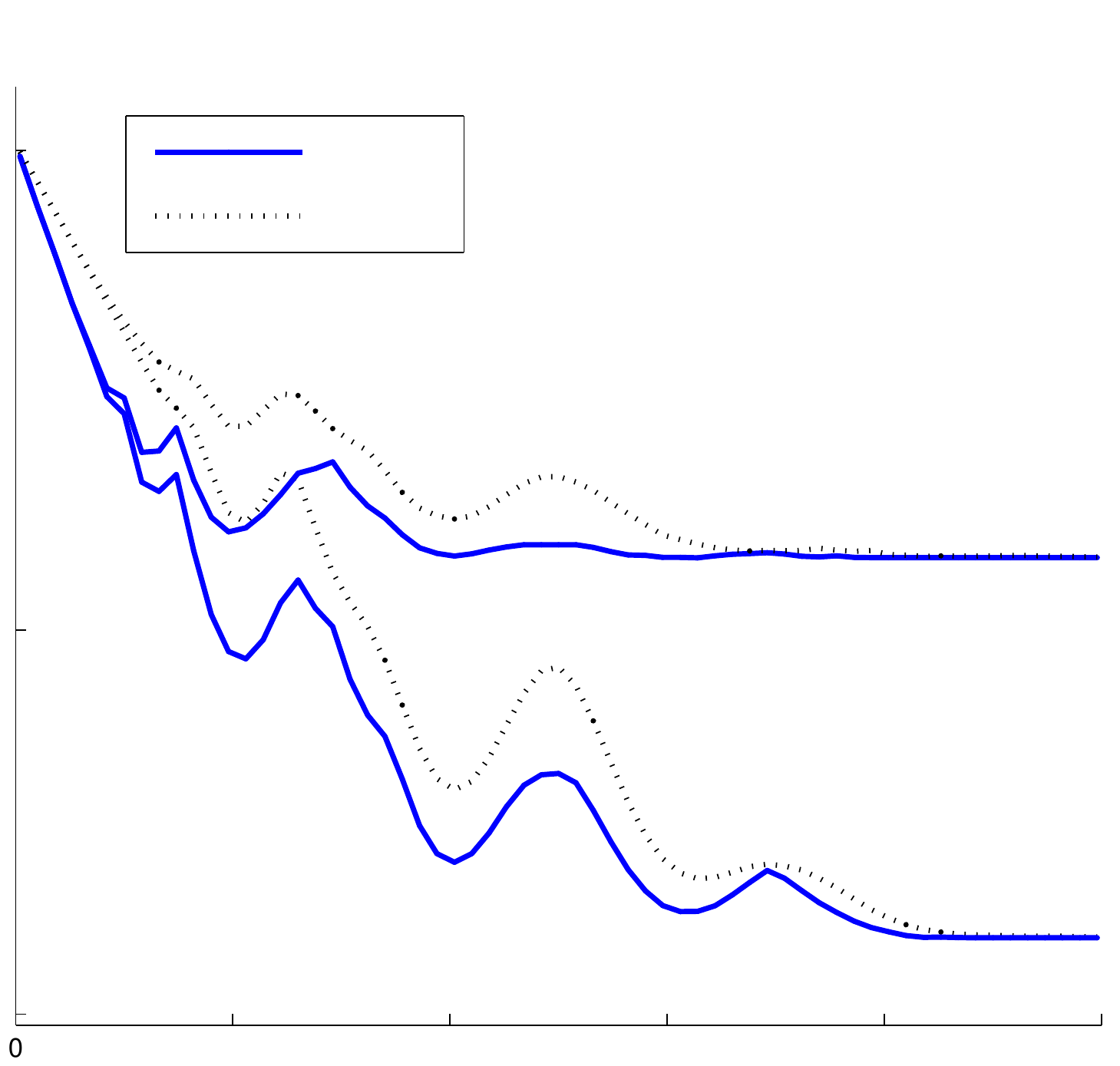}};
\node at (1.7,2.25) { \includegraphics[trim=2.5cm 8.5cm 1.7cm 7cm, clip=true, scale=0.24,frame]{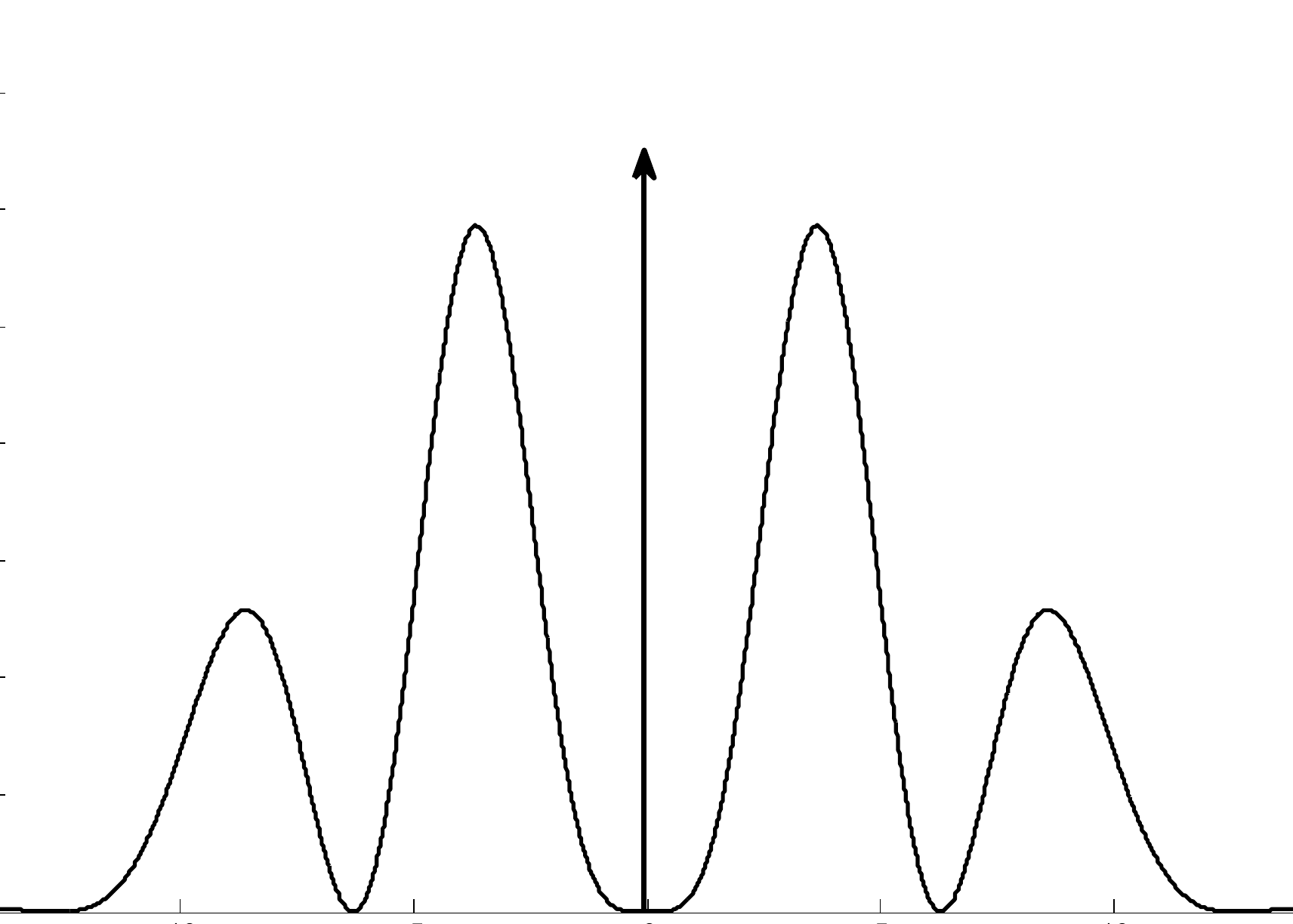} };
\node[fill=white] at (-1.5,2.85) {\scriptsize \color{blue} $H^\star(f)$};
\node[fill = white] at (-1.5,2.4) {\scriptsize  $\left| H(f) \right|\equiv 1$};

\draw[line width = 3, color=white] (-4,-3.52) -- (4,-3.52) ;
\draw[line width = 3, color=white] (-4.1,-3.52) -- (-4.1,3.3);
\draw[->,line width = 2] (-4.02,-3.5) -- (4,-3.5) node[right] {$f_s$};

\draw[->,line width = 2] (-4,-3.5) -- node[above, rotate = 90, xshift = 0.5cm] {\small $D$} (-4,3.5);

\node at (2.5,3.5) {\small $S_X(f)$};
\draw  (-4.05,2.9) node[left] {\small $\sigma_X^2$} -- (-3.95,2.9);

\draw  (-4.05,-3.4) node[left] {\small $\frac{\sigma_X^2}{5}$} -- (-3.95,-3.4);

\draw[dotted]  (2.9,-3.57) node[below] {$f_{Nyq}$} -- (2.9,0.7);

\draw[dashed] (-4,-0.1) -- node[below,xshift = -3.3cm] {\small $D_X(R_0)$} (4,-0.1);
\draw[dashed] (-4,-2.87) -- node[below, xshift = -3.3cm, yshift = 0.5cm] {\small $D_X(4R_0)$} (4,-2.87);

\end{tikzpicture}
\caption{\label{fig:OPSF} The functions $D^\star(f_s,R)$ and $D(f_s,R)$ at two fixed values of $R$. The plain and dashed curves were obtained using the optimal pre-sampling filter ($H(f)=H^\star(f)$) and without ($\left|H(f)\right|\equiv1$), respectively, for the same source statistic with $S_\eta(f)\equiv 0$ and $S_X(f)$ as given in the small frame.}
\end{center} 
\end{figure}
 
%
%
%
%
%
%
 
\setlength\parindent{10pt}
\section{Multi-branch sampling \label{sec:main_multi}}
We now generalize our analysis to the case where the sampling operation can be described by a multi-branch sampler as given in Fig.~\ref{fig:sampling_scheme}(b). Similar to the case of single branch sampling, we first consider the discrete-time counterpart and use it to derive our main result.

\subsection{Multi-branch decimation}
In the discrete-time counterpart of the combined sampling and source coding problem with multi-branch sampling, the source is the discrete-time process $X\left[\cdot\right]$ and the sampling operation at each branch is replaced by decimation by a factor $PM$, where $P\in \mathbb N$ is the number of sampling branches and  $M \in \mathbb N$ is the average number of time units at which $\mathbf Y[\cdot]$ samples $X[\cdot]$. 
\begin{theorem}[discrete-time multi-branch sampling]
\label{thm:filterbanks_discrete}
For $M\in \mathbb N$ and $p=1,\ldots,P$, let $Y_p[\cdot]$ be a decimation by a factor of $PM$ of the process $Z_p[\cdot]$, namely,
\begin{align*}
\mathbf{Y}\left[n\right] & =\left(Z_1[PMn],\ldots,Z_P[PMn]\right),
\end{align*}
where $X[\cdot]$ and $Z_{p}[\cdot]$ are jointly Gaussian stationary processes with spectral densities 
\[
S_{Z_{p}}\left(e^{2\pi i\phi}\right)=S_{X+\eta}\left(e^{2\pi i\phi}\right)\left|H_p\left(e^{2\pi i\phi}\right)\right|^{2},
\]
and
\[
S_{XZ_{p}}\left(e^{2\pi i\phi}\right)=S_X\left(e^{2\pi i\phi}\right)H_p^{*}\left(e^{2\pi i\phi}\right).
\]
The iDRF of the process $X[\cdot]$ given $\mathbf{Y}[\cdot]=\left(Y_1[\cdot],\ldots,Y_{P}[\cdot]\right)$,
is 
\begin{align*}
R\left(P,M,\theta\right)=
\frac{1}{2}\sum_{p=1}^{P}&\int_{-\frac{1}{2}}^{\frac{1}{2}}\log^{+}\left[\lambda_{p}\left(\mathbf J_{M}\left(e^{2\pi i\phi}\right)\right)\theta^{-1}\right]d\phi
\end{align*}
\begin{align*}
D\left(P,M,\theta\right)&  =\mmse_{X|\mathbf Y}\\ & ~~~~ + \sum_{p=1}^{P}\int_{-\frac{1}{2}}^{\frac{1}{2}}\min\left\{ \lambda_{p}\left(\mathbf J_M\left(e^{2\pi i\phi}\right) \right),\theta\right\} d\phi,\\
  &=\sigma_{X}^{2}-\sum_{p=1}^{P}\int_{-\frac{1}{2}}^{\frac{1}{2}}\left[\lambda_{p}\left(\mathbf J_{M}\left(e^{2\pi i\phi}\right)\right)-\theta\right]^{+}d\phi,
\end{align*}
where $\lambda_{1}\left(\mathbf J_M\left(e^{2\pi i\phi}\right)\right)\leq...\leq\lambda_P\left(\mathbf  J_M\left(e^{2\pi i\phi}\right)\right)$
are the eigenvalues of the $P\times P$ matrix
\begin{equation} \label{eq:J_M_def}
\mathbf J_M\left(e^{2\pi i\phi}\right)\triangleq\  {\mathbf S_{\mathbf Y}}^{-\frac{1}{2}*}\left(e^{2\pi i\phi}\right)\mathbf{K}_{M}\left(e^{2\pi i\phi}\right)\mathbf{S_{\mathbf Y}}^{-\frac{1}{2}}\left(e^{2\pi i\phi}\right).
\end{equation}
Here $\mathbf S_{\mathbf Y}\left(e^{2\pi i \phi} \right)$ is the PSD matrix of the process $\mathbf Y[\cdot]$ and is given by
\begin{align*}
\left(\mathbf S_{\mathbf Y}\left(e^{2\pi i\phi}\right)\right)_{i,j}& \triangleq\frac{1}{MP}\sum_{r=0}^{MP-1}S_{Z_{i}Z_{j}}\left(e^{2\pi i\frac{\phi-r}{MP}}\right) \\
& = \frac{1}{MP}\sum_{r=0}^{MP-1}\left\{S_{X+\eta}H_{i}^{*}H_{j}\right\}\left(  e^{2\pi i\frac{\phi-r}{MP}}\right),
\end{align*}
and $\mathbf S_{\mathbf Y}^{\frac{1}{2}}\left(e^{2\pi i\phi}\right)$ is such
that $\mathbf S_{\mathbf Y}\left(e^{2\pi i\phi}\right)={\mathbf S_{\mathbf Y}}^{\frac{1}{2}}\left(e^{2\pi i\phi}\right)\mathbf S_{\mathbf Y}^{\frac{1}{2}*}\left(e^{2\pi i\phi}\right)$. The $(i,j)^{\textrm{th}}$ entry of the $P\times P$ matrix $\mathbf{K}_M\left(e^{2\pi i\phi}\right)$ is given by 
\begin{align*}
\left(\mathbf{K}_M \right)_{i,j}\left(e^{2\pi i\phi}\right) & \triangleq\frac{1}{(MP)^2}\sum_{r=0}^{MP-1}\left\{S_{X}^{2}H_i^{*}H_j\right\}\left(e^{2\pi i\frac{\phi-r}{MP}}\right).\\
\end{align*}
\end{theorem}
\subsubsection*{Remark}
The case where the matrix $\mathbf S_{\mathbf Y}\left(e^{2\pi i\phi}\right)$ is not invertible for some $\phi \in \left(-\frac{1}{2},\frac{1}{2}\right)$ corresponds to linear dependency between the spectral components of the vector $\mathbf Y[\cdot]$. In this case, we can apply the theorem to the process $\mathbf Y'[\cdot]$ which is obtained from $\mathbf Y[\cdot]$ by removing linearly dependent components.\\
\begin{proof}
The proof is a multi-dimensional extension of the proof of Theorem~\ref{thm:discrete_decimation_rate_distortion}. Details are provided in Appendix~\ref{sec:proof_filtersbank_discrete}.
\end{proof}

\subsection{Main result: multi-branch sampling}
\begin{theorem}[filter-bank sampling]
\label{thm:filterbanks_main} For each $p=1,\ldots,P$, let $Z_p(\cdot)$ be the process obtained by passing a Gaussian stationary source $X(\cdot)$ corrupted by a Gaussian stationary noise $\eta(\cdot)$ through an LTI system $H_p$. Let
$Y_p[\cdot]$, be the samples of the process $Z_p(\cdot)$ at frequency $f_s/P$, namely
\[ 
Y_p[n]=Z_p(nP/f_s)=h_p*\left(X+\eta\right)(nP/f_s),\quad p=1,\ldots,P.
\]
The indirect distortion-rate function of $X(\cdot)$ given $\mathbf Y[\cdot]=\left(Y_1[\cdot],\ldots,Y_P[\cdot]\right)$, is given by
 
\begin{subequations}
\label{eq:multi_main}
\begin{align}
\label{eq:multi_rate}
R\left(\theta\right)=\frac{1}{2}\sum_{p=1}^{P}\int_{-\frac{f_{s}}{2}}^{\frac{f_{s}}{2}}\log^{+}\left[\lambda_{p}\left(\widetilde{\mathbf S}_{X|\mathbf Y}(f) \right)-\theta\right]df
\end{align}
\begin{align}
\label{eq:multi_dist}
D\left(f_s,\theta\right) & =\mmse_{X|\mathbf  Y}(f_s)+ \sum_{p=1}^{P}\int_{-\frac{f_{s}}{2}}^{\frac{f_{s}}{2}}\min\left\{ \lambda_{p}\left(\widetilde{\mathbf S}_{X|\mathbf Y}(f)\right),\theta\right\} df,\nonumber \\
 & =\sigma_{X}^{2}-\sum_{p=1}^{P}\int_{-\frac{f_{s}}{2}}^{\frac{f_{s}}{2}}\left[\lambda_{p}\left(\widetilde{\mathbf S}_{X|\mathbf Y}(f)\right)-\theta\right]^{+}df, 
\end{align}
\end{subequations}
where $\lambda_1\left(\widetilde{\mathbf S}_{X|\mathbf Y}(f)\right)\leq...\leq\lambda_P\left(\widetilde{\mathbf S}_{X|\mathbf Y}(f)\right)$
are the eigenvalues of the $P\times P$ matrix 
\[
\widetilde{\mathbf S}_{X|\mathbf Y}(f)=\tilde{\mathbf S}_{\mathbf Y}^{-\frac{1}{2}*}(f)\mathbf{K}(f)\tilde{\mathbf S}_{\mathbf Y}^{-\frac{1}{2}}(f),
\]
and the $(i,j)^{\textrm{th}}$ entry of the matrices $\tilde{\mathbf S}_{\mathbf Y}(f),{\mathbf K}(f) \in \mathbb C^{P\times P}$ are given by 
\begin{align*}
\left(\tilde{\mathbf S}_{\mathbf Y}\right)_{i,j}(f) & =\sum_{k\in\mathbb{Z}}\left\{ S_{X+\eta}H_{i}H^*_{j}\right\} \left(f-f_{s}k\right),
\end{align*}
 and 
\[
\mathbf{K}_{i,j}(f)=\sum_{k\in\mathbb{Z}}\left\{ S_{X}^{2}H_{i}H^*_{j}\right\} \left(f-f_{s}k\right).
\]
\end{theorem}

\begin{proof}
A full proof can be found in Appendix~\ref{sec:proof_of_filtersbank_main}.
The idea is similar to the proof of Theorem \ref{thm:main_result}:
approximate the continuous time processes $X\left(\cdot\right)$
and $Z\left(\cdot\right)$ by discrete time processes, then take
the limit in the discrete counterpart of the problem given by Theorem
\ref{thm:filterbanks_discrete}.
\end{proof}

\subsection{Optimal pre-sampling filter bank \label{subsec:optimal_multi_branch}}
A similar analysis as in the case of single branch sampling will show that for a fixed $R$, the distortion is a non-increasing function of the eigenvalues of $\widetilde{\mathbf S}_{X|Y}(f)$. This implies that the optimal pre-sampling filters 
$H_1^\star(f),\ldots,H_P^\star(f)$ that minimize the distortion for a given $R$ and $f_s$ are the same filters that minimize the MMSE  in the estimation of $X(\cdot)$ from the samples $\mathbf Y[\cdot]=\left(Y_1[\cdot],\ldots,Y_P[\cdot]\right)$ at sampling frequency $f_s$, given in Theorem~\ref{thm:mmse_opt_filters_bank}. Therefore, the following theorem applies:

\begin{theorem}
\label{thm:opt_filters_bank}
Given $f_s>0$, the optimal pre-sampling filters $H^\star_1(f),\ldots,H^\star_P(f)$ that minimize $D(P,f_s,R)$, for all $R\geq0$, are given by 
\begin{equation}
H_{p}^{\star}(f)=\begin{cases}
1 & f\in F_p^\star,\\
0 & f\notin F_p^\star,
\end{cases}\label{eq:filterbanks_discrete_Hdef_dist},\quad p=1,\ldots,P,
\end{equation}
where $F_1^\star,\ldots,F_P^\star$ satisfy conditions $(i)$ and $(ii)$ in Theorem~\ref{thm:mmse_opt_filters_bank}. The minimal distortion-rate function obtained this way is given by
\begin{subequations} \label{eq:multi_branch_optimal}
\begin{align} \label{eq:multi_branch_optimal_R}
R^\star\left(P,f_s,\theta\right)=\frac{1}{2}\sum_{p=1}^{P}\int_{F_p^\star} \log^{+}\left[\frac{S_{X}^{2}(f)}{S_{X+\eta}(f)}-\theta\right]df
\end{align}
\begin{align} \label{eq:multi_branch_optimal_D}
D^\star\left(P,f_s,\theta\right) & = \mmse_{X|\mathbf Y}^\star(f_s)+\sum_{p=1}^{P}\int_{F_{p}^\star}\min\left\{ \frac{S_{X}^{2}(f)}{S_{X+\eta}(f)},\theta\right\} df, \nonumber \\
 & =\sigma_{X}^{2} -\sum_{p=1}^{P}\int_{F_{p}^\star}\left[\frac{S_{X}^{2}(f)}{S_{X+\eta}(f)}-\theta\right]^{+}df. 
\end{align}
\end{subequations}
\\
\end{theorem}

\begin{proof}
The filters $H_1^\star(f),\ldots,H_P^\star(f)$ given by Theorem~\ref{thm:mmse_opt_filters_bank} maximize the eigenvalues of the matrix $\widetilde{\mathbf S}_{X|Y}(f)$ of \eqref{eq:J_mat_def} for every $f\in \left(-f_s/2,f_s/2\right)$. Since $D\left(P,f_s,R\right)$ is monotone non-increasing in these eigenvalues,  $H_1^\star(f),\ldots,H_P^\star(f)$ also minimize $D\left(P,f_s,R\right)$. For this choice of $H_1(f),\ldots,H_P(f)$, \eqref{eq:multi_main} reduces to \eqref{eq:multi_branch_optimal}.\\
\end{proof}

\subsection{Increasing the number of sampling branches}
We have seen in Theorem~\ref{thm:mmse_Landau} that minimizing the MMSE in sub-Nyquist sampling at frequency $f_s$ is equivalent to choosing a set of frequencies $\mathcal F^\star$ with $\mu(\mathcal F^\star)\leq f_s$ such that
\begin{equation}\label{eq:f_star}
\int_{{\mathcal F^\star}} \frac{S_X^2(f)}{S_{X+\eta}(f)}df  = 
\sup_{\mu({ F})\leq f_s} \int_F \frac{S_X^2(f)}{S_{X+\eta}(f)}df. 
\end{equation}
As in the case of Subsection~\ref{subsec: optimal_sampling} we see that for a given $R$ and $f_s$, by multi-branch uniform sampling we cannot achieve distortion lower than 
\begin{align} \label{eq:D_optimal_bound} 
 D_X^\dagger\left(f_s,R(\theta)\right)\triangleq\sigma_X^2-\int_{\mathcal F^\star} \left[\frac{S_X^2(f)}{S_{X+\eta}(f)}-\theta\right]^+ df,
\end{align}
where $\theta$ is determined by
\begin{align} \label{eq:R_optimal_bound}
R=\int_{\mathcal F^\star}  \log^+\left[\frac{S_X^2(f)}{S_{X+\eta}(f)}\theta^{-1} \right]df.
\end{align}
This is because Proposition~\ref{prop:min_max} asserts that in a parametric reverse water-filling representation of the form \eqref{eq:multi_branch_optimal}, an increment in 
\[
\int_{\bigcup_{i=1}^P \mathcal F^\star_p} \frac{S_X^2(f)}{S_{X+\eta}(f)}df
\]
reduces distortion. But for any $P$,  $\mu\left( {\bigcup_{i=1}^P F^\star_p}\right)\leq f_s$ so we conclude that $D_X^\dagger(f_s,R)\leq D^\star(P,f_s,R)$. The following theorem shows that $D_X^\dagger(f_s,R)$ can be achieved using enough sampling branches.
\begin{theorem} \label{thm:optimal_linear}
For any $f_s>0$ and $\epsilon>0$, there exists $P\in \mathbb N$ and a set of LTI filters $H^\star_1(f),\ldots,H^\star_P(f)$ such that using $P$ uniform sampling branches we have
\begin{subequations} \label{eq:optimal_linear}
\begin{align}\label{eq:optimal_linear_D}
D^\star(P,f_s,R) -\epsilon < \sigma_{X}^2-\int_{\mathcal F^\star} \left[\frac{S_X^2(f)}{S_{X+\eta}(f)}-\theta\right]^+df,
\end{align}
where $\theta$ is determined by
\begin{align}\label{eq:optimal_linear_R}
R=\int_{\mathcal F^\star}  \log^+\left[\frac{S_X^2(f)}{S_{X+\eta}(f)}\theta^{-1} \right]df,
\end{align}
\end{subequations}
and $\mathcal F^\star$ is defined by \eqref{eq:f_star}.
\end{theorem}
\begin{proof}
In Theorem~\ref{thm:mmse_Landau} we found a set of pre-sampling filters $H_1^\star(f),\ldots,H_P^\star(f)$ such that 
\[
\mmse_{X|\mathbf Y}^\star(f_s) - \epsilon < \sigma_X^2-\int_{\mathcal F^\star} \frac{S_X^2(f)}{S_{X+\eta}(f)}df.
\]
Since 
\[
\mmse_{X|\mathbf Y}^\star(f_s) = \sigma_X^2-\sum_{p=1}^P\int_{F_p^\star}\frac{S_X^2(f)}{S_{X+\eta}(f)}df, 
\]
where for $p=1,\ldots,P$, $H_p^\star(f) = \mathbf 1_{F^\star_p}(f)$, we conclude that
\[
\int_{\bigcup_{p=1}^PF_p^\star}  \frac{S_X^2(f)}{S_{X+\eta}(f)}df + \epsilon >  \int_{\mathcal F^\star}\frac{S_X^2(f)}{S_{X+\eta}(f)}df.
\]
By Proposition~\ref{prop:min_max},  maximizing $\sum_{p=1}^P \frac{S_X^2(f)}{S_{X+\eta}(f)}$ minimizes the distortion, so the distortion $D_X^\star(P,f_s,R)$ obtained by using $H_1^\star(f),\ldots,H_P^\star(f)$ is arbitrarily close to $D_X^\dagger(f_s,R)$. 
\end{proof}
An immediate corollary of Theorems \ref{thm:optimal_linear} and \ref{thm:opt_filters_bank} is
\[
\lim_{P\rightarrow \infty} D^\star(P,f_s,R) = D_X^\dagger(f_s,R),
\]
where $D^\dagger(f_s,R)$ is defined in \eqref{eq:D_optimal_bound}. The function $D_X^\dagger(f_s,R)$ is plotted in Fig.~\ref{fig:DistMulti} as a function of $f_s$ for two values of $R$. 

\subsection{Discussion}
The function $D_X^\dagger(f_s,R)$ is monotone in $f_s$ by its definition \eqref{eq:D_optimal_bound}, which is in contrast to $D(P,f_s,R)$ and $D^\star(P,f_s,R)$ that are not guaranteed to be monotone in $f_s$ as the example in Fig.~\ref{fig:DistMulti} shows. Fig.~\ref{fig:DistMulti} also suggests that multi-branch sampling can significantly reduce distortion for a given sampling frequency $f_s$ and source coding rate $R$ over single-branch sampling. Moreover, Theorem~\ref{thm:optimal_linear} shows that multi-branch sampling can achieve the bound $D_X^\dagger(f_s,R)$ with a sufficiently large number of sampling branches. Since having fewer branches is more appealing from a practical point of view, it is sometimes desired to use alternative sampling techniques yielding the same performance as uniform multi-branch sampling with less sampling branches. For example, it was noted in \cite{YuxinNonUniform} that a system with a large number of uniform sampling branches can be replaced by a system with fewer branches with a different sampling frequency at each branch, or by a single branch sampler with modulation. Fig. \ref{fig:DistMulti} also raises the possibility of reducing the sampling frequency without significantly affecting performance, as the function $D^\star(P,f_s,R)$ for $P>1$ approximately achieves the asymptotic value of $D^\star(f_{Nyq},R)$ at $f_s \approx f_{Nyq}/3$. 

\begin{figure}
\begin{center}
\begin{tikzpicture}
\node at (-0.07,0) {\includegraphics[trim=0cm 0cm 0cm 0cm, clip=true, scale=0.55] {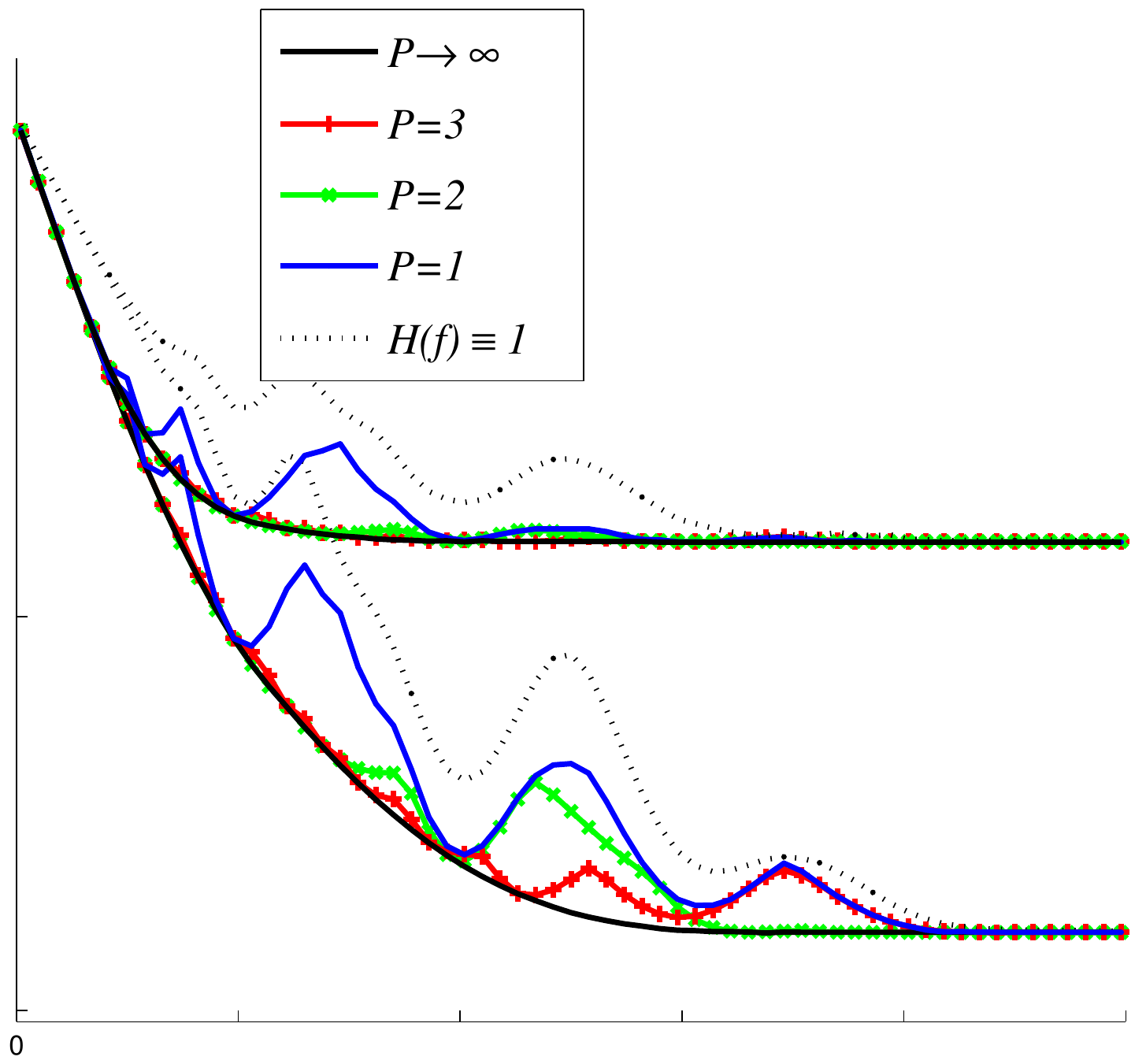}};
\node at (2.3,2.25) { \includegraphics[trim=2.5cm 8.5cm 1.7cm 7cm, clip=true, scale=0.24,frame]{wiggly_spectrum} };
\draw  [fill = white, draw = none] (0.7,-3.87) rectangle (1.2,-4.3);

\draw[->,line width = 2] (-4,-3.5) -- (4,-3.5) node[right] {$f_s$};
\draw[->,line width = 2] (-4,-3.5) -- node[above, rotate = 90, xshift = 0.5cm] {\small $D$} (-4,3.3);

\node at (2.5,3.5) {\small $S_X(f)$};
\draw  (-4.05,3) node[left] {\small $\sigma_X^2$} -- (-3.95,3);

\draw  (-4.05,-3.4) node[left] {\small $\frac{\sigma_X^2}{5}$} -- (-3.95,-3.4);

\draw[dotted]  (2.9,-3.5) node[below] {$f_{Nyq}$} -- (2.9,0.7);

\draw[dashed] (-4,-0.1) -- node[below,xshift = -0.6cm] {\small $D_X(R_0)$} (4,-0.1);
\draw[dashed] (-4,-2.87) -- node[below, xshift = -1cm, yshift = 0] {\small $D_X(4R_0)$} (4,-2.87);

\end{tikzpicture}

\caption{\label{fig:DistMulti}$D^\star(P,f_s,R)$ as a function of $f_s$ for $P=1,2,3$ and two fixed values of the source coding rate $R$. The spectrum of the source is given in the small frame.}
\end{center}
\end{figure}


\section{Conclusions and Future Work \label{sec:Concluding-Remarks}}
We considered a combined sampling and source coding problem, and derived an expression for the indirect distortion-rate function $D(f_s,R)$ of a continuous-time stationary Gaussian process corrupted by noise, given the uniform samples of this process obtained by single branch sampling and multi-branch sampling. By doing so we have generalized and unified the  Shannon-Whittaker-Kotelnikov sampling theorem and Shannon's rate-distortion theory for the important case of Gaussian stationary processes. An optimal design of the sampling structure that minimizes the distortion for a given sampling frequency $f_s$ and any source coding rate $R$ is shown to be the same as the sampling structure that minimizes the MMSE of signal reconstruction under regular sub-Nyquist sampling, i.e., without the bitrate constrained on the samples. This optimal sampling structure extracts the frequency components with the highest SNR. The function $D^{\star}\left(f_{s},R\right)$ associated with the optimal sampling structure is expressed only in terms of the spectral density of the source and the noise. It therefore describes a fundamental trade-off in information theory and signal processing associated with any Gaussian stationary source. \par
Since the optimal design of the sampling structure that leads to $D^\star(f_s,R)$ is tailored for a specific source statistic, it would be interesting to obtain a more universal sampling system which gives optimal performance in the case where the source statistic is unknown and taken from a family of possible distributions. For example, one may consider a `minmax' distortion approach which can be seen as the source coding dual of the channel coding problem considered in \cite{chen2013minimax}. This extension will incorporate signals with an unknown possibly sparse spectral support into our combined sampling and source coding setting. \par
The functions $D(f_s,R)$ and $D^\star(f_s,R)$ fully describe the amount of information lost in uniform sampling of an analog stationary Gaussian process, in the sense that any sampling and quantization scheme with the same constraining parameters must result in a worse distortion in reconstruction. A comparison between the distortion obtained by existing analog to digital conversion (ADC) techniques and the information theoretic bound $D^\star(f_s,R)$ can motivate the search for new ADC schemes or establish the optimality of existing ones. 
More generally, the combined source coding and sampling problem considered in this work can be seen as a source coding problem in which a constraint on the code to be a function of samples of the analog source at frequency $f_s$ is imposed. In practical ADC implementation other restrictions such as limited memory at the encoder and causality may apply. In order to understand the information theoretic bounds on such systems, it would be beneficial to extend our model to incorporate such restrictions. In particular, it is interesting to understand which restrictions lead to a non-trivial trade-off between the average number of bits per second used to represent the process and the sampling frequency of the system.


\appendices

\section{\label{sec:proof_mmse} }
In this Appendix we prove Propositions \ref{prop:mmse_single} and \ref{prop:opt_single}  and their multi-branch counterparts Theorems \ref{thm:mmse_multi} and \ref{thm:opt_filters_bank}. For ease of reading and understanding, we provide different proofs for the single-branch case and multi-branch case, although the former is clearly a special case of the latter.

\subsection{Proof of proposition \ref{prop:mmse_single}: MMSE in single branch sub-Nyquist sampling}
The result is obtained by evaluating \eqref{eq:mmse_average_equiv}. 
We first find $S_{X_\Delta|Y}(e^{2\pi i \phi})  = \frac{\left|S_{X_\Delta Y}\left(e^{2\pi i \phi}\right)\right|^2} {S_{Y}\left(e^{2\pi i \phi}\right)}$. We have 
\begin{align*}
S_{Y}(e^{2\pi i \phi}) & = \sum_{k\in\mathbb Z} S_{X+\eta}\left(f_s(\phi-k)\right)\left|H\left(f_s(\phi-k)\right)\right|^2,
\end{align*}
\begin{align*}
S_{X_\Delta Y}(e^{2\pi i \phi}) & = 
\sum_{l\in \mathbb Z}\mathbb E\left[X\left( \frac{n+l+\Delta}{f_s}\right)Z\left(\frac{n}{f_s}\right) \right]e^{-2\pi i l \phi} \\
 =  \sum_{k\in \mathbb Z}& S_X\left(f_s(\phi-k)\right)H^*\left(f_s(\phi-k)\right)e^{2\pi i k \Delta},
\end{align*}
where we used the fact that the spectral density of the polyphase component $X_\Delta[\cdot]$ equals
\[
S_{X_\Delta}\left(e^{2\pi i \phi}\right)=\sum_{k\in \mathbb Z}S_{X}\left(f_s(\phi-k) \right)e^{2\pi i k \Delta \phi}.
\]
This leads to
\begin{equation}\label{eq:mmse_single_proof}
\begin{split}
& S_{X_\Delta|Y}(e^{2\pi i \phi})  = \frac{\left|S_{X_\Delta Y}\left(e^{2\pi i \phi}\right)\right|^2} {S_{Y}\left(e^{2\pi i \phi}\right)}  \\
& = \frac{ \sum_{k,m\in \mathbb Z} S_XH^* \left(f_s(\phi-k)\right) S_XH\left(f_s(\phi-m)\right) e^{2\pi i (k-m) \Delta}}{\sum_{k\in \mathbb Z}S_Y\left(f_s(\phi-k)\right)}. 
 \end{split}
\end{equation}
Integrating \eqref{eq:mmse_single_proof} over $\Delta$ from $0$ to $1$ gives
 \begin{align}
\frac{\sum_{k\in \mathbb Z}S_X^2\left(f_s(\phi-k)\right)\left|H\left(f_s(\phi-k)\right)\right|^2}{\sum_{k\in \mathbb Z} S_{X+\eta}\left(f_s(\phi-k)\right)\left|H\left(f_s(\phi-k)\right)\right|^2}.
\label{eq:mmse_single_proof2}
\end{align}
Substituting \eqref{eq:mmse_single_proof2} into \eqref{eq:mmse_average_equiv} and changing the integration variable from $\phi$ to $f/f_s$ leads to \eqref{eq:mmse_single_theorem}.

\subsection{Proof of theorem~\ref{prop:opt_single} \label{sec:proof_mmse_opt_single}: optimal pre-sampling filter in single-branch sampling}

Since $\widetilde{ S}_{X|Y}(f) \geq0$, we can maximize the integral over $\widetilde{ S}_{X|Y}(f)$ by maximizing the latter for every $f$ in $\left(-\frac{f_{s}}{2},\frac{f_{s}}{2}\right)$. For a given $f$, denote $h_{k}=\left|H\left(f-f_{s}k\right)\right|^{2}$, $x_{k}=S_{X}^{2}\left(f-f_{s}k\right)$ and $y_{k}=S_{X+\eta}(f-f_sk)=S_{X}\left(f-f_{s}k\right)+S_{\eta}\left(f-f_{s}k\right)$.
We arrive at the following optimization problem 
\begin{align*}
\text{maximize} &  & \frac{\sum_{k\in\mathbb{Z}}x_{k}h_{k}}{\sum_{k\in\mathbb{Z}}y_{k}h_{k}}\\
\text{subject to} &  & h_{k}\geq0,\quad k\in\mathbb{Z}.
\end{align*}
Because the objective function is homogeneous in $\mathbf h=\left(...,h_{-1},h_{0},h_{1},...\right)$,
the last problem is equivalent to 
\begin{eqnarray*}
\text{maximize} &  & \sum_{k\in\mathbb{Z}}x_{k}h_{k}\\
\text{subject to} &  & h_{k}\geq0,\quad k\in\mathbb{Z},\\
 &  & \sum_{k\in\mathbb{Z}}y_{k}h_{k}=1.
\end{eqnarray*}
The optimal value of this problem is $\max_{k}\frac{x_{k}}{y_{k}}$, i.e. the maximal ratio over all pairs $x_{k}$ and $y_{k}$. The optimal $\mathbf h$ is the indicator for the optimal ratio: 
\[
h_{k}^{\star}=\begin{cases}
1 & k\in {\mathrm {argmax}}_{k}\frac{x_{k}}{y_{k}},\\
0 & otherwise.
\end{cases}
\]
If there is more than one $k$ that maximizes $\frac{x_k}{y_k}$, then we can arbitrarily decide on one of them. \par 
Going back to our standard notations, we see that for almost every $f\in\left(-\frac{f_{s}}{2},\frac{f_{s}}{2}\right)$,
the optimal $\widetilde{ S}_{X|Y}(f)$ is given by 
\begin{align*}
\widetilde{ S}_{X|Y}(f)^{\star}\left(f\right) & =  \max_{k\in\mathbb{Z}}\frac{S_{X}^{2}\left(f-f_{s}k\right)}{S_{X+\eta}\left(f-f_{s}k\right)},
\end{align*}
and the optimal $H\left(f\right)$ is such that $\left|H\left(f-f_{s}k\right)\right|^{2}$ is non-zero for the particular $k$ that achieves this maximum. This also implies that $F^\star$, the support of $H^\star(f)$, satisfies properties $(i)$ and $(ii)$ in Definition~\ref{def:aliasing_free}.

\subsection{\label{sec:proof_mmse_filters_bank} Proof of Theorem \ref{thm:mmse_multi}: MMSE in multi-branch sub-Nyquist sampling}

As in the proof of Theorem~\ref{prop:mmse_single} in the previous subsection, the result is obtained by evaluating 
\begin{equation} \label{eq:mmse_multi_proof}
mmse_{X_\Delta|\mathbf Y}(f_s)=\sigma_X^2-\int_{-\frac{1}{2}}^{\frac{1}{2}} \int_0^1 \mathbf S_{X_\Delta|\mathbf Y}\left(e^{2\pi i \phi} \right)d\Delta d\phi.
\end{equation}
Since 
\begin{align*}
 \mathbf C_{X_\Delta\mathbf Y}[k] & =\mathbb E\left[X_\Delta[n+k]\mathbf Y^*[n] \right]\\ & = \left( \mathbf C_{X_\Delta Y_1}[k],\ldots,C_{X_\Delta Y_P}[k] \right),
\end{align*}
we get
\begin{align*}
 \mathbf S_{X_{\Delta}\mathbf Y}\left(e^{2\pi i \phi}\right)=\left( S_{X_{\Delta} Y_1}\left(e^{2\pi i \phi}\right),\ldots,S_{X_{\Delta} Y_P}\left(e^{2\pi i \phi}\right)\right).
\end{align*}
Using $X_\Delta[n]=X\left(\frac{n+\Delta}{f_s}\right)$, for each $p=1,\ldots,P$, we have
\begin{align*}
S_{X_\Delta Y_p}\left(e^{2\pi i \phi}\right) & =\sum_{l\in \mathbb Z} \mathbb E\left[X\left(\frac{n+l+\Delta}{f_s} \right)Z_p\left(\frac{n}{f_s} \right) \right]e^{-2\pi i l \phi} \\
& = \sum_{k\in \mathbb Z}S_X\left(f_s\left(\phi-k\right)\right)H_p^*\left(f_s\left(\phi-k\right)\right)e^{2\pi i k \Delta}.
\end{align*}
In addition, the $\left(p,r\right)^\textrm{th}$ entry of the $P\times P$ matrix $ \mathbf S_{\mathbf Y}\left(e^{2\pi i \phi}\right)$ is given by
\[
\left\{  \mathbf S_{\mathbf Y}\left(e^{2\pi i \phi}\right)\right\}_{p,r} = \sum_{k\in \mathbb Z} \left\{S_{X+\eta}H_p^*H_r\right\} \left(f_s\left(\phi-k\right)\right),
\]
where we have used the shortened notation 
\[
\left\{S_1 S_2 \right\}(x)\triangleq S_1(x)S_2(x)
\] 
for two functions $S_1$ and $S_2$ with the same domain.
It follows that 
\[
 \mathbf S_{X_\Delta|\mathbf Y}\left(e^{2\pi i \phi}\right) = \left\{ \mathbf S_{X_\Delta \mathbf Y}  \mathbf S_{\mathbf Y}^{-1}  \mathbf S^*_{X_\Delta \mathbf Y}\right\}\left(e^{2\pi i \phi}\right)
\]
can also be written as 
\begin{equation}
 \mathbf S_{X_\Delta|\mathbf Y}\left(e^{2\pi i \phi}\right) = \mathrm {Tr}\left\{\mathbf S_{\mathbf Y}^{-\frac{1}{2}*}  \mathbf S_{X_\Delta \mathbf Y}^*  \mathbf S_{X_\Delta \mathbf Y}  \mathbf S_{\mathbf Y}^{-\frac{1}{2}} \right\} \left(e^{2\pi i \phi}\right), \label{eq:mmse_multi_proof1}
\end{equation}
where $ \mathbf S_{\mathbf Y}^{-\frac{1}{2}*}\left(e^{2\pi i \phi}\right)$ is the $P\times P$ matrix satisfying $ \mathbf S_{\mathbf Y}^{-\frac{1}{2}*}\left(e^{2\pi i \phi}\right)  \mathbf S_{\mathbf Y}^{-\frac{1}{2}}\left(e^{2\pi i \phi}\right)= \mathbf S_{\mathbf Y}^{-1}\left(e^{2\pi i \phi}\right)$. \par
The $(p,r)^\textrm{th}$ entry of $ \mathbf S_{X_\Delta \mathbf Y}^*\left(e^{2\pi i \phi}\right)  \mathbf S_{X_\Delta \mathbf Y}\left(e^{2\pi i \phi}\right)$ is given by 
\begin{align*}
& \left\{ \mathbf S_{X_\Delta \mathbf Y}^* \mathbf S_{X_\Delta \mathbf Y} \right\}_{p,r}\left(e^{2\pi i \phi}\right)  =  \sum_{k\in \mathbb Z}\left\{S_XH_p^* \right\}\left(f_s(\phi-k)\right)e^{2\pi i k \Delta} \nonumber \\
& \quad \quad \quad \quad \quad \times \sum_{l\in \mathbb Z}\left\{S_XH_r \right\}\left(f_s(\phi-l) \right) e^{-2\pi i  l \Delta}  \nonumber \\
= & \sum_{k,l\in \mathbb Z}\left[ \left\{S_XH_p^* \right\}\left(f_s(\phi-k) \right)
 \left\{S_XH_r \right\}\left(f_s(\phi-l) \right) e^{2\pi i \Delta(k-l)} \right],
\end{align*}
which leads to
\begin{align*}
\int_0^1 \left\{  \mathbf S_{X_\Delta \mathbf Y}^*  \mathbf S_{X_\Delta \mathbf Y} \right\}_{p,r}\left(e^{2\pi i \phi}\right) d\Delta = \sum_{k\in\mathbb Z} \left\{ S_X^2 H_p^* H_r \right\} \left(f_s(\phi-k)\right).
\end{align*}
From this we conclude that integrating \eqref{eq:mmse_multi_proof1} with respect to $\Delta$ from $0$ to $1$ results in  
\[
\mathrm {Tr}\left\{  \mathbf S_{\mathbf Y}^{-\frac{1}{2}*} \bar{\mathbf K}  \mathbf S_{\mathbf Y}^{-\frac{1}{2}} \right\} \left(e^{2\pi i \phi}\right),
\]
where $\bar{\mathbf K}\left(e^{2\pi i} \right)$ is the $P\times P$ matrix given by
\[
\bar{\mathbf K}_{p,r}\left(e^{2\pi i\phi} \right) = \sum_{k\in\mathbb Z} \left\{ S_X^2 H_p^* H_r \right\} \left(f_s(\phi-k)\right).
\]
The proof is completed by changing the integration variable in \eqref{eq:mmse_multi_proof} from $\phi$ to $f=\phi f_s$, so $S_{\mathbf Y}\left(e^{2\pi i \phi}\right)$ and $\bar{\mathbf K}\left(e^{2\pi i \phi}\right)$ are replaced by $\tilde{\mathbf{S}}_{\mathbf Y}(f)$ and $\mathbf K(f)$, respectively.

\subsection{Proof of theorem \ref{thm:mmse_opt_filters_bank}: optimal filter-bank in multi-branch sampling}
Let $\mathbf H(f) \in \mathbb C^{\mathbb Z \times P}$ be the matrix with $P$ columns of infinite length defined by
\[
\mathbf H(f) = \begin{pmatrix}
\vdots & \vdots & \cdots & \vdots \\
H_1(f-2f_s) & H_2(f-2f_s) & \cdots & H_P(f-2f_s) \\
H_1(f-f_s) & H_2(f-f_s) & \cdots & H_P(f-f_s) \\
H_1(f) & H_2(f) & \cdots & H_P(f) \\
H_1(f+f_s) & H_2(f+f_s) & \cdots & H_P(f+f_s) \\
H_1(f+2f_s) & H_2(f+2f_s) & \cdots & H_P(f+2f_s) \\
\vdots & \vdots & \cdots & \vdots \\ 
\end{pmatrix}.
\]
In addition, denote by $\mathbf S(f) \in \mathbb R^{\mathbb Z \times \mathbb Z}$  and $\mathbf S_n(f) \in \mathbb R^{\mathbb Z \times \mathbb Z}$ the infinite diagonal matrices with diagonal elements  $\left\{\mathbf S_X (f-f_s k),~ k\in \mathbb Z \right\}$ and $\left\{\mathbf S_{X+\eta} (f-f_s k),~ k\in \mathbb Z \right\}$, respectively. With this notation we can write 
\[
\widetilde{\mathbf S}_{X|Y}(f) = \left(\mathbf H^* \mathbf S_n \mathbf H \right)^{-\frac{1}{2}*} \mathbf H^* \mathbf S^2 \mathbf H \left(\mathbf H^* \mathbf S_n \mathbf H \right)^{-\frac{1}{2}},
\]
where we suppressed the dependency on $f$ in order to keep notation neat. Denote by $\mathbf H^\star(f)$ the matrix $\mathbf H(f)$ that corresponds to the filters $H^\star(f),\ldots,H^\star(f)$ that satisfy conditions $(i)$ and $(ii)$ in Theorem~\ref{thm:mmse_opt_filters_bank}. By part (iii) of the remark at the end of Theorem~\ref{thm:mmse_opt_filters_bank}, the structure of $\mathbf H^\star(f)$ can be described as follows: each column has a single non-zero entry, such that the first column indicates the largest among $\left\{\frac{S^2_X(f-f_sk)}{S_{X+\eta}(f-f_sk)},\,k\in \mathbb Z \right\}$, which is the diagonal of $\mathbf S(f) \mathbf S_n^{-1}(f) \mathbf S(f)$. The second column corresponds to the second largest entry of  $\left\{\frac{S^2_X(f-f_sk)}{S_{X+\eta}(f-f_sk)},\,k\in \mathbb Z \right\}$, and so on for all $P$ columns of $\mathbf H^\star(f)$. This means that $\widetilde{\mathbf  S}_{X|\mathbf Y}^\star(f)$ is a $P\times P$ diagonal matrix whose non-zero entries are the $P$ largest values among $\left\{\frac{S^2_X(f-f_sk)}{S_{X+\eta}(f-f_sk)},\,k\in \mathbb Z \right\}$, i.e, $\lambda_p \left( \widetilde{\mathbf  S}_{X|\mathbf Y}^\star(f) \right) = J_p^\star(f)$, for all $p=1,\ldots,P$. \par
It is left to establish the optimality of this choice of pre-sampling filters. Since the rank of $\widetilde{\mathbf S}_{X|Y}(f)$ is at most $P$, in order to complete the proof it is enough to show that for any $\mathbf H(f)$, the $P$ eigenvalues of the corresponding $\widetilde{\mathbf S}_{X|\mathbf Y}(f)$ are smaller then the $P$ largest eigenvalues of $\mathbf S(f) \mathbf S_n^{-1}(f) \mathbf S(f)$ compared by their respective order. Since the matrix $ $ entries of the diagonal matrices $\mathbf S(f)$ and $\mathbf S_n(f)$ are positive, the eigenvalues of $\widetilde{\mathbf S}_{X|\mathbf Y}(f)$ are identical to the $P$ non-zero eigenvalues of the matrix 
\[
\mathbf S \mathbf H \left(\mathbf H^* \mathbf S_n \mathbf H \right)^{-1} \mathbf H^* \mathbf S.
\]
It is enough to prove that the matrix
\[
\mathbf S \mathbf S_n^{-1} \mathbf S - \mathbf S \mathbf H \left(\mathbf H^* \mathbf S_n \mathbf H \right)^{-1} \mathbf H^* \mathbf S,
\]
is positive\footnote{In the sense that it defines a positive linear operator on the Hilbert space $\ell_2\left(\mathbb C \right)$.  The linear algebra notation we use here is consistent with the theory of positive operators on Hilbert spaces.}.  This is equivalent to 
\begin{equation} \label{eq:proof_mmse_opt_filterbank1}
\mathbf a^*\mathbf S \mathbf S_n^{-1} \mathbf S \mathbf a - \mathbf a^* \mathbf S \mathbf H \left(\mathbf H^* \mathbf S_n \mathbf H \right)^{-1} \mathbf H^* \mathbf S \mathbf a^* \geq 0,
\end{equation}
for any sequence $\mathbf a \in \ell_2\left(\mathbb C \right)$. By factoring out $\mathbf S \mathbf S_n^{-\frac{1}{2}}$ from both sides, \eqref{eq:proof_mmse_opt_filterbank1} reduces to
\begin{align} \label{eq:proof_mmse_opt_filterbank2}
\mathbf a^* \mathbf  a- \mathbf a^*\mathbf S_n^{\frac{1}{2}} \mathbf H \left(\mathbf H^* \mathbf S_n \mathbf H \right)^{-1} \mathbf H^* \mathbf S_n^{\frac{1}{2}}\mathbf a\geq 0.
\end{align}
The Cauchy-Schwartz inequality implies that
\begin{align} \label{eq:proof_mmse_opt_filterbank3}
& \left(\mathbf a^*\mathbf S_n^{\frac{1}{2}} \mathbf H \left(\mathbf H^* \mathbf S_n \mathbf H \right)^{-1} \mathbf H^* \mathbf S_n^{\frac{1}{2}}\mathbf a \right)^2 \nonumber \\  &\leq \mathbf a^* \mathbf a \times \mathbf a^* \mathbf S_n^{\frac{1}{2}} \mathbf H \left( \mathbf H^* \mathbf S_n \mathbf H\right)^{-1} \mathbf H^*  \mathbf S_n^{\frac{1}{2}} \mathbf S_n^{\frac{1}{2}} \mathbf H \left( \mathbf H^* \mathbf S_n \mathbf H\right)^{-1} \mathbf H^*\mathbf S_n^{\frac{1}{2}} \mathbf a \nonumber \\
& = \mathbf a^* \mathbf a \left( \mathbf a^* \mathbf S_n^{\frac{1}{2}} \mathbf H \left( \mathbf H^* \mathbf S_n \mathbf H\right)^{-1}  \mathbf H^*\mathbf S_n^{\frac{1}{2}} \mathbf a\right).
\end{align}
Dividing \eqref{eq:proof_mmse_opt_filterbank3} by $\left( \mathbf a^* \mathbf S_n^{\frac{1}{2}} \mathbf H \left( \mathbf H^* \mathbf S_n \mathbf H\right)^{-1}  \mathbf H^*\mathbf S_n^{\frac{1}{2}} \mathbf a\right)$ leads to \eqref{eq:proof_mmse_opt_filterbank2}.

\section{Proof of theorem \ref{thm:discrete_decimation_rate_distortion}: distortion-rate function in discrete-time sampling\label{sec:Proof-of-discrete}}

Note that $X\left[\cdot\right]$ and $Y\left[\cdot\right]$ are in general not jointly stationary for $M>1$, and we cannot use the discrete-time version of Theorem~\ref{thm:[Dobrushin-and-Tsybakov]} in \eqref{eq:DnT_discrete} as is. Instead we proceed as follows: For a given $M \in \mathbb N$, define the vector-valued process $\mathbf X^M[\cdot]$ by
\[
\mathbf X^M[n]=\left(X[Mn],X[Mn+1],\ldots,X[Mn+M-1] \right),
\] 
and denote by $X^M_m[\cdot]$ its $m^\textrm{th}$ coordinate, $m=0,\ldots,M-1$. For each $r,m=0,\ldots,M-1$ and $n,k \in \mathbb Z$, the covariance between $X^M_m[n]$ and $X^M_r[k]$ is given by
\[
C_{X^M_mX_r^M}[k]=\mathbb E\left[X^M_m[n+k]X^M_r[n]^* \right]=C_X[Mk+m-r].
\]
This shows that $Y[\cdot]=h[\cdot]*X^M_0[\cdot]$ is jointly stationary with the processes  $X^M[\cdot]$. \par
By properties of multi-rate signal processing (see for example
\cite{52200}),
\[
S_{Y}\left(e^{2\pi i\phi}\right)=\frac{1}{M}\sum_{m=0}^{M-1}S_{Z}\left(e^{2\pi i\frac{\phi-m}{M}}\right),
\]
\begin{align*}
S_{X^M_rX^M_s}\left(e^{2\pi i\phi}\right)
 & = S_{X^M_{r-s}X^M_0}\left(e^{2\pi i\phi}\right) \\
 & = \sum_{k\in\mathbb Z}C_{X}\left[Mk+r-s\right]e^{-2\pi ik\phi}\\
 & =  \frac{1}{M}\sum_{m=0}^{M-1}e^{2\pi i(r-s)\frac{\phi-m}{M}}S_{X}\left(e^{2\pi i\frac{\phi-m}{M}}\right), 
\end{align*}
and
\begin{align*}
S_{X^M_rY}\left(e^{2\pi i\phi}\right) & =  \sum_{k\in\mathbb Z}C_{X^M_rY}\left[k\right]e^{-2\pi ik\phi}\\
 & = \sum_{k\in\mathbb Z}C_{XZ}\left[Mk+r\right]e^{-2\pi ik\phi}\\
 & =  \frac{1}{M}\sum_{m=0}^{M-1}e^{2\pi ir\frac{\phi-m}{M}}S_{XZ}\left(e^{2\pi i\frac{\phi-m}{M}}\right),
\end{align*}
from which we can form the $M\times 1$ matrix 
\[
\mathbf S_{\mathbf{X}^M Y}\left(e^{2\pi i \phi }\right)=\begin{pmatrix}
S_{X^M_{0}Y}(e^{2\pi i \phi)})\\
\vdots\\
S_{X^M_{M-1}Y}(e^{2\pi i \phi})
\end{pmatrix}.
\]
The spectral density of the MMSE estimator of $\mathbf{X}\left[\cdot\right]$
from $Y\left[\cdot\right]$ equals $\mathbf S_{\mathbf{X}Y}S_{Y}^{-1} \mathbf S_{\mathbf{X}Y}^{*}\left(e^{2\pi i\phi}\right)$,
which is a matrix of rank one. Denote its non-zero eigenvalue by $J_{M}\left(e^{2\pi i\phi}\right)$,
which is given by the trace:
\begin{align*}
& J_{M}\left(e^{2\pi i\phi}\right)  =  \mathrm{Tr}~\mathbf S_{\mathbf{X}^M Y}S_{Y}^{-1} \mathbf S_{\mathbf{X}^MY}^{*}\left(e^{2\pi i \phi}\right)\\
& \quad \quad \quad \quad =  \frac{1}{S_{Y}\left(e^{2\pi i \phi}\right)}\sum_{r=0}^{M-1}\left|S_{X^M_rY }\left(e^{2\pi i \phi}\right)\right|^{2}\\
& = \frac{ \frac{1}{M} \sum_{r=0}^{M-1}\sum_{m=0}^{M-1}\sum_{l=0}^{M-1} S_{XZ}\left(e^{2\pi i\frac{\phi-m}{M}}\right)
 S_{XZ}^{*}\left(e^{2\pi i\frac{\phi-l}{M}}\right)e^{-2\pi ir\frac{m-l}{M}}} { \sum_{m=0}^{M-1}S_{Z}\left(e^{2\pi i\frac{\phi-m}{M}}\right) } \\
 &\quad\quad \quad \quad =  \frac{\sum_{m=0}^{M-1}\left|S_{XZ}\left(e^{2\pi i\frac{\phi-m}{M}}\right)\right|^{2}}{\sum_{m=0}^{M-1}S_{Z}\left(e^{2\pi i\frac{\phi-m}{M}}\right)}.
\end{align*}

By Theorem \ref{thm:indirect_vector_case}, the iDRF of $\mathbf{X}^M\left[\cdot\right]$ given $Y\left[\cdot\right]$
is 
\begin{subequations}
\label{eq:proof_discrete}
\begin{align}
\label{eq:proof_discreteR}
R\left(\theta\right) & = \frac{1}{2}\int_{-\frac{1}{2}}^{\frac{1}{2}}\log^{+}\left[J_{M}\left(e^{2\pi i\phi}\right)\theta^{-1}\right]d\phi,
\end{align}
\begin{align}
\label{eq:proof_discreteD}
D_{\mathbf X^M|Y} \left(\theta\right) & = \mmse_{\mathbf{X}^M|Y}+ \frac{1}{M}\int_{-\frac{1}{2}}^{\frac{1}{2}}\min\left\{ J_{M}\left(e^{2\pi i\phi}\right),\theta\right\} d\phi.
\end{align}
\end{subequations}
Note that
\begin{align*}
\mmse_{X|Y}  = & \lim_{N\rightarrow\infty}\frac{1}{2N+1}\sum_{n=-N}^{N}\mathbb{E}\left(X\left[n\right]-\mathbb{E}\left[X\left[n\right]|Y\left[\cdot\right]\right]\right)^{2}\\
= & \frac{1}{M}\sum_{m=0}^M\mathbb{E}\left(X^M_m\left[n\right]- \mathbb{E}\left[X_{m}^M\left[n\right]|Y\left[\cdot\right]\right]\right)^{2}\\
  = & \frac{1}{M} \sum_{m=0}^{M-1} \mmse_{X_{m}|Y}\\
  = & \mmse_{\mathbf{X}^M|Y}. 
\end{align*}

Since $\mathbf {X}^M[\cdot]$ is a stacked version of  $X[\cdot]$, both processes share the same indirect rate-distortion function given $Y\left[\cdot\right]$. Thus, the result is obtained by substituting $ \mmse_{\mathbf {X}^M | Y }$ and $J_M(e^{2\pi i \phi})$ in \eqref{eq:proof_discrete}.  \\

\section{Proof of theorem \ref{thm:main_result}\label{sec:proof:main_result}:  distortion-rate function in single branch sampling}

For each $M=1,2,...$ define $X^M\left[\cdot\right]$ and $Z^M\left[\cdot\right]$
to be the processes obtained by uniformly sampling $X\left(\cdot\right)$
and $Z\left(\cdot\right)$ at frequency $f_{s}M$, i.e. $X^M[n]=X\left(\frac{n}{f_{s}M}\right)$
and $Z^M [n]=Z\left(\frac{n}{f_{s}M}\right)$. The spectral
density of $X^M[\cdot]$ is 
\[
S_{X^M}\left(e^{2\pi i\phi}\right)=Mf_{s}\sum_{k\in\mathbb Z}S_{X}\left(Mf_{s}\left(\phi-k\right)\right).
\]
Using similar considerations as in the proof of Theorem \ref{thm:discrete_decimation_rate_distortion},
we see that $X^M[\cdot]$ and $Z^M[\cdot]$
are jointly stationary processes with cross correlation function
\[
C_{X^M Z^M}\left[k\right]=C_{XZ}\left(\frac{k}{Mf_{s}}\right),
\]
and cross spectral density
\[
S_{X^{M}Z^{M}}\left(e^{2\pi i\phi}\right)=Mf_{s}\sum_{k\in\mathbb Z}S_{XZ}\left(Mf_{s}\left(\phi-k\right)\right).
\]
Note that $Y\left[\cdot\right]$ is a factor-$M$ down-sampled version of $Z^M\left[\cdot\right]$, and the indirect rate-distortion function of $X^{M}\left[\cdot\right]$ given $Y\left[\cdot\right]$ is obtained by Theorem~\ref{thm:discrete_decimation_rate_distortion} as follows:  
\begin{equation}
\bar{R}_{X^M|Y}\left(\theta\right)=\frac{1}{2}\int_{-\frac{1}{2}}^{\frac{1}{2}}\log^{+}\left[\frac{1}{M}J_{M}\left(e^{2\pi i\phi}\right)\theta^{-1}\right]d\phi,\label{eq:Rate_M}
\end{equation}
\begin{align}
D_{X^M|Y}\left(\theta\right)  = &\mmse_{X^M|Y}(M) 
 +\int_{-\frac{1}{2}}^{\frac{1}{2}}\min\left\{ J_{M}\left(e^{2\pi i\phi}\right),\theta\right\} d\phi \nonumber  \\
  = & \sigma_{X^M}^2  -\int_{-\frac{1}{2}}^{\frac{1}{2}}\left[J_{M}\left(e^{2\pi i\phi}\right)-\theta\right]^{+}d\phi. \label{eq:Dist_M}
\end{align}

Since the sampling operation preserves the $L_2$ norm of the signal, we have $\sigma^2_{X^M}=\sigma^2_X$. In our case $J_M\left(e^{2\pi i\phi}\right)$ is obtained by substituting the spectral densities $S_{X^M Z^M}(e^{2\pi i \phi})$ and $S_{Z^M}(e^{2\pi i \phi})$,
\begin{align}
J_{M}\left(e^{2\pi i\phi}\right)  = &  \frac{1}{M}\frac{\sum_{m=0}^{M-1}\left|S_{X^M Z^M}\left(e^{2\pi i\frac{\phi-m}{M}}\right)\right|^{2}}{\sum_{m=0}^{M-1}S_{Z^M}\left(e^{2\pi i\frac{\phi-m}{M}}\right)} \nonumber \\
 = & \frac{f_s \sum_{m=0}^{M-1}\left|\sum_{k\in\mathbb Z}S_{XZ} \left(f_s M\left(\phi-m-Mk\right)\right)\right|^2 } { \sum_{m=0}^{M-1}  \sum_{k\in\mathbb Z}S_Z \left(f_s M\left(\phi-m-Mk\right)\right)}.
\label{eq:K_M_proof_main}
\end{align}
We now take the limit $M\rightarrow\infty$ in \eqref{eq:Rate_M} and \eqref{eq:Dist_M}. Under the assumption of Riemann integrability, the distortion between almost any sample path of $X\left(\cdot\right)$ and any reasonable reconstruction of it from $X^{M}[\cdot]$ (e.g., \emph{sample and hold}) will converge to zero. It follows that the distortion in reconstructing $X^{M}[\cdot]$ form $\hat{Y}[\cdot]$ must also converge to the distortion in reconstructing $X\left(\cdot\right)$ from $\hat{Y}\left[\cdot\right]$, and the indirect distortion-rate function of $X\left(\cdot\right)$ given $Y\left[\cdot\right]$ is obtained by this limit. For a detailed explanation on the converges of the DRF of a sampled source to the DRF of the continuous-time version we refer to \cite{neuhoff2013information}. Thus, all that remains is to show that  
\begin{equation} \label{eq:main_proof_to_show1}
\lim_{M\rightarrow\infty}\int_{-\frac{1}{2}}^{\frac{1}{2}} J_{M}\left(e^{2\pi i\phi}\right)d\phi  = f_s\int_{-\frac{1}{2}}^{\frac{1}{2}}J(f_s\phi)d\phi = \int_{-\frac{f_s}{2}}^{\frac{f_s}{2}}\widetilde{S}_{X|Y}(f)df.
\end{equation}
Denote 
\[
g\left(f\right)\triangleq \sum_{n\in\mathbb Z}\left|S_{XZ}\left(f-f_{s}n\right)\right|^{2},
\] 
\[
h\left(f\right)\triangleq \sum_{n\in\mathbb Z}S_{Z}\left(f-f_{s}n\right),
\]
and
\begin{align*}
g_{M}\left(f\right)\triangleq \sum_{m=0}^{M-1}\left|\sum_{k\in\mathbb Z}S_{XZ}\left(f-f_{s}\left(m-Mk\right)\right)\right|^{2} \\
 = \sum_{m=0}^{M-1} \sum_{k\in\mathbb Z}S_{XZ}^*\left(f-f_{s}\left(m-Mk\right)\right) \sum_{l\in\mathbb Z} S_{XZ}\left(f-f_{s}\left(m-Mk\right)\right).
\end{align*}
Note that since $S_Z(f)$ and $|S_{XZ}(f)|^2/{S_Z(f)}$ are $\mathbf L_1(\mathbb R)$ functions, $g(f)$, $h(f)$ and $g^M(f)$ are almost surely bounded periodic functions. Since the denominator in \eqref{eq:K_M_proof_main} reduces to $\sum_{n\in\mathbb Z}S_{Z}\left(f_{s}\left(\phi-n\right)\right)$, \eqref{eq:main_proof_to_show1} can be written as
\begin{equation} \label{eq:main_proof_to_show2}
\lim_{M\rightarrow \infty} \int_{-\frac{1}{2}}^\frac{1}{2} \frac{g_M(f_s \phi)}{h(f_s\phi)}d\phi=\int_{-\frac{1}{2}}^\frac{1}{2} \frac{g(f_s\phi)}{h(f_s \phi)} d\phi.
\end{equation}

Since the function $h(f)$ is periodic with period $f_s$, we can write the RHS of \eqref{eq:main_proof_to_show2} as
\begin{align}
\int_{-\frac{1}{2}}^\frac{1}{2}\frac{g_M (f_s\phi)}{h(f_s\phi)} d\phi \nonumber 
= & \int_{-\frac{1}{2}}^\frac{1}{2} \sum_{m=0}^{M-1} \left \{ \sum_{k\in\mathbb Z}\frac{S_{XZ}\left(f_s(\phi-m+Mk)\right)}{\sqrt{h\left(f_s(\phi-m+Mk) \right)}} \right. \\
& \times \left.  \sum_{l\in\mathbb Z} \frac{S_{XZ}^*\left(f_s(\phi-m+Mk)\right)}{\sqrt{h\left(f_s(\phi-m+Mk) \right)}} \right\} d\phi. \label{eq:main_proof_before_lemma}
\end{align}
Denoting 
\[
f_1(\phi)=\frac{S_{XZ}\left(\phi f_s\right)}{\sqrt{h\left(\phi f_s \right)}},
\]
and $f_2(\phi)=f_1^*(\phi)$, \eqref{eq:main_proof_to_show2} follows from the following lemma:
\begin{lemma} \label{lem:proof_main}
Let $f_1(\varphi)$ and $f_2(\varphi)$ be two complex valued bounded functions such that $\int_{-\infty}^\infty \left|f_i(\varphi) \right|^2 d\varphi<\infty $, $i=1,2$. Then for any $f_s > 0$,
\begin{equation} \label{eq:main_proof_lemma1}
\int_{-\frac{1}{2}}^{\frac{1}{2}} \sum_{m=0}^{M-1} \sum_{k\in \mathbb Z} f_1\left(\phi+m+kM\right) \sum_{l\in\mathbb Z} f_2\left(\phi+m+lM\right) d\phi
\end{equation}
converges to
\begin{equation} \label{eq:main_proof_lemma2}
\int_{-\frac{1}{2}}^{\frac{1}{2}} \sum_{n\in \mathbb Z} f_1\left(\phi-n\right)  f_2\left(\phi-n\right) d\phi,
\end{equation}
as $M$ goes to infinity.
\end{lemma}

\subsubsection*{Proof of Lemma~\ref{lem:proof_main}}
Equation \eqref{eq:main_proof_lemma1} can be written as
\begin{align}
\label{eq:main_proof_lemma3}
&\int_{-\frac{1}{2}}^{\frac{1}{2}} \sum_{m=0}^{M-1} \sum_{k\in \mathbb Z} f_1\left(\phi+m+kM\right) f_2\left(\phi+m+kM\right) d\phi \\
\label{eq:main_proof_lemma4}
+ &\int_{-\frac{1}{2}}^{\frac{1}{2}} \sum_{m=0}^{M-1} \sum_{k\neq l} f_1\left(\phi+m+kM\right) f_2\left(\phi+m+lM\right) d\phi.
\end{align}
Since the term \eqref{eq:main_proof_lemma3} is identical to \eqref{eq:main_proof_lemma2}, all that is left is to show that \eqref{eq:main_proof_lemma4} vanishes as $M\rightarrow \infty$. Take $M$ large enough such that 
\[
\int_{\mathbb R \setminus [-\frac {M+1}{2},\frac {M+1}{2}]} \left|f_i(\phi)\right|^2 d\phi< \epsilon^2, \quad i=1,2.
\]
We can assume this $M$ is even without losing generality. By a change of variables  \eqref{eq:main_proof_lemma4} can be written as
\begin{equation} \label{eq:main_proof_lemma4a}
\sum_{k\neq l}\int_{-\frac{M+1}{2}}^\frac{M+1}{2}  f_1\left(\varphi+\frac{M}{2}+kM\right) f_2\left(\varphi+\frac{M}{2}+lM\right) d\varphi.
\end{equation}
We split the indices in the last sum into three disjoint sets:
\begin{enumerate}
\item $\mathcal I=\left\{k,l \in \mathbb Z\setminus\{0,-1\},\,k\neq l \right\}$,
\begin{align} 
&\left|\sum_{ \mathcal I} \int_{-\frac{M+1}{2}}^\frac{M+1}{2} f_1\left(\varphi+\frac{M}{2}+kM\right)
 f_2\left(\varphi+\frac{M}{2}+lM\right) d\varphi \right| \nonumber \\ \nonumber
&\overset{a}{\leq} \sum_{ \mathcal I} \int_{-\frac{M+1}{2}}^\frac{M+1}{2} \left|f_1\left(\varphi+\frac{M}{2}+kM\right)\right|^2 d\varphi \\ \nonumber
&\quad \quad  + \sum_{ \mathcal I} 
\int_{-\frac{M+1}{2}}^\frac{M+1}{2} \left|
 f_2\left(\varphi+\frac{M}{2}+lM\right) \right|^2 d\varphi  \\ \nonumber
 & \leq \int_{\mathbb R\setminus {[-\frac{M+1}{2},\frac{M+1}{2}]} } \left|f_1(\varphi)\right|^2 d\varphi \\ 
 &\quad \quad +\int_{\mathbb R\setminus {[-\frac{M+1}{2},\frac{M+1}{2}]} } \left|f_2(\varphi)\right|^2 d\varphi   \leq 2\epsilon^2, \label{eq:main_proof_lemma5}
\end{align}
where $(a)$ is due to the triangle inequality and since for any two complex numbers $a,b$, $|ab|\leq \frac{|a|^2+|b|^2}{2}\leq |a|^2+|b|^2$.
\item $k=0,l=-1$,
\begin{align} 
&\int_{-\frac{M+1}{2}}^\frac{M+1}{2}  f_1\left(\varphi+\frac{M}{2}\right) f_2\left(\varphi-\frac{M}{2}\right) d\varphi \nonumber \\ \nonumber & = \int_{-\frac{M+1}{2}}^0  f_1\left(\varphi+\frac{M}{2}\right) f_2\left(\varphi-\frac{M}{2}\right) d\varphi \\ \nonumber &\quad \quad +\int_0^{\frac{M+1}{2}} f_1\left(\varphi+\frac{M}{2}\right) f_2\left(\varphi-\frac{M}{2}\right) d\varphi \\ \nonumber & \overset{a}{\leq} \sqrt{\int_{-\frac{M+1}{2}}^0  f_1^2\left(\varphi+\frac{M}{2}\right)  d\varphi} \sqrt{\int_{-\frac{M+1}{2}}^0  f_2^2\left(\varphi-\frac{M}{2}\right)  d\varphi} \\ \nonumber
& \quad + \sqrt{\int_0^{\frac{M+1}{2}}  f_1^2\left(\varphi+\frac{M}{2}\right)  d\varphi} \sqrt{\int_0^{\frac{M+1}{2}}  f_2^2\left(\varphi-\frac{M}{2}\right)  d\varphi} \\ \nonumber
& \leq \sqrt{\int_{-\frac{M+1}{2}}^0  f_1^2\left(\varphi+\frac{M}{2}\right)  d\varphi} \sqrt{ \int_{\mathbb R \setminus [-\frac {M+1}{2},\frac {M+1}{2}]} f_2^2\left(\phi\right) d\phi } \\ \nonumber
& \quad + \sqrt{ \int_{\mathbb R \setminus [-\frac {M+1}{2},\frac {M+1}{2}]} f_1^2\left(\phi\right) d\phi } \sqrt{\int_0^{\frac{M+1}{2}}  f_2^2\left(\varphi-\frac{M}{2}\right)  d\varphi} \\ 
& \leq \epsilon \|f_1\|_2 + \epsilon \|f_2\|_2, \label{eq:main_proof_lemma6}
\end{align}
where $(a)$ follows from the Cauchy-Schwartz inequality.
\item $k=-1,l=0$, using the same arguments as in the previous case,
\begin{align} \label{eq:main_proof_lemma7}
\int_{-\frac{M+1}{2}}^\frac{M+1}{2}&  f_1\left(\varphi+\frac{M}{2}\right) f_2\left(\varphi-\frac{M}{2}\right) d\varphi \nonumber \\ &\quad \quad \leq \epsilon\left(\|f_1\|_2+\|f_2\|_2 \right). 
\end{align}
\end{enumerate} 

From \eqref{eq:main_proof_lemma5}, \eqref{eq:main_proof_lemma6} and \eqref{eq:main_proof_lemma7}, the sum \eqref{eq:main_proof_lemma4a} can be bounded by
\[
2\epsilon\left(\|f_1\|_2+\|f_2\|_2 \right)+2\epsilon^2,
\]
which can be made as close to zero as required.
Since \eqref{eq:main_proof_to_show2} follows from \eqref{eq:main_proof_before_lemma} and Lemma~\ref{lem:proof_main}, the proof is complete.

\section{Proof of Theorem~\ref{thm:filterbanks_discrete}: Discrete Multi-branch Sampling
\label{sec:proof_filtersbank_discrete} }
Similar to the proof of Theorem~\ref{thm:discrete_decimation_rate_distortion}, the iDRF of $X[\cdot]$ given $\mathbf Y[\cdot]$ coincides with the iDRF of the vector-valued process $\mathbf X^{PM}[\cdot]$ defined by
\[
{\mathbf X}^{MP}[n]=\left(X[PMn],X[PMn+1],\ldots X[PMn+PM-1] \right).
\]
For a given $M\in \mathbb N$. ${\mathbf X}^{MP}[\cdot]$ is a stationary Gaussian process with PSD matrix 
\[
\left( \mathbf {S_X}\right)_{r,s}\left(e^{2\pi i \phi} \right)=\frac{1}{MP}\sum_{m=0}^{MP-1}e^{2\pi i (r-s)\frac{\phi-m}{PM}}S_X\left(e^{2\pi i \frac{\phi-m}{MP}} \right).
\]
The processes $\mathbf Y[\cdot]$ and $\mathbf X^{PM}[\cdot]$ are jointly Gaussian and stationary with a $PM \times P$ cross PSD  whose $(m+1,p)^\textrm{th}$ entry is given by
\begin{align*}
\left(\mathbf S_{\mathbf X^{PM}\mathbf Y}\right)_{m,p}\left(e^{2\pi i \phi} \right) & = S_{\mathbf X_m^{PM}Y_p}\left(e^{2\pi i\phi}\right) \\
& = \sum_{k\in \mathbb Z} \mathbb E\left[X[PMk+m] Z_p[0] \right]e^{-2\pi i \phi k } \\
& = \frac{1}{PM} \sum_{r=0}^{PM-1}e^{2\pi i m \frac{\phi-r}{PM}} S_{XZ_p}\left(e^{2\pi i \frac{\phi-r}{PM}} \right),
\end{align*}
where we denoted by $X^{PM}_m$ the $m^\textrm{th}$ coordinate of $\mathbf X^{PM}[\cdot]$. The PSD of the MMSE estimator of $\mathbf X^{PM}[\cdot]$ from $\mathbf Y[\cdot]$ is given by
\begin{align} \label{eq:proof_multi_discrete_CPSD}
\mathbf S_{\mathbf X^{PM}|\mathbf Y}\left(e^{2\pi i \phi} \right) & = \left\{\mathbf S_{\mathbf X^{MP} \mathbf Y} \mathbf S^{-1}_{\mathbf Y} \mathbf S_{\mathbf X^{MP} \mathbf Y}^* \right\} \left(e^{2\pi i \phi} \right),
\end{align}
Since only the non-zero eigenvalues of  $\mathbf S_{\mathbf X^{PM}|\mathbf Y}\left(e^{2\pi i \phi} \right)$ contribute to the distortion in \eqref{eq:vector_RD_D}, we are interested in the non-zero eigenvalues of \eqref{eq:proof_multi_discrete_CPSD}. These are identical to the non-zero eigenvalues of 
\begin{equation} \label{eq:proof_multi_discrete_eigen}
\left\{\mathbf S_{\mathbf Y}^{-\frac{1}{2}*}\mathbf S_{\mathbf X^{PM} \mathbf Y}^*\mathbf S_{\mathbf X^{PM} \mathbf Y} \mathbf S_{\mathbf Y}^{-\frac{1}{2}}\right\} \left(e^{2\pi i \phi} \right),
\end{equation}
where $\left\{\mathbf S_{\mathbf Y}^{-\frac{1}{2}*}\mathbf S_{\mathbf Y}^{-\frac{1}{2}}\right\}\left(e^{2\pi i \phi} \right)=\mathbf S_{\mathbf Y}^{-1}\left(e^{2\pi i \phi} \right)$. The $(p,q)^\textrm{th}$ entry of the  $P\times P$ matrix  $\left\{ \mathbf S_{\mathbf X^{PM} \mathbf Y}^*\mathbf S_{\mathbf X^{PM} \mathbf Y}\right\} \left(e^{2\pi i \phi} \right)$ is given by 
\begin{align*}
&\frac{1}{(PM)^2} \sum_{l=0}^{PM-1} \sum_{r=0}^{PM-1} e^{-2\pi i l \frac{\phi-r}{PM} } S_{XZ_p} \left(e^{2\pi i \frac{\phi-r}{PM}} \right)  \\& ~~~~~~  \times \sum_{k=0}^{PM-1} e^{2\pi i l \frac{\phi-k}{PM} } S_{XZ_q} \left(e^{2\pi i \frac{\phi-k}{PM}} \right) \\
= & \frac{1}{(PM)^2}  \sum_{r=0}^{PM-1}  S_{XZ_p} \left(e^{2\pi i \frac{\phi-r} {PM}} \right)  S_{XZ_q} \left(e^{2\pi i \frac{\phi-r} {PM}} \right) \\
= & \frac{1}{(PM)^2}  \sum_{r=0}^{PM-1}  \left\{S_X^2 H^*_pH_q\right\} \left(e^{2\pi i \frac{\phi-r} {PM}} \right) ,
\end{align*}
which is the matrix $\mathbf K_M\left(e^{2\pi i \phi} \right)$ defined in Theorem~\ref{thm:filterbanks_discrete}. Applying Theorem~\ref{thm:indirect_vector_case} with the eigenvalues of \eqref{eq:proof_multi_discrete_eigen} completes the proof.

\section{Proof of theorem \ref{thm:filterbanks_main} \label{sec:proof_of_filtersbank_main}: distortion-rate function in multi-branch sampling}
For $M\in\mathbb{N}$, define $X^{M}\left[\cdot\right]$
and $Z_p^M\left[\cdot\right]$, $p=1,...,P$ to be the processes obtained by uniformly sampling $X\left(\cdot\right)$ and $Z_p\left(\cdot\right)$ at frequency $f_sMP$, i.e. $X^M[n]=X\left(\frac{n}{f_sMP}\right)$
and $Z_p^M[n]=Z\left(\frac{n}{f_sMP}\right)$. We have
\[
S_{X^{M}}\left(e^{2\pi i\phi}\right)=MPf_{s}\sum_{k\in\mathbb{Z}}S_{X}\left(MPf_{s}\left(\phi-k\right)\right),
\]
and
\[
S_{Z_{p}^{M}}\left(e^{2\pi i\phi}\right)=MPf_{s}\sum_{k\in\mathbb{Z}}\left\{ S_{X}\left|H_{p}\right|^{2}\right\} \left(MPf_{s}\left(\phi-k\right)\right).
\]
In addition,  $X^{M}\left[\cdot\right]$ and
$Z_{p}^{M}\left[\cdot\right]$ are jointly stationary processes with
cross spectral densities 
\[
S_{Z_p^{M}Z_r^M}\left(e^{2\pi i\phi}\right)=MPf_s\sum_{m=0}^{MP-1}\sum_{k\in\mathbb{Z}}S_{Z_pZ_r}\left(f_s\left(\frac{\phi-m}{MP}-k\right)\right),
\]
for all $p,r=1,\ldots,P$, and
\[
S_{X^{M}Z_p^{M}}\left(e^{2\pi i\phi}\right)=MPf_s\sum_{m=0}^{MP-1}\sum_{k\in\mathbb{Z}}S_{XZ_{p}}\left(f_s\left(\frac{\phi-m}{MP}-k\right)\right).
\]
Since $Y_p\left[\cdot\right]$ is a factor $M$
down-sampled version of $Z_p^M[\cdot]$, the indirect
distortion-rate function of $X^M[\cdot]$ given $\mathbf{Y}[\cdot]=\left(Y_1[\cdot],\ldots,Y_P[\cdot]\right)$
was found in Theorem \ref{thm:filterbanks_discrete} to be 
\begin{equation}
\bar{R}_{X^M|\mathbf Y}\left(P,M,\theta\right)=\frac{1}{2}\sum_{p=1}^{P}\int_{-\frac{1}{2}}^{\frac{1}{2}}\log^{+}\left[\lambda_{p}\left(\mathbf J_M(e^{2\pi i \phi})\right)\theta^{-1}\right]d\phi,\label{eq:Rate_M-1}
\end{equation}
\begin{align}
D_{X^M|\mathbf Y}\left(P,M,\theta\right) & =  \sigma_{X^M}^{2}-\sum_{p=1}^{P}\int_{-\frac{1}{2}}^{\frac{1}{2}}\left[\lambda_{p}\left(\mathbf J_M(e^{2\pi i \phi})\right)-\theta\right]^+d\phi,\label{eq:Dist_M-1}
\end{align}
where \begin{align*}
\mathbf J_M\left(e^{2\pi i\phi}\right) & =  \mathbf S_{\mathbf Y}^{-\frac{1}{2}*}\mathbf{K}_{M}\mathbf S_{\mathbf Y}^{-\frac{1}{2}}\left(e^{2\pi i\phi}\right),
\end{align*}
${\mathbf S}_{\mathbf Y}(e^{2\pi i \phi})$  is the spectral density matrix of the process $\mathbf Y[\cdot]$ with $(p,r)^{\textrm{th}}$ entry 
\begin{align}
\left({\mathbf S}_{\mathbf Y}\left(e^{2\pi i\phi}\right)\right)_{p,r} 
& = \frac{1}{MP}\sum_{m=0}^{MP-1}S_{Z_p^M Z_r^M}\left(e^{2\pi i\frac{\phi-m}{MP}}\right)\nonumber \\
 & =f_{s}\sum_{m=0}^{MP-1}\sum_{k\in\mathbb{Z}}S_{Z_p Z_r}\left(f_{s}\left(\phi-m-MPk\right)\right)\nonumber \\
 & =f_{s}\sum_{n\in\mathbb{Z}}S_{Z_p Z_r}\left(f_s\left(\phi-n\right)\right)\nonumber \nonumber \\
 & = f_s \sum_{n\in \mathbb Z} \left\{S_{X+\eta}H_p H_r^* \right\}\left(f_s\phi-f_sn \right) \nonumber \\
 & =\left(\tilde{\mathbf S}_{\mathbf Y}\right)_{p,r}\left(\phi f_{s}\right)f_s, \label{eq:filter_bank_Sy}
\end{align}
and
\begin{align}
\left(\mathbf{K}_{M}\right)_{p,r} & =\frac{1}{(MP)^2}\sum_{m=0}^{MP-1}\left\{S_{X^{M}}^{2}H_p H_r\right\}\left(e^{2\pi i\frac{\phi-m}{MP}}\right) \nonumber \\
 & =\frac{1}{(MP)^2}\sum_{m=0}^{MP-1}\left\{ S_{X^{M}Z_p^M}S_{X^MZ_r^M}^{*}\right\}\left(e^{2\pi i\frac{\phi-m}{MP}}\right)\nonumber \\
  &=f_s^2\sum_{m=0}^{MP-1}\left[\sum_{k\in\mathbb{Z}}S_{XZ_p}\left(f_{s}\left(\phi-m-kMP\right)\right) \right.  \label{eq:filter_bank_K_M} \\
  &\quad  \quad\quad\quad \left . \times \sum_{l\in\mathbb{Z}}S_{XZ_r}^{*}\left(f_{s}\left(\phi-m-lMP\right)\right)\right]. \nonumber
\end{align}

The idea now is that under the assumption of Riemann integrability,
the distortion between almost any sample path of $X\left(\cdot\right)$ and any reasonable reconstruction of it from $X^{M}\left[\cdot\right]$ will converge to zero as $M\rightarrow\infty$. It follows that the distortion in reconstructing $X^M\left[\cdot\right]$ form $\hat{\mathbf Y}\left[\cdot\right]$
must also converges to the distortion in reconstructing $X\left(\cdot\right)$ from $\hat{\mathbf Y}\left[\cdot\right]$, and the indirect distortion-rate function of $X\left(\cdot\right)$ given $\mathbf Y\left[\cdot\right]$ is obtained by this limit. That is, we are looking to evaluate \eqref{eq:Rate_M-1} and \eqref{eq:Dist_M-1} in the limit $M\rightarrow \infty$. \\

First note that 
\begin{align*}
\sigma_{X_{M}}^{2} & =\int_{-\frac{1}{2}}^{\frac{1}{2}}S_{X^{M}}\left(e^{2\pi i\phi}\right)d\phi\\
 & =\int_{-\frac{1}{2}}^{\frac{1}{2}}MPf_{s}\sum_{k\in\mathbb{Z}}S_{X}\left(f_{s}MP\left(\phi-k\right)\right)d\phi\\
 & =\int_{-\infty}^{\infty}S_{X}(f)df=\sigma_{X}^{2}.
\end{align*}
In addition, by a change of the integration variable from $f$ to $\phi=f/f_s$, we can write \eqref{eq:multi_main} as
\begin{subequations}
\label{eq:multi_main_proof}
\begin{align}
\label{eq:multi_rate_proof}
R\left(P,f_s,\theta\right)=\frac{f_s}{2}\sum_{p=1}^{P}\int_{-\frac{1}{2}}^{\frac{1}{2}}\log^{+}\left[\lambda_p\left(\bar{\mathbf J}(e^{2\pi i\phi})\right)/\theta\right]d\phi
\end{align}
\begin{align}
\label{eq:multi_dist_proof}
D\left(P,f_s,\theta\right)  =\sigma_{X}^{2}-\sum_{p=1}^{P}\int_{-\frac{1}{2}}^{\frac{1}{2}}\left[\lambda_{p}\left(\bar{\mathbf J}(e^{2\pi i\phi})\right)-\theta\right]^{+}d\phi, 
\end{align}
\end{subequations}
where in $\phi \in \left(-\frac{1}{2},\frac{1}{2}\right)$, the matrix $\bar{\mathbf J}(e^{2\pi i\phi})$ is given by
\[
\bar{\mathbf J}(e^{2\pi i\phi})={\mathbf S}_{\mathbf Y}^{-\frac{1}{2}*}(e^{2\pi i\phi})\bar{\mathbf{K}}(e^{2\pi i\phi}){\mathbf S}_{\mathbf Y}^{-\frac{1}{2}}(e^{2\pi i\phi}),
\]
and $\bar{\mathbf K}(e^{2\pi i\phi})=f_s^2{\mathbf K}(f_s\phi)$. It follows that in order to complete the proof, it is enough to show that the eigenvalues of ${\mathbf J}_M (e^{2\pi i\phi})$ seen as $L_1\left(-\frac{1}{2},\frac{1}{2} \right)$ functions in $\phi$ converge to the eigenvalues of $\bar{\mathbf J} (e^{2\pi i\phi})$. Since
\begin{align*}
\|\mathbf{S_{\mathbf Y}}^{-\frac{1}{2}*}\mathbf{K}_{M}\mathbf{S_{\mathbf Y}}^{-\frac{1}{2}}&-\mathbf{S_{\mathbf Y}}^{-\frac{1}{2}*}\bar{\mathbf{K}}\mathbf{S_{\mathbf Y}}^{-\frac{1}{2}}\|_2 \leq \|\mathbf{S_{\mathbf Y}}\|^{-1}_2 \|f_s\mathbf{K}_{M}-\bar{\mathbf{K}}\|_2,
\end{align*}
it is enough to prove convergence in $L_1\left(-\frac{1}{2},\frac{1}{2} \right)$ for each entry, i.e. that
\begin{equation} \label{eq:filters_bank_to_show}
\lim_{M\rightarrow \infty} \int _{-\frac{1}{2}}^\frac{1}{2} \frac{  \left|\left( \mathbf K_M\right)_{p,r}\left(e^{2\pi i \phi} \right) -\left( \bar{\mathbf K}\right)_{p,r}\left(e^{2\pi i \phi} \right) \right|} {\| \mathbf{S_{\mathbf Y}} \left(e^{2\pi i \phi} \right) \|_2} d\phi =0
\end{equation}
for all  $p,r=1,\ldots,P$. Since
\begin{align*}
\left( \bar{\mathbf K}\right)_{p,r}\left(e^{2\pi i \phi} \right) & = f_s^2 \left(\mathbf K\right)_{p,r}(f_s \phi) \\
& = f_s^2 \sum_{k\in \mathbb Z} \left\{S_X^2H_iH_j^* \right\}\left(f_s(\phi-k)\right),
\end{align*}
\eqref{eq:filters_bank_to_show} follows by applying  Lemma~\ref{lem:proof_main} to \eqref{eq:filter_bank_K_M} with 
\[
f_1(\phi) = \frac{S_{XZ_p}\left(f_s\phi\right)}{\sqrt{\|\mathbf S_{\mathbf Y}\left(e^{2\pi i \phi} \right)\|_2}},
\]
\[
f_2(\phi) = \frac{S_{XZ_r}^*\left(f_s\phi\right)}{\sqrt{\|\mathbf S_{\mathbf Y}\left(e^{2\pi i \phi} \right)\|_2}}.
\]

\addtolength{\textheight}{-0cm}   

\section*{ACKNOWLEDGMENT}
The authors would like to thank to I. E. Aguerri and D. G\"{u}nd\"{u}z for valuable discussions regarding the problem formulation, to S. Rini and Y. Chen for helpful remarks and discussions and to the anonymous reviewers for extensive and constructive comments that greatly improved the paper.


\bibliographystyle{IEEEtran}
\bibliography{IEEEfull,/Users/Alon1/LaTex/bibtex/sampling}

\begin{IEEEbiographynophoto} 
{Alon Kipnis} (S'14) received the B.Sc. degree in mathematics (summa cum laude) and the B.Sc. degree in electrical engineering (summa cum laude) in 2010, and the M.Sc in mathematics in 2012, all from Ben-Gurion University of the Negev. He is currently a Ph.D candidate in the department of electrical engineering at Stanford University. His research interests include information theory, data compression and stochastic analysis.
\end{IEEEbiographynophoto}

\begin{IEEEbiographynophoto}
{Andrea J. Goldsmith} (S’90–M’93–SM’99–F’05) 
Andrea Goldsmith is the Stephen Harris professor in the School of Engineering and a professor of Electrical Engineering at Stanford University. She was previously on the faculty of Electrical Engineering at Caltech. Her research interests are in information theory and communication theory, and their application to wireless communications and related fields. She co-founded and served as Chief Scientist of Wildfire.Exchange, and previously co-founded and served as CTO of Quantenna Communications, Inc. She has also held industry positions at Maxim Technologies, Memorylink Corporation, and AT\&T Bell Laboratories. Dr. Goldsmith is a Fellow of the IEEE and of Stanford, and has received several awards for her work, including the IEEE ComSoc Edwin H. Armstrong Achievement Award as well as Technical Achievement Awards in Communications Theory and in Wireless Communications,  the National Academy of Engineering Gilbreth Lecture Award, the IEEE ComSoc and Information Theory Society Joint Paper Award, the IEEE ComSoc Best Tutorial Paper Award, the Alfred P. Sloan Fellowship, the WICE Technical Achievement Award, and the Silicon Valley/San Jose Business Journal’s Women of Influence Award. She is author of the book ``Wireless Communications'' and co-author of the books ``MIMO Wireless Communications'' and ``Principles of Cognitive Radio", all published by Cambridge University Press, as well as an inventor on 28 patents. She received the B.S., M.S. and Ph.D. degrees in Electrical Engineering from U.C. Berkeley. \par
Dr. Goldsmith has served on the Steering Committee for the IEEE Transactions on Wireless Communications and as editor for the IEEE Transactions on Information Theory, the Journal on Foundations and Trends in Communications and Information Theory and in Networks, the IEEE Transactions on Communications, and the IEEE Wireless Communications Magazine. She participates actively in committees and conference organization for the IEEE Information Theory and Communications Societies and has served on the Board of Governors for both societies. She has also been a Distinguished Lecturer for both societies, served as President of the IEEE Information Theory Society in 2009, founded and chaired the student committee of the IEEE Information Theory society, and chaired the Emerging Technology Committee of the IEEE Communications Society. At Stanford she received the inaugural University Postdoc Mentoring Award, served as Chair of Stanford’s Faculty Senate in 2009 and currently serves on its Faculty Senate, Budget Group, and Task Force on Women and Leadership.
\end{IEEEbiographynophoto}

\begin{IEEEbiographynophoto}
{Yonina C. Eldar} (S’98–M’02–SM’07-F'12) received the B.Sc. degree in Physics in 1995 and the B.Sc. degree in Electrical Engineering in 1996 both from Tel-Aviv University (TAU), Tel-Aviv, Israel, and the Ph.D. degree in Electrical Engineering and Computer Science in 2002 from the Massachusetts Institute of Technology (MIT), Cambridge.\par
From January 2002 to July 2002 she was a Postdoctoral Fellow at the Digital Signal Processing Group at MIT. She is currently a Professor in the Department of Electrical Engineering at the Technion - Israel Institute of Technology, Haifa, Israel, where she holds the Edwards Chair in Engineering. 
 She is also a Research Affiliate with the Research Laboratory of Electronics at MIT and was a Visiting Professor at Stanford University, Stanford, CA.
 Her research interests are in the broad areas of
statistical signal processing, sampling theory and compressed sensing, optimization methods, and their applications to biology and optics. \par
Dr. Eldar has received numerous awards for excellence in research and teaching, including the IEEE Signal Processing Society Technical Achievement Award (2013), the IEEE/AESS Fred Nathanson Memorial Radar Award (2014), and the IEEE Kiyo Tomiyasu Award (2016). She was a Horev Fellow of the Leaders in Science and Technology program at the Technion and an Alon Fellow. She received the Michael Bruno Memorial Award from the Rothschild Foundation, the Weizmann Prize for Exact Sciences, the Wolf Foundation Krill Prize for Excellence in Scientific Research, the Henry Taub Prize for Excellence in Research (twice), the Hershel Rich Innovation Award (three times), the Award for Women with Distinguished Contributions, the Andre and Bella Meyer Lectureship, the Career Development Chair at the Technion, the Muriel \& David Jacknow Award for Excellence in Teaching, and the Technion’s Award for Excellence in Teaching (2 times). She received several best paper awards and best demo awards together with her research students and colleagues including the SIAM outstanding Paper Prize and the IET Circuits, Devices and Systems Premium Award, and was selected as one of the 50 most influential women in Israel. \par
 She is a member of the Young Israel Academy of Science and Humanities and the Israel Committee for Higher Education, and an IEEE Fellow. She is the Editor in Chief of Foundations and Trends in Signal Processing, a member of the IEEE Sensor Array and Multichannel Technical Committee and serves on several other IEEE committees. In the past, she was a Signal Processing Society Distinguished Lecturer, member of the IEEE Signal Processing Theory and Methods and Bio Imaging Signal Processing technical committees, and served as an associate editor for the IEEE Transactions On Signal Processing, the EURASIP Journal of Signal Processing, the SIAM Journal on Matrix Analysis and Applications, and the SIAM Journal on Imaging Sciences.  She was Co-Chair and Technical Co-Chair of several international conferences and workshops. \par
She is author of the book ``Sampling Theory: Beyond Bandlimited Systems" and co-author of the books ``Compressed Sensing" and ``Convex Optimization Methods in Signal Processing and Communications", all published by Cambridge University Press.
\end{IEEEbiographynophoto}

\begin{IEEEbiographynophoto}
{Tsachy Weissman} (S’99-M’02-SM’07-F’13) graduated summa cum laude with a B.Sc. in electrical engineering from the Technion in 1997, and earned his Ph.D. at the same place in 2001. He then worked at Hewlett Packard
Laboratories with the information theory group until 2003, when he joined Stanford University, where he is currently Professor of Electrical Engineering and incumbent of the STMicroelectronics chair in the School of Engineering.
He has spent leaves at the Technion, and at ETH Zurich.
Tsachy's research is focused on information theory, compression, communication, statistical signal processing, the interplay between them, and their applications. He is recipient of several best paper awards, and prizes for excellence in research and teaching. He served on the editorial board of the
IEEE TRANSACTIONS ON INFORMATION THEORY from Sept. 2010 to Aug. 2013, and currently serves on the editorial board of Foundations and Trends in Communications and Information Theory. He is Founding Director of the Stanford Compression Forum. 
\end{IEEEbiographynophoto}

\end{document}